\documentclass[letter,11pt]{article}

  \newcommand{\fullversion}[1]{#1}
\newcommand{\submversion}[1]{}
 
\usepackage{bbold} 
 
\usepackage{amsmath,amssymb}  
\usepackage{xspace}
\fullversion{
\usepackage{fullpage}
}
\usepackage{verbatim}
\usepackage{todonotes}
 \usepackage[colorlinks]{hyperref}
  \hypersetup{linkcolor=blue,filecolor=blue,citecolor=blue,urlcolor=blue}
  \usepackage{babel}

\usepackage{textcomp}

\fullversion{
\usepackage[T1]{fontenc}
\usepackage{textcomp}
\usepackage[varqu,varl]{inconsolata}
\linespread{1.05}
\usepackage[T1]{fontenc}
}

\fullversion{
\usepackage{amsthm}
\newtheorem{theorem}{Theorem}
\newtheorem{proposition}[theorem]{Proposition}
\newtheorem{definition}[theorem]{Definition}
\newtheorem{lemma}[theorem]{Lemma}
\newtheorem{claim}[theorem]{Claim}
\newtheorem{remark}[theorem]{Remark}
\newtheorem{corollary}[theorem]{Corollary}

}

\newcommand{\bfs}{\mathbf{s}}
\newcommand{\bfA}{\mathbf{A}}
\newcommand{\bfe}{\mathbf{e}}

\newcommand{\bfu}{\mathbf{u}}
\newcommand{\simulator}{\mathsf{Sim}}

\newcommand{\cktclass}{\mathcal{C}}
\newcommand{\lobf}{\mathsf{LO.Obf}}
\newcommand{\leval}{\mathsf{LO.Eval}}
\newcommand{\lockC}{\mathbf{C}}

\newcommand{\shO}{\mathsf{shO}}
\newcommand{\qfhe}{\mathsf{QFHE}}

\newcommand{\process}{\realevent}

\newcommand{\setup}{\mathsf{Setup}}
\newcommand{\enc}{\mathsf{Enc}}\newcommand{\dec}{\mathsf{Dec}}
\newcommand{\sk}{\mathsf{sk}}
\newcommand{\pk}{\mathsf{pk}}
\newcommand{\ct}{\mathsf{ct}}

\newcommand{\ft}{\mathsf{FT}}

\newcommand{\return}{\mathsf{Check}}

\newcommand{\tr}{\text{Tr}}
\newcommand{\cE}{\mathcal{E}}
\newcommand{\nonnegl}{\text{non-negl}}
\newcommand{\sfY}{\mathbf{Y}}
\newcommand{\sfA}{\mathbf{A}}

\newcommand{\cH}{\mathcal{H}}
\newcommand{\cC}{\mathcal{C}}
\newcommand{\cR}{\mathcal{R}}
\newcommand{\cF}{\mathcal{F}}
\newcommand{\Id}{\mathsf{Id}}
\newcommand{\bfE}{\mathbf{E}}

\submversion{
\usepackage{cite}
}

\newcommand{\fclass}{\mathcal{F}}

\newcommand{\secparam}{\lambda}

\newcommand{\run}{\mathsf{Run}}

\newcommand{\hyb}{\bullet\mathsf{Hyb}}
\newcommand{\hybrid}{\mathsf{Hybrid}}

\newcommand{\cO}{\mathcal{O}}
\newcommand{\cD}{\mathcal{D}}
\newcommand{\cA}{\mathcal{A}}
\newcommand{\cB}{\mathcal{B}}
\newcommand{\cM}{\mathcal{M}}

\newcommand{\negl}{\text{negl}}

\newcommand{\msg}{\mathsf{msg}}
\newcommand{\poly}{\mathrm{poly}}

\newcommand{\distr}{\mathcal{D}}

\newcommand{\bra}[1]{\langle #1|}
\newcommand{\ket}[1]{|#1\rangle}
\newcommand{\braket}[2]{\langle #1|#2\rangle}

\newcommand{\trD}[2]{\left\|#1-#2\right\|_{\text{tr}}}

\newcommand{\adversary}{\mathcal{A}}
\newcommand{\prob}{\mathsf{Pr}}

\newcommand{\lessor}{\mathsf{Lessor}}

\newcommand{\crs}{\mathsf{crs}}

\newcommand{\Zq}{\mathbb{Z}_q}
\newcommand{\sfX}{\mathbf{X}}
\newcommand{\sampler}{\mathsf{Sampler}}
\newcommand{\obf}{\mathsf{Obf}}
\newcommand{\eval}{\mathsf{Eval}}
\newcommand{\obfuscator}{\mathcal{O}}

\newcommand{\prover}{\mathcal{P}}

\newcommand{\newcktclass}{\mathcal{G}}
\newcommand{\search}{\mathcal{S}}

\newcommand{\rel}{\mathcal{R}}

\newcommand{\ssnizk}{\mathsf{qseNIZK}}
\newcommand{\crsgen}{\mathsf{CRSGen}}
\newcommand{\iho}{\mathsf{qIHO}}
\newcommand{\fksetup}{\mathsf{FkGen}}
\newcommand{\td}{\mathsf{td}}
\newcommand{\simr}{\mathsf{Sim}}
\newcommand{\st}{\mathsf{st}}
\newcommand{\reduction}{\mathcal{B}}

\newcommand{\po}{\mathsf{PO}}
\newcommand{\qiho}{\mathsf{qIHO}}
\newcommand{\cnc}{\mathsf{cnc}}

\newcommand{\qpke}{\mathsf{qPKE}}
\newcommand{\nizk}{\mathsf{NIZK}}
\newcommand{\lang}{\mathcal{L}}

\newcommand{\gen}{\mathsf{Gen}}
\newcommand{\verify}{\mathcal{V}}
\newcommand{\pp}{\mathsf{pp}}
\newcommand{\Reg}{\mathsf{R}}
\newcommand{\distrc}{\mathcal{D}_{\cktclass}}

\newcommand{\aux}{\mathsf{aux}}

\newcommand{\bfH}{\mathbf{H}}
\newcommand{\hill}{\mathsf{HILL}}

\newcommand{\unpred}{\mathsf{unpred}}
\newcommand{\pseud}{\mathsf{pseud}}
\newcommand{\smpo}{\mathsf{po}}
\newcommand{\hybprocess}{\mathsf{HybProcess}}
\newcommand{\sz}{s}
\newcommand{\realevent}{\mathsf{Process}}

\newcommand{\prab}[1]{{\color{blue} Prab: #1}}
\newcommand{\rolo}[1]{{\color{red} Rolo: #1}}

\fullversion{
\title{Secure Software Leasing}
\author{Prabhanjan Ananth \vspace{0.3cm}\\ UCSB\\ \texttt{prabhanjan@cs.ucsb.edu}   \and Rolando L. La Placa \vspace{0.3cm}\\ MIT\\  \texttt{rlaplaca@mit.edu}}
}

\submversion{
\title{Secure Software Leasing}
\author{}
\institute{}
}

\date{}

\begin{document}

\maketitle
\begin{abstract}
Formulating cryptographic definitions to protect against software piracy is an important research direction that has not received much attention. Since natural definitions using classical cryptography are impossible to achieve (as classical programs can always be copied), this directs us towards using techniques from quantum computing. The seminal work of Aaronson [CCC'09] introduced the notion of quantum copy-protection precisely to address the problem of software anti-piracy. However, despite being one of the most important problems in quantum cryptography, there are no provably secure solutions of quantum copy-protection known for {\em any} class of functions. \\
\vspace{-0.01cm}
\ \ \ \ We formulate an alternative definition for tackling software piracy, called secure software leasing (SSL). While weaker than quantum copy-protection, SSL is still meaningful and has interesting applications in software anti-piracy. \\
\vspace{-0.01cm}
\ \ \ \ \ \ \ We present a construction of SSL for a subclass of evasive circuits (that includes natural implementations of point functions, conjunctions with wild cards, and affine testers) based on concrete cryptographic assumptions. Our construction is the first provably secure solution, based on concrete cryptographic assumptions, for software anti-piracy. To complement our positive result, we show, based on cryptographic assumptions, that there is a class of quantum unlearnable functions for which SSL does not exist. In particular, our
impossibility result also rules out quantum copy-protection [Aaronson CCC'09]
for an arbitrary class of quantum unlearnable functions; resolving an important open problem on the possibility of constructing copy-protection for arbitrary quantum unlearnable circuits.
\end{abstract}
\section{Introduction}

Almost all proprietary software requires a legal document, called software license, that governs the use against illegal distribution of software, also referred to as pirating. The main security requirement from such a license is that any malicious user no longer has access to the functionality of the software after the lease associated with the software license expires. While ad hoc solutions existed in the real world, for a long time,  no theoretical treatment of this problem was known.\\

\noindent This was until Aaronson, who in his seminal work~\cite{Aar09} introduced and formalized the notion of quantum software copy-protection, a quantum cryptographic primitive that uses quantum no-cloning techniques to prevent pirating of software by modeling software as boolean functions. Roughly speaking, quantum copy-protection says\footnote{More generally, Aaronson considers the setting where the adversary gets multiple copies computing $f$ and not just one.} that given a quantum state computing a function $f$, the adversary cannot produce two quantum states (possibly entangled) such that each of the states individually  computes $f$. This prevents a pirate from being able to create a new software from his own copy and re-distribute it; of course it can circulate its own copy to others but it will lose access to its own copy. 

\paragraph{Need for Alternate Notions.}  While quantum  copy-protection does provide a solution for software piracy, constructing quantum copy-protection has been notoriously difficult. Despite being introduced more than a decade ago, not much is known on the existence of quantum copy-protection. There are no known provably secure constructions of quantum copy-protection for {\em any} class of circuits. All the existing constructions of quantum copy-protection are either proven in an oracle model \cite{Aar09,ALZ20} or are heuristic\footnote{That is, there is no known reduction to concrete cryptographic assumptions.} candidates for very simple functions such as point functions~\cite{Aar09}. In a recent blog post, Aaronson~\cite{aaronson_2020} even mentioned constructing quantum copy-protection from cryptographic assumptions as one of the five big questions he wishes to solve. 
\par This not only prompted us to explore the possibility of copy-protection but also look for alternate notions to protect against software piracy. Specifically, we look for application scenarios where the full power of quantum copy-protection is not needed and it suffices to settle for weaker notions. Let us consider one such example. 

\paragraph{Example: Anti-Piracy Solutions for Microsoft Office.} Microsoft Office is one of the most popular software tools used worldwide. Since Microsoft makes a sizeable portion of their revenue from this tool~\cite{revenue}, it is natural to protect Microsoft Office from falling prey to software piracy. A desirable requirement is that pirated copies cannot be sold to other users such that these copies can run successfully on other Microsoft Windows systems. Importantly, it does not even matter if the pirated copies can be created as long as they cannot be executed on other Windows systems; this is because, only the pirated copies that run on Windows systems are the ones that bite into the revenue of Microsoft. Indeed, there are open source versions of Office publicly available but our aim is to prevent these open source versions from being sold off as authentic versions of Microsoft Office software. 
\par This suggests that instead of quantum copy-protection -- which prevents the adversary from creating {\em any} pirated copy of the copy-protected software -- we can consider a weaker variant that only prevents the adversary from being able to create {\em authenticated} pirated copies (for instance, that runs only on specific operating systems). To capture this, we present a new definition called secure software leasing. 



\paragraph{Our Work: Secure Software Leasing (SSL).} Roughly speaking, a secure leasing scheme allows for an authority (the lessor\footnote{The person who leases the software to another.}) to lease a classical circuit $C$ to a user (the lessee\footnote{The person to whom the software is being leased to.}) by providing a corresponding quantum state $\rho_C$. The user can execute $\rho_C$ to compute $C$ on any input it desires. Leases can expired, requiring $\rho_C$ to be returned at a later point in time, specified by the lease agreement. After it returns the state, we require the security property that the lessee can no longer compute $C$.    

In more detail, a secure software leasing scheme (SSL) for a family of circuits $\cktclass$ is a collection, $(\gen,\lessor,\run,\return)$, of quantum polynomial-time algorithms (QPT) satisfying the following conditions. $\gen(1^\secparam)$, on input a security parameter $\secparam$, outputs a secret key $\sk$ that will be used by a lessor to validate the states being returned after the expiration of the lease. For any circuit $C:\{0,1\}^n \rightarrow \{0,1\}^m$ in $\cktclass$, $\lessor(\sk,C)$ outputs a quantum state $\rho_C$, where $\rho_C$ allows $\run$ to evaluate $C$.  Specifically, for any $x\in\{0,1\}^n$, we want that $\run(\rho_C,x)=C(x)$; this algorithm is executed by the lessee. Finally, $\return(\sk,\rho_C)$ checks if $\rho_C$ is a valid leased state. Any state produced by the lessor is a valid state and will pass the verification check. 

A SSL scheme can have two different security guarantees depending on whether the leased state is supposed to be returned or not.
\begin{itemize}
\setlength\itemsep{1em}

\item \textit{Infinite-Term Lessor Security}: In this setting, there is no time duration associated with the leased state and hence, the user can keep this leased state forever\footnote{Although the lessor will technically be the owner of the leased state.}. Informally, we require the guarantee that the lessee, using the leased state, cannot produce two {\em authenticated} copies of the leased state. Formally speaking, any (malicious) QPT user $\cA$ holding a leased state $\cA(\rho_C)$ (produced using classical circuit $C$) cannot output a (possibly entangled) bipartite state $\sigma^*$ such that both $\sigma_1^*=\tr_2[\sigma^*]$ and $\sigma_2^*=\tr_1[\sigma^*]$ can be used to compute $C$ with $\run$.

    \item \textit{Finite-Term Lessor Security}: On the other hand, we could also consider a weaker setting where the leased state is associated with a fixed term. In this setting, the lessee is obligated to return back the leased state after the term expires. We require the property that after the lessee returns back the state, it can no longer produce another {\em authenticated} state having the same functionality as the leased state.  
    \par Formally speaking, we require that any (malicious) QPT user $\cA$ holding a leased state $\rho_C$ (produced using $C$) cannot output a (possibly entangled) bipartite states $\sigma^*$  such that $\sigma^*_1:=\tr_2[\sigma^*]$\footnote{This denotes tracing out the second register.} passes the lessor's verification ($\return(\sk,\sigma^*_1)=1$) and such that the the resulting state, after the first register has been verified by the lessor, on the second register, $\sigma_2^*$, can also be used to evaluate $C$ with the $\run$ algorithm, $\run(\sigma^*_2,x)=C(x)$. 
\end{itemize}
\noindent A SSL scheme satisfying infinite-term security would potentially be useful to tackle the problem of developing anti-piracy solutions for Microsoft Office. However, there are scenarios where finite-term security suffices. We mention two examples below. 

\paragraph{Trial Versions.} Before releasing the full version of a program $C$, a software vendor might want to allow a selected group of people\footnote{For instance, they could be engineers assigned to test whether the beta version contains bugs.} to run a beta version of it, $C_\beta$, in order to test it and get user feedback. Naturally, the vendor would not want the beta versions to be pirated and distributed more widely. Again, they can lease the beta version $C_{\beta}$, expecting the users to return it back when the beta test is over. At this point, they would know if a user did not return their beta version and they can penalize such a user according to their lease agreement.


\paragraph{Subscription Models.} Another example where finite-term SSL would be useful is for companies that use a subscription model for their revenue.  For example, Microsoft has a large library of video games for their console, the Xbox, which anyone can have access to for a monthly subscription fee. A malicious user could subscribe in order to have access to the collection of games, then make copies of the games intending to keep them after cancelling the subscription. The same user will not be able to make another copy of a game that also runs on Xbox.

\subsection{Our Results}
We present a construction of SSL for a restricted class of unlearnable circuits; in particular, our construction is defined for a subclass of evasive circuits. This demonstrates the first provably secure construction for the problem of software anti-piracy in the standard model (i.e., without using any oracles). 
\par Our construction does not address the possibility of constructing SSL for an arbitrary class of unlearnable circuits. Indeed, given the long history of unclonable quantum cryptographic primitives (see Section ~\ref{ssec:RW}) along with the recent boom in quantum cryptographic techniques~\cite{Zha19,mahadev2018classical,mahadev2018verification,broadbent2019uncloneable,broadbent2019quantum,coladangelo2019non,broadbent2019zero,coladangelo2019smart,ben2016quantum,AGKZ20}, one might hope that existing techniques could lead us to achieve a general result. We show, rather surprisingly, assuming cryptographic assumptions, there exists a class of unlearnable circuits such that no SSL exists for this class. This also rules out the existence of quantum copy-protection for arbitrary unlearnable functions\footnote{Both the notions (quantum copy-protection and secure software leasing) are only meaningful for unlearnable functions: if a function is learnable, then one could learn the function from the quantum state and create another authenticated quantum state computing the same function.}; {\em thus resolving an important open problem in quantum cryptography}.
\par We explain our results in more detail. We first start with the negative result before moving on to the positive result. 


\subsubsection{Impossibility Result.}
\noindent To demonstrate our impossibility result, we identify a class of classical circuits $\cktclass$ that we call a {\em de-quantumizable} circuit class. This class has the nice property that given \textit{any} efficient quantum implementation of $C \in \cktclass$, we can efficiently `de-quantumize' it to obtain a classical circuit $C' \in \cktclass$ that has the same functionality as $C$. If $\cktclass$ is learnable then, from the definition of learnability, there could be a QPT algorithm that finds $C'$. To make the notion interesting and non-trivial, we add the additional requirement that this class of circuits is quantum unlearnable. A circuit class $\cktclass$ is quantum unlearnable if given black-box access to $C \in \cktclass$, any QPT algorithm cannot find a quantum implementation of $C$. 
\par We show the existence of a de-quantumizable circuit class from cryptographic assumptions.

\begin{proposition}[Informal]
Assuming the quantum hardness of learning with errors (QLWE), and asssuming the existence of quantum fully homomorphic encryption\footnote{We need additional properties from the quantum fully homomorphic encryption scheme but these properties are natural and satisfied by existing schemes~\cite{mahadev2018classical,Bra18}. Please refer to Section~\ref{sec:prelims:qfhe} for a precise description of these properties.} (QFHE), there exists a de-quantumizable class of circuits.  
\end{proposition}

\noindent We show how non-black-box techniques introduced in seemingly different contexts -- proving impossibility of obfuscation~\cite{BGIRSVY01,BP16,alagic2016quantum} and constructing zero-knowledge protocols~\cite{BKP19,BS19,AL19} -- are relevant to proving the above proposition. We give an overview, followed by a formal construction, in Section~\ref{sec:imp}.

We then show that for certain de-quantumizable class of circuis, there does not exist a SSL scheme (with either finite or infinite-term security) for this class. Combining this with the above proposition, we have the following: 

\begin{theorem}[Informal]
Assuming the quantum hardness of learning with errors (QLWE), and asssuming the existence of quantum fully homomorphic encryption (QFHE), there exists a class of quantum unlearnable circuits $\cktclass$ such that there is no SSL for $\cktclass$. 
\end{theorem}

\noindent{\em On the Assumption of QFHE}: There are lattice-based constructions of QFHE proposed by~\cite{mahadev2018classical,Bra18} although we currently don't know how to base them solely on the assumption of LWE secure against QPT adversaries (QLWE). Brakerski~\cite{Bra18} shows that the security of QFHE can be based on QLWE and a circular security assumption.


\paragraph{Impossibility of Quantum Copy-Protection.} Since copy-protection implies SSL, we have the following result. 

\begin{corollary}[Informal] Assuming the quantum hardness of learning with errors (QLWE), and asssuming the existence of quantum fully homomorphic encryption (QFHE), there exists a class of quantum unlearnable circuits $\cktclass$ that cannot be quantumly copy-protected.
\end{corollary}


\subsubsection{Main Construction.} 

\noindent Our impossibility result does not rule out the possibility of constructing SSL schemes for specific circuit classes. For instance, it does not rule out the feasibility of SSL for \textit{evasive functions}; this is a class of functions with the property that given black-box access, an efficient algorithm cannot find an accepting input (an input on which the output of the function is 1).
\par We identify a subclass of evasive circuits for which we can construct SSL.
infinite
\paragraph{Searchable Compute-and-Compare Circuits.} We consider the following circuit class $\cktclass$: every circuit in $ \cktclass$, associated with a circuit $C$ and a lock $\alpha$, takes as input $x$ and outputs 1 iff $C(x)=\alpha$. This circuit class has been studied in the cryptography literature in the context of constructing program obfuscation~\cite{WZ17,GKW17}. We require this circuit class to additionally satisfy a {\em searchability} condition: there is an efficient (classical) algorithm, denoted by $\search$, such that given any $C \in \cktclass$, $\search(C)$ outputs $x$ such that $C(x)=1$. 
\par There are natural and interesting sub-classes of compute-and-compare circuits: 
\begin{itemize}
    \item Point circuits $C(\alpha, \cdot)$: the circuit $C(\alpha,\cdot)$ is a point circuit if it takes as input $x$ and outputs $C(\alpha,x)=1$ iff $x=\alpha$. If we define the class of point circuits suitably, we can find $\alpha$ directly from the description of $C(\alpha,\cdot)$; for instance, $\alpha$ is the value assigned to the input wires of $C$. 
    \item Conjunctions with wild cards $C(S,\alpha, \cdot)$: the circuit $C(S,\alpha,\cdot)$ is a conjunction with wild card if it takes as input $x$ and outputs $C(S,\alpha,x)=1$ iff $y = \alpha$, where $y$ is such that $y_i=x_i$ for all $i \in S$ and $y_i=0$ for all $i \notin S$.
    Again, if we define this class of circuits suitably, we can find $S$ and $\alpha$ directly from the description of $C(S,\alpha,\cdot)$. 
\end{itemize}

\noindent {\em On Searchability}: We note that the searchability requirement in our result statement is natural and is implicit in the description of the existing constructions of copy-protection by Aaronson~\cite{Aar09}. Another point to note is that this notion is associated with circuit classes rather than a function family. \\

\noindent We prove the following result.  Our construction is in the common reference string (CRS) model. In this model, we assume that both the lessor and the lessee will have access to the CRS produced by a trusted setup. We note that our impossibility result also holds in the CRS model. 

\begin{theorem}[SSL for Searchable Compute-and-Compare Circuits; Informal]
\label{thm:SSL:cnc}
Assuming the existence of: (a) quantum-secure subspace obfuscators~\cite{Zha19} and, (b) learning with errors secure against sub-exponential quantum algorithms, there exists an infinite-term secure SSL scheme in the common reference string model for searchable compute-and-compare circuits. 
\end{theorem}

\par Notice that for applications in which the lessor is the creator of software, the lessor can dictate how the circuit class is defined and thus would choose an implementation of the circuit class that is searchable.\\ 

\noindent {\em On the Assumptions in Theorem~\ref{thm:SSL:cnc}.} A discussion about the primitives described in the above theorem statement is in order. A subspace obfuscator takes as input a subspace $A$ and outputs a circuit that tests membership of $A$ while hiding $A$ even against quantum adversaries. This was recently constructed by~\cite{Zha19} based on the quantum-security of indistinguishability obfuscation~\cite{GGHRSW13}. Moreover, recently, there has been exciting progress in constructing quantum-secure indistinguishability obfuscation schemes~\cite{GP20,WW20,BDGM20} from cryptographic assumptions that hold against quatum adversaries.

With regards to the assumption of learning with errors against sub-exponential quantum algorithms, we firstly note that classical sub-exponential security of learning with errors has been used in the construction of many cryptographic primitives and secondly, there are no known significant quantum speedups known to solving this problem.  
\par In the technical sections, we prove a more general theorem.

\begin{theorem}[SSL for General Evasive Circuits; Informal]
Let $\cktclass$ be a searchable class of circuits. Assuming the existence of: (a) quantum-secure input-hiding obfuscators~\cite{BBCKP14} for $\cktclass$, (b) quantum-secure subspace obfuscators~\cite{Zha19} and, (c) learning with errors secure against sub-exponential quantum algorithms, there exists an infinite-term secure SSL scheme in the setup model for $\cktclass$. 
\end{theorem}

\noindent An input-hiding obfuscator is a compiler that converts a circuit $C$ into another functionally equivalent circuit $\widetilde{C}$ such that given $\widetilde{C}$ it is computationally hard to find an accepting point. To achieve Theorem~\ref{thm:SSL:cnc}, we instantiate searchable input-hiding obfuscators for compute-and-compare circuits from quantum hardness of learning with errors. However, we can envision quantum-secure instantiations of input-hiding obfuscators for more general class of searchable evasive circuits; we leave this problem open. 
\par We admittedly use heavy cryptographic hammers to prove our result, but as will be clear in the overview given in the next section, each of these hammers will be necessary to solve the different technical challenges we face. 

\paragraph{Concurrent Work on qVBB.} Our impossibility result also rules out the existence of quantum VBB for classical circuits assuming quantum FHE and quantum learning of errors; this was stated as an open problem by Alagic and Fefferman~\cite{alagic2016quantum}. 
Concurrently,~\cite{ABDS20} also rule out quantum virtual black-box obfuscation under the assumption of quantum hardness of learning with errors; unlike our work they don't additionally assume the existence of quantum FHE.


\subsection{Overview of Construction of SSL}
\label{sec:overview:cons}
\noindent For this overview, we only focus on constructing a SSL sscheme satisfying finite-term lessor security. Our ideas can be easily adapted to the infinite-term lessor security. 
\par To construct a SSL scheme in the setup model $(\setup,\gen,\lessor,\run,\return)$ against arbitrary quantum poly-time (QPT)  pirates, we first focus on two weaker class of adversaries, namely, {\em duplicators} and {\em maulers}. Duplicators are adversaries who, given $\rho_C$ generated by the lessor for a circuit $C$ sampled from a distribution $\distrc$, produce $\rho_C^{\otimes 2}$; that is, all they do is replicate the state. Maulers, who given $\rho_C$, output $\rho_C \otimes \rho^*_C$, where $\rho_C^*$ is far from  $\rho_C$ in trace distance and $\rho_C$ is the copy returned by the mauler back to the lessor; that is the second copy it produces is a modified version of the original copy. 
\par While our construction is secure against arbitrary pirates, it will be helpful to first focus on these restricted type of adversaries. We propose two schemes: the first scheme is secure against QPT maulers and the second scheme against QPT duplicators. Once we discuss these schemes, we will then show how to combine the techniques from these two schemes to obtain a construction secure against arbitrary pirates. 

\paragraph{SSL against Maulers.} To protect SSL against a mauler, we attempt to construct a scheme using only classical cryptographic techniques. The reason why it could be possible to construct such a scheme is because maulers never produce a pirated copy $\rho_C^*$ that is the same as the original copy $\rho_C$. 
\par A natural attempt to construct a SSL scheme is to use virtual black-box obfuscation~\cite{BGIRSVY01} (VBB): this is a compiler that transforms  a circuit $C$ into another functionally equivalent circuit $\widetilde{C}$ such that $\widetilde{C}$ only leaks the input-output behavior of $C$ and nothing more. This is a powerful notion and implies almost all known cryptographic primitives. We generate the leased state $\rho_C$ to be the VBB obfuscation of $C$, namely $\widetilde{C}$. The hope is that a mauler will not  output another leased state $\rho_C^*$ that is different from $\widetilde{C}$. 
\par Unfortunately, this scheme is insecure. A mauler on input $\widetilde{C}$, obfuscates $\widetilde{C}$ once more to obtain $\widetilde{\widetilde{C}}$ and outputs this re-obfsuscated circuit. Moreover, note that the resulting re-obfuscated circuit still computes $C$. This suggests that program obfuscation is insufficient for our purpose. In hindsight, this should be unsurprising: VBB guarantees that given an obfuscated circuit, an efficient adversary should not learn anything about the implementation of the circuit, but this doesn't prevent the adversary from being able to re-produce modified copies of the obfuscated circuit. 
\par To rectify this issue, we devise the following strategy: 
\begin{itemize}
    \item Instead of VBB, we start with a different obfuscation scheme that has the following property: given an obfuscated circuit $\widetilde{C}$, where $C$ corresponds to an evasive function, it is computationally infeasible to determine an accepting input for $C$. 
    \ \\
    \item We then combine this with a special proof system that guarantees the property: suppose an adversary, upon receiving $\widetilde{C}$ and a proof, outputs a {\em different} but functionally equivalent obfuscated circuit $\widetilde{C}^*$ along with a new proof. Then we can extract an accepting input for $\widetilde{C}$ from the adversary's proof. But this would contradict the above bullet and hence, it follows that its computationally infeasible for the adversary to output a different circuit $\widetilde{C}^*$. 
\end{itemize}

\noindent To realize the above strategy, we need two separate cryptographic tools, that we define below. \\


\noindent {\em Input-Hiding Obfuscators~\cite{BBCKP14}}: We recall the notion of input-hiding obfuscators~\cite{BBCKP14}. An input-hiding obfuscator guarantees that given an obfuscated circuit $\widetilde{C}$, any efficient adversary cannot find an accepting input $x$, i.e., an input $x$ such that $\widetilde{C}(x)=1$. Of course this notion is only meaningful for an evasive class of functions: a function is evasive if given oracle access to this function, any efficient adversary cannot output an accepting point. The work of Barak et al.~\cite{BBCKP14} proposed candidates for input-hiding obfuscators. \\
\ \\
\noindent {\em Simulation-Extractable NIZKs~\cite{Sah99,DeS+01}}: Another primitive we consider is simulation-extractable non-interactive zero-knowledge~\cite{Sah99,DeS+01} (seNIZKs). A seNIZK system is a non-interactive protocol between a prover and a verifier with the prover trying to convince the verifier that a statement belongs to the NP language. By non-interactive we mean that the prover only sends one message to the verifier and the verifier is supposed to output the decision bit: accept or reject. Moreover, this primitive is defined in the common reference string model. In this model, there is a trusted setup that produces a common reference string and both the prover and the verifier have access to this common reference string.  
\par As in a traditional interactive protocol, we require a seNIZK to satisfy the completeness property. Another property we require is simulation-extractability.  Simulation-extractability, a property that implies both zero-knowledge and soundness, guarantees that if there exists an efficient adversary $\cA$ who upon receiving a {\em simulated} proof\footnote{A simulated proof is one that is generated by an efficient algorithm, called a simulator, who has access to some private coins that was used to generate the common reference string. Moreover, a simulated proof is indistinguishable from an honestly generated proof. A simulator has the capability to generate simulated proofs for YES instances even without knowing the corresponding witness for these instances.} for an instance $x$, produces an accepting proof for a different instance $x'$, i.e., $x' \neq x$, then there also exists an adversary $\reduction$ that given the same simulated proof produces an accepting proof for $x'$ along with simultaneously producing a valid witness for $x'$. \\

\noindent {\em Combining Simulation-Extractable NIZKs and Input-Hiding Obfuscators}: We now combine the two tools we introduced above to obtain a SSL scheme secure against maulers. Our SSL scheme will be associated with searchable circuits; given a description of a searchable circuit $C$, an input $x$ can be determined efficiently such that $C(x)=1$.
\par To lease a circuit $C$, do the following: 
\begin{itemize}
 \setlength\itemsep{0.5em}
    \item Compute an input-hiding obfuscation of $C$, denoted by $\widetilde{C}$, 
    \item Produce a seNIZK proof $\pi$ that proves knowledge of an input $x$ such that $C(x)=1$. Note that we can find this input using the searchability property. 
\end{itemize}
\noindent Output $( \widetilde{C},\pi)$ as the leased circuit. To evaluate on any input $x$, we first check if $\pi$ is a valid proof and if so, we compute $\widetilde{C}$ on $x$ to obtain $C(x)$. 
\par To see why this scheme is secure against maulers, suppose an adversary $\cA$ given  $( \widetilde{C},\pi )$ produces $( \widetilde{C}^*,\pi^* )$, where $\widetilde{C}^* \neq \widetilde{C}$. Since $\cA$ is a valid mauler we are guaranteed that $\widetilde{C}^*$ is functionally equivalent to $C$. We first run the seNIZK simulator to simulate $\pi$ and once this is done, we no longer need $x$ to generate $\pi$. Now, we invoke the simulation-extractability property to convert $\cA$ into one who not only produces $( \widetilde{C}^*,\pi^* )$ but also simultaneously produces $x$ such that $\widetilde{C}^*(x)=1$. Since $\widetilde{C}^*$ is functionally equivalent to $C$, it follows that $C(x)=1$ as well. But this violates the input-hiding property which says that no efficient adversary given $\widetilde{C}$ can produce an accepting input. \\

\noindent {\em Issue: Checking Functional Equivalence.} There is a subtlety we skipped in the proof above. The maulers that we consider have multi-bit output which is atypical in the cryptographic setting where the focus is mainly on boolean adversaries. This causes an issue when we switch from the honestly generated proof to a simulated proof. Upon receiving the honestly generated proof, $\cA$ outputs $( \widetilde{C}^*,\pi^* )$ such that $\widetilde{C}^*$ is functionally equivalent to $C$ but upon receiving the simulated proof, the adversary outputs $( \widetilde{C}^*,\pi^* )$ where $\widetilde{C}^*$ differs from $C$ on one point. From $\cA$, we need to extract one bit that would help distinguish the real and simulated proofs. To extract this bit, we rely upon sub-exponential security. Given $\widetilde{C}^*$, we run in time $2^n$, where $n$ is the input length, and check if $\widetilde{C}^*$ is still functionally equivalent to $C$; if indeed $\widetilde{C}^*$ is not functionally equivalent to $C$ then we know for a fact that the adversary was given a simulated proof, otherwise it received an honestly generated proof. We set the security parameter in the seNIZK system to be sufficiently large (for eg, $\poly(n)$) such that the seNIZK is still secure against adversaries running in time $2^n$. 

\paragraph{SSL against Duplicators.} Next we focus on constructing SSL secure against duplicators. If our only goal was to protect against duplicators, we could achieve this with a simple scheme. The lessor, in order to lease $C$, will output $(\ket{\psi}, C)$ where $\ket{\psi}$ is a random quantum state generated by applying a random polynomial sized quantum circuit $U$ on input $\ket{0^{\otimes \secparam}}$.  $\run$  on input $(\ket{\psi},C,x)$ ignores the quantum state $\ket{\psi}$, and outputs $C(x)$. By quantum no-cloning, an attacker cannot output two copies of $(\ket{\psi},C)$, which means that this scheme is already secure against duplicators. 
\par Recall that we focused on designing SSL for duplicators in the hope that it will be later helpful for designing SSL for arbitrary pirates. But any SSL scheme in which $\run$ ignores the quantum part would not be useful for obtaining SSL secure against arbitrary pirates; an attacker can simply replace the quantum state as part of the leased state with its own quantum state and copy the classical part. To overcome this insufficiency, we need to design SSL schemes where the Run algorithm only computes correctly when the input leased state belongs to a sparse set of quantum states. This suggests that the Run algorithm implicitly satisfies a verifiability property; it should be able to verify that the input quantum state lies in this sparse set. \\

\noindent \textit{Publicly Verifiable Unclonable States.}
We wish to construct a family of efficiently preparable states $\{\ket{\psi_s}\}_s$ with the following verifiability property. For any state $\ket{\psi_s}$ in the family, there is a way to sample a classical description $d_s$ for $\ket{\psi_s}$ in such a way that it can be verified that $d_s$ is a corresponding description of $\ket{\psi_s}$.  To be more precise, there should be a verification algorithm $\mathsf{Ver}(\ket{\psi_s},d)$ that accepts if $d$ is a valid description for $\ket{\psi_s}$. Furthermore, we want the guarantee that given a valid pair $(\ket{\psi_s},d_s)$, no QPT adversary can produce $\ket{\psi_s}^{\otimes 2}$.

 Our requirement has the same flavor as public-key quantum money, but a key difference is that we do not require any secret parameters associated with the scheme. Moreover, we allow anyone to be able to generate such tuples $(\ket{\psi_s},d_s)$ and not just the minting authority (bank). 

Given such verifiable family, we can define the $\run$ algorithm as follows,\\

\noindent {$\run(C, (\ket{\psi_s},d),x)$:}
\begin{itemize}
    \item If $\mathsf{Ver}(\ket{\psi_s},d)=0$, output $\bot$.
    \item Otherwise, output $C(x)$.
\end{itemize}

\noindent Any lessor can now lease a state $(\ket{\psi_s},d_s,C)$, which would allow anyone to compute $C$ using $\run$. Of course, any pirate that is given $(\ket{\psi_s},d_s,C)$ can prepare their own $(\ket{\psi_{s'}},d_{s'})$ and then input $(\ket{\psi_{s'}},d_{s'},C)$ into $\run$. But recall that we are only interested in ruling out \textit{duplicators}. From the public verifiable property of the quantum states, we have the fact that no QPT pirate could prepare $\ket{\psi_s}^{\otimes 2}$ from $(\ket{\psi_s},d_s)$ and thus, it is computationally infeasible to duplicate the leased state.\\

\noindent \textit{Publicly Verifiable Unclonable States from Subspace Hiding Obfuscation.} The notion of publicly verifiable unclonable states was first realized by Zhandry~\cite{Zha19}. The main tool used in Zhandry's construction is yet another notion of obfuscation, called subspace hiding obfuscation. Roughly speaking, a subspace hiding obfuscator ($\shO$) takes as input a description of a linear subspace $A$, and outputs a circuit that computes the membership function for $A$, i.e. $\shO(A)(x)=1$ iff $x\in A$. Zhandry shows that for a uniformly random  $\frac{\secparam}{2}$-dimensional subspace $A \subset \Zq^{\secparam}$, given $\ket{A}:=\frac{1}{\sqrt{q^{\secparam/2}}}\underset{a \in A}{\sum} \ket{a}$ along with $\widetilde{g}\leftarrow\shO(A),\widetilde{g_\perp} \leftarrow \shO(A^\perp)$, no QPT algorithm can prepare $\ket{A}^{\otimes 2}$ with non-negligible probability. Nevertheless, because $\widetilde{g}$ and $\widetilde{g_\perp}$ compute membership for $A$ and  $A^\perp$ respectively, it is possible to project onto $\ket{A}\bra{A}$ using $(\widetilde{g},\widetilde{g_\perp})$. This lets anyone check the tuple $(\ket{\psi},(\widetilde{g},\widetilde{g_\perp}))$ by measuring $\ket{\psi}$ with the projectors $\{\ket{A}\bra{A},I-\ket{A}\bra{A}\}$. 

\paragraph{Main Template: SSL against Pirates.} Our goal is to construct SSL against arbitrary QPT pirates and not just duplicators or maulers. To achieve this goal, we combine the techniques we have developed so far. 
\par To lease a circuit $C$, do the following:
\begin{enumerate}
        \item First prepare the state the state $\ket{A} = \frac{1}{\sqrt{q^{\secparam/2}}} \underset{a \in A}{\sum} \ket{a}$, along with $\tilde{g}\leftarrow \shO(A)$ and $\widetilde{g_\perp}\leftarrow \shO(A^\perp)$.
        \item Compute an input-hiding obfuscation of $C$, namely  $\widetilde{C}$. 
        \item Let $x$ be an accepting point of $C$. This can be determined using the searchability condition. 
        \item Compute a seNIZK proof $\pi$ such that: (1) the obfuscations $( \widetilde{g},\widetilde{g_\perp}, \widetilde{C})$ were computed correctly,  as a function of $(A, A^\perp, C)$, and, (2) $C(x)=1$. 
        \item Output $\ket{\psi_C} =( \ket{A},\widetilde{g},\widetilde{g_{\bot}},\widetilde{C},\pi)$. 
\end{enumerate}
\noindent The $\run$ algorithm on input $(\ket{\psi_C},\widetilde{g},\widetilde{g_\perp},\widetilde{C},\pi)$ and $x$, first checks the proof $\pi$, and outputs $\bot$ if it does not accept the proof. If it accepts the proof, it knows that  $\widetilde{g}$ and $\widetilde{g_\perp}$ are subspace obfuscators for some subspaces $A$ and $A^\perp$ respectively; it can use them to project $\ket{\psi_C}$ onto $\ket{A}\bra{A}$. This checks whether $\ket{\psi_C}$ is the same as $\ket{A}$ or not. If it is not, then it outputs $\bot$.  If it has not output $\bot$ so far, then it computes $\widetilde{C}$ on $x$ to obtain $C(x)$.\\
\ \\
\noindent {\em Proof Intuition}: To prove the lessor security of the above scheme, we consider two cases depending on the behavior of the pirate: 
    \begin{itemize}
    \setlength\itemsep{0.7em}
        \item {\em Duplicator}: in this case, the pirate produces a new copy that is of the form $( \sigma^*,\widetilde{g},\widetilde{g_{\bot}},\widetilde{C},\pi)$; that is, it has the same classical part as before. If $\sigma^*$ is close to $\ket{A}\bra{A}$, it would violate the no-cloning theorem. On the other hand, if $\sigma^*$ is far from $\ket{A}\bra{A}$, we can argue that the execution of $\run$ on the copy produced by the pirate will not compute $C$. The reason being that at least one of the two subspace obfuscators $\widetilde{g},\widetilde{g_{\bot}}$ will output $\bot$ on the state $\sigma^*$. 
        \item {\em Mauler}: suppose the pirate produces a new copy that is of the form $( \sigma^*,\widetilde{g}^*,\widetilde{g_{\bot}}^*,\allowbreak \widetilde{C}^*,\pi^*)$ such that $(\widetilde{g}^*,\widetilde{g_{\bot}}^*,\widetilde{C}^* ) \neq (\widetilde{g},\widetilde{g_{\bot}},\widetilde{C})$. We invoke the simulation-extractability property to find an input $x$ such that $\widetilde{C}^*(x)=1$. Since $\widetilde{C}^*$ is assumed to have the same functionality as $C$, this means that $C(x)=1$. This would contradict the security of input-hiding obfuscation, since any QPT adversary even given $\widetilde{C}$ should not be able to find an accepting input $x$ such that $C(x)=1$. 
    \end{itemize}

\submversion{
\paragraph{Organization.} We provide the related works  and preliminary background in the Supplementary material, in Sections~\ref{ssec:RW} and~\ref{sec:prelims} respectively. We present the formal definition of secure software leasing in Section~\ref{sec:SSL}. The impossibility result is presented in Section~\ref{sec:imp}. Finally, we present the positive result in Section~\ref{sec:const}.     
}

\fullversion{
\submversion{
\section{Related Work}
}
\fullversion{
\subsection{Related Work}
}
\label{ssec:RW}

SSL is an addition to the rapidly growing list of quantum cryptographic primitives with the desirable property of unclonability, and hence impossible to achieve classically. Besides the aforementioned connections to software copy-protection, our work on SSL is related to the following previous works.

\paragraph{Quantum Money and Quantum Lightning.}  Using quantum mechanics to achieve unforgeability has a history that predates quantum computing itself. Wiesner \cite{wiesner1983conjugate} informally introduced the notion of unforgeable quantum money -- unclonable quantum states that can also be (either publicly or privately) verified to be valid states. A few constructions~\cite{Aar09,Lut+09,Gav11,Farhi12,AC12} achieved quantum money with various features and very recently, in a breakthrough work, Zhandry~\cite{Zha19} shows how to construct publicly-verifiable quantum money from cryptographic assumptions. Zhandry also introduced a stronger notion of quantum money, which he coined quantum lightning, and constructed it from cryptographic assumptions.  

\paragraph{Certifiable Deletion and Unclonable Encryption.} Unclonability has also been studied in the context of encryption schemes. The work of~\cite{gottesman2003uncloneable} studies the problem of quantum tamper detection. Alice can use a quantum state to send Bob an encryption of a classical message $m$ with the guarantee that any eavsdropper could not have cloned the ciphertext. After Bob receives the ciphertext, he can check if the state has been tampered with, and if this is not the case, he would know that a potential eavsdropper did not keep a copy of the ciphertext. In recent work, Broadbent and Lord \cite{broadbent2019uncloneable} introduced the notion of unclonable encryption. Roughly speaking, an unclonable encryption allows Alice to give Bob and Charlie an encryption of a classical message $m$, in the form of a quantum state $\sigma(m)$, such that Bob and Charlie cannot `split' the state among them.  

In a follow-up work, Broadbent and Islam \cite{broadbent2019quantum}, construct a one-time use encryption scheme with certifiable deletion. An encryption scheme has certifiable deletion, if there is an algorithm to check that a ciphertext was deleted.  The security guarantee is that if an adversary is in possession of the ciphertext, and it then passes the certification of deletion, the issuer of the encryption can now give the secret key to the adversary. At this point, the adversary still can't distinguish which plaintext correspond to the ciphertext it was given. 

\paragraph{Quantum Obfuscation.} Our proof of the impossibility of SSL is inspired by the proof of Barak et al. \cite{Bar01} on the impossibility of VBB for arbitrary functions. Alagic and Fefferman\cite{alagic2016quantum} formalized the notion of program obfuscation via quantum tools, defining quantum virtual black-box obfuscation (qVBB) and quantum indistinguishability obfuscation (qiO), as the natural quantum analogues to the respective classical notions (VBB and iO). They also proved quantum analogues of some of the previous impossibility results from \cite{Bar01}, as well as provided quantum cryptographic applications from qVBB and qiO.

\paragraph{Quantum One-Time Programs and One-Time Tokens.} One natural question to ask is if quantum mechanics alone allows the existence of `one-time' use cryptographic primivites. Quantum One-Time programs, that use only quantum information, are not possible even under computational assumptions \cite{broadbent2013quantum}. This rules out the possibility of having a copy-protection scheme where a single copy of the software is consumed by the evaluation procedure. Despite the lack of quantum one-time programs, there are constructions of secure `one-time' signature tokens in the oracle models \cite{ben2016quantum} \cite{AGKZ20}. A quantum token for signatures is a quantum state that would let anyone in possession of it to sign an arbitrary document, but only once.  The token is destroyed in the signing process.   

\paragraph{Quantum Tomography.} Quantum tomography is the task of learning a description of a mixed state $\rho$ given multiple copies of it \cite{haah2017sample} \cite{o2016efficient}.  One possible way to break SSL (or copy-protection) would be to learn a valid description of the state $\rho_C$ directly from having access to multiple copies of the leased program, $\rho_C^{\otimes k}$. Indeed, in recent work in this area, Aaronson \cite{aaronson2018shadow} showed that in order for copy-protection to be possible at all it must be based on computational assumptions.

\paragraph{Recent Work on Copy-Protection.} While finishing this manuscript, we became aware of very recent work on copy-protection. Aaronson et al. \cite{aaronson2020quantum} constructed copy-protection for unlearnable functions relative to a classical oracle. Our work complements their results, since we show that obtaining copy-protection in the standard model (i.e., without oracles) is not possible.
}

\fullversion{
\paragraph{Acknowledgements.} We thank Alex Dalzell for helpful discussions. During this work, RL was funded by NSF grant CCF-1729369 MIT-CTP/5204.
}

\fullversion{
\newcommand{\qnizk}{\mathsf{qNIZK}}
\newcommand{\pke}{\mathsf{qPKE}}
\newcommand{\fkgen}{\fksetup}
\newcommand{\expt}{\mathsf{Expt}}
\newcommand{\evnt}{\mathsf{Process}}
\newcommand{\view}{\mathsf{View}}
\section{Preliminaries}
\label{sec:prelims}
We assume that the reader is familiar with basic cryptographic notions such as negligible functions and computational indistinguishability (see~\cite{Gol01}). 
\par The security parameter is denoted by $\secparam$ and we denote $\negl(\secparam)$ to be a negligible function in $\secparam$. We denote (classical) computational indistiguishability of two distributions $\distr_0$ and $\distr_1$ by $\distr_0 \approx_{c,\varepsilon} \distr_1$. In the case when $\varepsilon$ is negligible, we drop $\varepsilon$ from this notation. 

\subsection{Quantum}
\label{ssec:notation}

For completeness, we present some of the basic quantum definitions, for more details see \cite{nielsen2002quantum}.
\paragraph{Quantum states and channels.} Let $\cH$ be any finite Hilbert space, and let $L(\cH):=\{\cE:\cH \rightarrow \cH \}$ be the set of all linear operators from $\cH$ to itself (or endomorphism). Quantum states over $\cH$ are the positive semidefinite operators in $L(\cH)$ that have unit trace, we call these density matrices, and use the notation $\rho$ or $\sigma$ to stand for density matrices when possible. Quantum channels or quantum operations acting on quantum states over $\cH$ are completely positive trace preserving (CPTP) linear maps from $L(\cH)$ to $L(\cH')$ where $\cH'$ is any other finite dimensional Hilbert space. 
We use the trace distance, denoted by $\trD{\rho}{\sigma}$, as our distance measure on quantum states, 
$$\trD{\rho}{\sigma} = \frac{1}{2} \tr \left[\sqrt{\left(\rho-\sigma\right)^\dagger\left(\rho-\sigma\right)}\right] $$

A state over $\cH=\mathbb{C}^2$ is called a qubit. For any $n \in \mathbb{N}$, we refer to the quantum states over $\cH = (\mathbb{C}^2)^{\otimes n}$ as $n$-qubit quantum states. To perform a standard basis measurement on a qubit means projecting the qubit into $\{\ket{0},\ket{1}\}$. A quantum register is a collection of qubits. A classical register is a quantum register that is only able to store qubits in the computational basis.

A unitary quantum circuit is a sequence of unitary operations (unitary gates) acting on a fixed number of qubits. Measurements in the standard basis can be performed at the end of the unitary circuit. A (general) quantum circuit is a unitary quantum circuit with $2$ additional operations: $(1)$ a gate that adds an ancilla qubit to the system, and $(2)$ a gate that discards (trace-out) a qubit from the system. A quantum polynomial-time algorithm (QPT) is a uniform collection of quantum circuits $\{C_n\}_{n \in \mathbb{N}}$. We always assume that the QPT adversaries are non-uniform -- a QPT adversary $\cA$ acting on $n$ qubits could be given a quantum auxiliary state with $\poly(n)$ qubits.

\paragraph{Quantum Computational Indistinguishability.}

When we talk about quantum distinguishers, we need the following definitions, which we take from \cite{Wat09}.
\begin{definition}[Indistinguishable collections of states] Let $I$ be an infinite subset $I \subset \{0,1\}^*$, let $p : \mathbb{N} \rightarrow \mathbb{N}$ be a polynomially bounded function, and let $\rho_{x}$ and $\sigma_x$ be $p(|x|)$-qubit states. We say that $\{\rho_{x}\}_{x \in I}$ and $\{\sigma_x\}_{x\in I}$ are \textbf{quantum computationally indistinguishable collections of quantum states} if for every QPT $\cE$ that outputs a single bit, any polynomially bounded  $q:\mathbb{N}\rightarrow \mathbb{N}$, and any auxiliary $q(|x|)$-qubits state $\nu$, and for all $x \in I$, we have that
$$\left|\Pr\left[\cE(\rho_x\otimes \nu)=1\right]-\Pr\left[\cE(\sigma_x \otimes \nu)=1\right]\right| \leq \epsilon(|x|) $$
for some function $\epsilon:\mathbb{N}\rightarrow [0,1]$.  We use the following notation 
$$\rho_x \approx_{Q,\epsilon} \sigma_x$$
and we ignore the $\epsilon$ when it is understood that it is a negligible function.
\end{definition}

\begin{definition}[Indistinguishability of channels] Let $I$ be an infinite subset $I \subset \{0,1\}^*$, let $p,q: \mathbb{N} \rightarrow \mathbb{N}$ be polynomially bounded functions, and let $\cD_x,\cF_x$
be quantum channels mapping $p(|x|)$-qubit states to $q(|x|)$-qubit states. We say that $\{\cD_x\}_{x \in I}$ and $\{\cF_x\}_{x \in I}$ are \textbf{quantum computationally indistinguishable collection of channels} if for every QPT $\cE$ that outputs a single bit, any polynomially bounded $t : \mathbb{N} \rightarrow \mathbb{N}$, any $p(|x|)+t(|x|)$-qubit quantum state $\rho$, and for all $x\in I$, we have that
$$ \left|\Pr\left[\cE\left((\cD_x\otimes \Id)(\rho)\right)=1\right]-\Pr\left[\cE\left((\cF_x\otimes \Id)(\rho)\right)=1\right]\right|\leq \epsilon(|x|) $$
for some function $\epsilon:\mathbb{N}\rightarrow [0,1]$. We will use the following notation
$$ \cD_x(\cdot) \approx_{Q,\epsilon} \cF_x(\cdot)$$
and we ignore the $\epsilon$ when it is understood that it is a negligible function.

\end{definition}

\paragraph{Quantum Fourier Transform and Subspaces.} Our main construction uses the same type of quantum states (superpositions over linear subspaces) considered by\cite{aaronson2012quantum,Zha19} in the context of constructing quantum money. 

We recall some key facts from these works relevant to our construction. 
Consider the field $\Zq^{\secparam}$ where $q\geq2$,and let $\ft$ denote the quantum fourier transfrom over $\Zq^{\secparam}$.

For any linear subspace $A$, let $A^\perp$ denote its orthogonal (dual) subspace, $$A^\perp = \{v \in \Zq^\secparam| \langle v, a\rangle = 0 \}.$$

\noindent Let $\ket{A} = \frac{1}{\sqrt{|A|}} \underset{a\in A}{\sum}\ket{a}$.  The quantum fourier Transform, $\ft$, does the following:
$$\ft \ket{A} = \ket{A^\perp}. $$
Since $(A^\perp)^\perp = A$, we also have $\ft \ket{A^\perp} = \ket{A}$.

Let $\Pi_{A}=\underset{a \in A}{\sum} \ket{a}\bra{a}$, then as shown in Lemma 21 of~\cite{aaronson2012quantum},
$$\ft (\Pi_{A^\perp}) \ft \Pi_{A} = \ket{A}\bra{A}.$$

\paragraph{Almost As Good As New Lemma.} We use the Almost As Good As New Lemma \cite{aaronson2004limitations}, restated here verbatim from \cite{aaronson2016complexity}.

\begin{lemma}[Almost As Good As New]\label{clm:ASGAN} Let $\rho$ be a mixed state acting on $\mathbb{C}^d$. Let $U$ be a unitary and $(\Pi_0,\Pi_1=1-\Pi_0)$ be projectors all acting on $\mathbb{C}^d \otimes \mathbb{C}^d$. We interpret $(U,\Pi_0,\Pi_1)$ as a measurement performed by appending an acillary system of dimension $d'$ in the state $\ket{0}\bra{0}$, applying $U$ and then performing the projective measurement $\{\Pi_0,\Pi_1\}$ on the larger system. Assuming that the outcome corresponding to $\Pi_0$ has probability $1-\varepsilon$, i.e., $\tr[\Pi_0(U\rho \otimes \ket{0}\bra{0}U^\dagger)]=1-\varepsilon$, we have $$\trD{\rho}{\widetilde{\rho}} \leq \sqrt{\varepsilon} ,$$
where $\widetilde{\rho}$ is state after performing the measurement and then undoing the unitary $U$ and tracing out the ancillary system: $$\widetilde{\rho} = \tr_{d'}\left(U^\dagger \left( \Pi_0U \left( \rho \otimes \ket{0}\bra{0} \right)U^\dagger \Pi_0 + \Pi_1U \left( \rho \otimes \ket{0}\bra{0} \right)U^\dagger \Pi_1\right) U \right) $$
\end{lemma}

We use this Lemma to argue that whenever a QPT algorithm $\cA$ on input $\rho$, outputs a particular bit string $z$ with probability $1-\varepsilon$, then $\cA$ can be performed in a way that also lets us recover the initial state. In particular, given the QPT description for $\cA$, we can implement $\cA$ with an acillary system, a unitary, and only measuring in the computational basis after the unitary has been applied, similarly to Lemma~\ref{clm:ASGAN}. Then, it is possible to uncompute in order to also obtain $\widetilde{\rho}$.

\paragraph{Notation about Quantum-Secure Classical Primitives.} For a classical primitive X, we use the notation q-X to denote the fact that we assume X to be secure against QPT adversaries. 

\subsection{Learning with Errors}
\label{sec:prelims:lwe}

\noindent We consider the decisional learning with errors (LWE) problem, introduced by Regev~\cite{Reg05}. We define this problem formally below. 

\begin{quote}
    The problem $(n,m,q,\chi)$-LWE, where $n,m,q \in \mathbb{N}$ and  $\chi$ is a distribution supported over $\mathbb{Z}$, is to distinguish between the distributions $(\bfA,\bfA \bfs + \bfe)$ and $(\bfA,\bfu)$, where $\bfA \xleftarrow{\$} \mathbb{Z}_q^{m \times n},\bfs \xleftarrow{\$} \mathbb{Z}_q^{n \times 1},\bfe \xleftarrow{\$} \chi^{m \times 1}$ and $\bfu \leftarrow \mathbb{Z}_q^{m \times 1}$.
\end{quote}

\noindent The above problem has been believed to be hard against classical PPT algorithms -- also referred to as LWE assumption -- has had many powerful applications in cryptography. In this work, we conjecture the above problem to be hard even against QPT algorithms; this conjecture referred to as QLWE assumption has been useful in the constructions of interesting primitives such as quantum fully-homomorphic encryption~\cite{mahadev2018classical,Bra18}. We refer to this assumption as QLWE assumption. 

\begin{quote}
    {\bf QLWE assumption}: This assumption is parameterized by $\secparam$. Let $n=\poly(\secparam)$, $m=\poly(n \cdot \log(q))$ and $\chi$ be a discrete Gaussian distribution\footnote{Refer~\cite{Bra18} for a definition of discrete Gaussian distribution.} with parameter $\alpha q > 0$, where $\alpha$ can set to be any non-negative number.  
    \par Any QPT distinguisher (even given access to polynomial-sized advice state) can solve $(n,m,q,\chi)$-LWE only with  probability $\negl(\secparam)$, for some negligible function $\negl$. 
\end{quote}

\begin{remark}
We drop the notation $\secparam$ from the description of the assumption when it is clear. 
\end{remark}

$(n,m,q,\chi)$-LWE is shown~\cite{Reg05,PRS17} to be as hard as approximating shortest independent vector problem (SIVP) to within a factor of $\gamma=\tilde{O}(n/\alpha)$ (where $\alpha$ is defined above). The best known quantum algorithms for this problem run in time $2^{\tilde{O}(n/\log(\gamma))}$. 
\par For our construction of SSL, we require a stronger version of QLWE that is secure even against sub-exponential quantum adversaries. We state this assumption formally below. 

\begin{quote}
    {\bf $T$-Sub-exponential QLWE Assumption}: This assumption is parameterized by $\secparam$ and time $T$. Let $n=T+\poly(\secparam)$, $m=\poly(n \cdot \log(q))$ and $\chi$ be a discrete Gaussian distribution with parameter $\alpha q > 0$, where $\alpha$ can set to be any non-negative number.  
    \par Any quantum distinguisher (even given access to polynomial-sized advice state) running in time $2^{\widetilde{O}(T)}$ can solve $(n,m,q,\chi)$-LWE only with probability $\negl(\secparam)$, for some negligible function $\negl$. 
\end{quote}


\subsection{Quantum Fully Homomorphic Encryption}
\label{sec:prelims:qfhe}

A fully homomorphic encryption scheme allows for publicly evaluating an encryption of $x$ using a function $f$ to obtain an encryption of $f(x)$. Traditionally $f$ has been modeled as classical circuits but in this work, we consider the setting when $f$ is modeled as quantum circuits and when the messages are quantum states. This notion is referred to as quantum fully homomorphic encryption (QFHE). We state our definition verbatim from \cite{broadbent2015quantum}.

\begin{definition} Let $\cM$ be the Hilbert space associated with the message space (plaintexts), $\cC$ be the Hilbert space associated with the ciphertexts, and $\cR_{evk}$ be the Hilbert space associated with the evaluation key. A quantum fully homomorphic encryption scheme is a tuple of QPT algorithms  $\qfhe=(\gen,\enc,\dec,\allowbreak \eval)$ satisfying 
\begin{itemize}
    \item $\qfhe.\gen(1^\secparam)$: outputs a a public and a secret key, $(\pk,\sk)$, as well as a quantum state $\rho_{evk}$, which can serve as an evaluation key. 
    \item $\qfhe.\enc(\pk,\cdot):L(\cM)\rightarrow L(\cC)$:  takes as input a state $\rho$ and outputs a ciphertext $\sigma$
    \item $\qfhe.\dec(\sk,\cdot):L(\cC)\rightarrow L(\cM)$: takes a quantum ciphertext $\sigma$, and outputs a qubit $\rho$ in the message space $L(\cM)$.
    \item $\qfhe.\eval(\cE, \cdot ):L(\cR_{evk}\otimes \cC^{\otimes n})\rightarrow L(\cC^{\otimes m})$: takes as input a quantum circuit $\cE: L(\cM^{\otimes n}) \rightarrow L(\cM^{\otimes m})$, and a ciphertext in $L(\cC^{\otimes n})$ and outputs a ciphertext in $L(\cC^{\otimes m})$, possibly consuming the evaluation key $\rho_{evk}$ in the proccess.

\end{itemize}
\end{definition}
\noindent Semantic security and compactness are defined analogously to the classical setting, and we defer to~\cite{broadbent2015quantum} for a definition.
\noindent For the impossibility result, we require a $\qfhe$ scheme where ciphertexts of classical plaintexts are also classical. Given any $x \in \{0,1\}$, we want $\qfhe.\enc_\pk(\ket{x}\bra{x})$ to be a computational basis state $\ket{z}\bra{z}$ for some $z \in \{0,1\}^l$ (here, $l$ is the length of ciphertexts for 1-bit messages). In this case, we write $\qfhe.\enc_\pk(x)$. We also want the same to be true for evaluated ciphertexts, i.e. if $\cE(\ket{x}\bra{x})=\ket{y}\bra{y}$ for some $x \in \{0,1\}^n$ and $y \in \{0,1\}^m$, then 
$$\qfhe.\enc_\pk(y) \leftarrow \qfhe.\eval(\rho_{evk}, \cE, \qfhe.\enc_\pk(x)) $$
is a classical ciphertext of $y$.

\paragraph{Instantiation.} The works of~\cite{mahadev2018classical,Bra18} give lattice-based candidates for quantum fully homomorphic encryption schemes; we currently do not know how to base this on learning with errors alone\footnote{Brakerski~\cite{Bra18} remarks that the security of their candidate can be based on a circular security assumption that is also used to argue the security of existing constructions of unbounded depth multi-key FHE~\cite{CM15,MW16,PS16,BP16}.}. The desirable property required from the quantum FHE schemes, that classical messages have classical ciphertexts, is satisfied by both candidates~\cite{mahadev2018classical,Bra18}.

\subsection{Circuit Class of Interest: Evasive Circuits}
\noindent The circuit class we consider in our construction of SSL is a subclass of evasive circuits. We recall the definition of evasive circuits below. 

\paragraph{Evasive Circuits.} Informally, a class of circuits is said to be evasive if a circuit drawn from a suitable distribution outputs 1 on a fixed point with negligible probability. 

\begin{definition}[Evasive Circuits]
A class of circuits $\cktclass=\{\cktclass_{\secparam}\}_{\secparam \in \mathbb{N}}$, associated with a distribution $\distrc$, is said to be {\bf evasive} if the following holds: for every $\secparam \in \mathbb{N}$, every $x \in \{0,1\}^{\poly(\secparam)}$,
$$\underset{C \leftarrow \distrc}{\prob} \left[ C(x)=1  \right] \leq \negl(\secparam),$$
\end{definition}

\paragraph{Compute-and-compare Circuits.} The subclass of circuits that we are interested in is called compute-and-compare circuits, denoted by $\cktclass_{\cnc}$. A compute-and-compare circuit is of the following form: $\lockC[C,\alpha]$, where $\alpha$ is called a lock and $C$ has output length $|\alpha|$, is defined as follows: 
$$ \lockC[C,\alpha](x) =  \left\{\substack{1, \ \text{ if } C(x)=\alpha,\\ \ \\ 0,\ \text { otherwise }  } \right. $$   

\paragraph{Multi-bit compute-and-compare circuits.} We can correspondingly define the notion of multi-bit compute-and-compare circuits. A multi-bit compute-and-compare circuit is of the following form: 
$$ \lockC[C,\alpha,\msg](x) =  \left\{\substack{\msg, \ \text{ if } C(x)=\alpha,\\ \ \\ 0,\ \text { otherwise }  } \right.,$$
where $\msg$ is a binary string.

\noindent We consider two types of distributions as defined by~\cite{WZ17}. 

\begin{definition}[Distributions for Compute-and-Compare Circuits]
\label{def:distr:cnc}
We consider the following distributions on $\cktclass_{\cnc}$: 
\begin{itemize}
    \item {\bf $\distr_{\unpred}(\secparam)$}: For any $(\lockC[C,\alpha])$ along with $\aux$  sampled from this unpredictable distribution, it holds that $\alpha$ is computationally unpredictable given $(C,\aux)$.   
    \item {\bf $\distr_{\pseud}(\secparam)$}: For any $\lockC[C,\alpha]$ along with $\aux$ sampled from this distribution, it holds that $\bfH_{\hill}\left(\alpha|(C,\aux) \right) \geq \secparam^{\varepsilon}$, for some constant $\epsilon > 0$, where $\bfH_{\hill}(\cdot)$ is the HILL entropy~\cite{HILL99}. 
\end{itemize}
\end{definition}

\noindent Note that with respect to the above distributions, the compute-and-compare class of circuits $\cktclass_{\cnc}$ is evasive.  

\paragraph{Searchability.} For our construction of SSL for $\cktclass$, we crucially use the fact that given a circuit $C \in \cktclass$, we can read off an input $x$ from the description of $C$ such that $C(x)=1$. We formalize this by defining a search algorithm $\search$ that on input a circuit $C$ outputs an accepting input for $C$. For many interesting class of functions, there do exist a corresponding efficiently implementable class of circuits associated with a search algorithm $\search$. 

\begin{definition}[Searchability]
A class of circuits $\cktclass=\{\cktclass_{\secparam}\}_{\secparam \in \mathbb{N}}$ is said to be {\bf $\search$-searchable}, with respect to a PPT algorithm $\search$, if the following holds: on input $C$, $\search(C)$ outputs $x$ such that $C(x)=1$.  
\end{definition}

\paragraph{Searchable Compute-and-Compare Circuits: Examples.} As mentioned in the introduction, there are natural and interesting classes of searchable compute-and-compare circuits. For completeness, we state them again below with additional examples~\cite{WZ17}. 
\begin{itemize}
    \item Point circuits $C(\alpha, \cdot)$: the circuit $C(\alpha,\cdot)$ is a point circuit if it takes as input $x$ and outputs $C(\alpha,x)=1$ iff $x=\alpha$. If we define the class of point circuits suitably, we can find $\alpha$ directly from $C_{\alpha}$; for instance, $\alpha$ can be  the value assigned to the input wires of $C$. 
    \item Conjunctions with wild cards $C(S,\alpha, \cdot)$: the circuit $C(S,\alpha,\cdot)$ is a conjunction with wild cards if it takes as input $x$ and outputs $C(S,\alpha,x)=1$ iff $y = \alpha$, where $y$ is such that $y_i=x_i$ for all $i \in S$.
    Again, if we define this class of circuits suitably, we can find $S$ and $\alpha$ directly from the description of $C(S,\alpha,\cdot)$. Once we find $S$ and $\alpha$, we can find the accepting input.  
    \item Affine Tester: the circuit $C(\mathbf{A},\alpha,\cdot)$ is an affine tester, with $\mathbf{A},\mathbf{y}$ where $\mathbf{A}$ has a non-trivial kernel space, if it takes as input $\mathbf{x}$ and outputs $C(\mathbf{A},\alpha,\mathbf{x})=1$ iff $\mathbf{A} \cdot \mathbf{x} = \alpha$. By reading off $\mathbf{A}$ and $\alpha$ and using Gaussian elimination we can find $\mathbf{x}$ such that $\mathbf{A} \cdot \mathbf{x} = \alpha$.  
    \item Plaintext equality checker $C(\sk,\alpha,\cdot)$: the circuit $C(\sk,\alpha,\cdot)$, with hardwired values decryption key $\sk$ associated with a private key encryption scheme, message $\alpha$, is a plaintext equality checker if it takes as input a ciphertext $\ct$ and outputs $C(\sk,\alpha,\ct)=1$ iff the decryption of $\ct$ with respect to $\sk$ is $\alpha$. By reading off $\alpha$ and $\sk$, we can find a ciphertext such that $\ct$ is an encryption of $\alpha$. 
\end{itemize}

\begin{remark}
We note that both the candidate constructions of copy-protection  for point functions by Aaronson~\cite{Aar09} use the fact that the accepting point of the point function is known by whoever is generating the copy-protected circuit.
\end{remark}

\subsection{Obfuscation}
In this work, we use different notions of cryptographic obfucation. We review all the required notions below, but first we recall the functionality of obfuscation. 

\begin{definition}[Functionality of Obfuscation]
Consider a class of circuits $\cktclass$. An obfuscator $\obfuscator$ consists of two PPT algorithms $\obf$ and $\eval$ such that the following holds: for every $\secparam \in \mathbb{N}$, circuit $C \in \cktclass$, $x \in \{0,1\}^{\poly(\secparam)}$, we have $C(x) \leftarrow \eval(\widetilde{C},x)$ where $\widetilde{C} \leftarrow \obf(1^{\secparam},C)$.  
\end{definition}

\subsubsection{Lockable Obfuscation}
\label{ssec:prelims:lobfs}

\noindent In the impossibility result, we will make use of program obfuscation schemes that are (i) defined for compute-and-compare circuits and, (ii) satisfy distributional virtual black box security notion~\cite{BGIRSVY01}. Such obfuscation schemes were first introduced by~\cite{WZ17,GKW17} and are called lockable obfuscation schemes. We recall their definition, adapted to quantum security, below.

\begin{definition}[Quantum-Secure Lockable Obfuscation]
An obfuscation scheme $(\lobf,\leval)$ for a class of circuits $\cktclass$ is said to be a \textbf{quantum-secure lockable obfuscation scheme} if the following properties are satisfied: 
\begin{itemize}
    \item It satisfies the functionality of obfuscation. 
    \item {\bf Compute-and-compare circuits}: Each circuit $\lockC$ in $\cktclass$ is parameterized by strings $\alpha \in \{0,1\}^{\poly(\secparam)},\beta \in \{0,1\}^{\poly(\secparam)}$ and a poly-sized circuit $C$ such that on every input $x$, $\lockC(x)$ outputs $\beta$ if and only if $C(x)=\alpha$. 
    \item {\bf Security}: For every polynomial-sized circuit $C$, string $\beta \in \{0,1\}^{\poly(\secparam)}$,for every QPT adversary $\adversary$ there exists a QPT simulator $\simulator$ such that the following holds: sample $\alpha \xleftarrow{\$}  \{0,1\}^{\poly(\secparam)}$,
    $$\left\{ \lobf \left( 1^{\secparam},\lockC \right) \right\} \approx_{Q,\varepsilon} \left\{\simulator\left(1^{\secparam},1^{|C|} \right) \right\},$$
    where $\lockC$ is a circuit parameterized by $C,\alpha,\beta$ with $\varepsilon \leq \frac{1}{2^{|\alpha|}}$.  
\end{itemize}

\end{definition}

\paragraph{Instantiation.} The works of~\cite{WZ17,GKW17,GKVW19} construct a lockable obfuscation scheme based on polynomial-security of learning with errors (see Section~\ref{sec:prelims:lwe}). Since learning with errors is conjectured to be hard against QPT algorithms, the obfuscation schemes of~\cite{WZ17,GKW17,GKVW19} are also secure against QPT algorithms. Thus, we have the following theorem. 

\begin{theorem}[\cite{GKW17,WZ17,GKVW19}]
Assuming quantum hardness of learning with errors, there exists a quantum-secure lockable obfuscation scheme. 
\end{theorem}

\subsubsection{q-Input-Hiding Obfuscators}
\noindent One of the main tools used in our construction is q-input-hiding obfuscators. The notion of input-hiding obfuscators was first defined in the classical setting by Barak et al.~\cite{BBCKP14}. We adopt the same notion except that we require the security of the primitive to hold against QPT adversaries. \par The notion of q-input-hiding obfuscators states that given an obfuscated circuit, it should be infeasible for a QPT adversary to find an accepting input; that is, an input on which the circuit outputs 1. Note that this notion is only meaningful for the class of evasive circuits. 

The definition below is suitably adapted from Barak et al.~\cite{BBCKP14}; in particular, our security should hold against QPT adversaries. 

\begin{definition}[q-Input-Hiding Obfuscators~\cite{BBCKP14}]
An obfuscator $\qiho=(\obf,\eval)$ for a class of circuits associated with distribution $\distrc$ is {\bf q-input-hiding} if for every non-uniform QPT adversary $\adversary$, for every sufficiently large $\secparam \in \mathbb{N}$,
$$\prob \left[C(x)=1\ :\ \substack{C \leftarrow \distrc(\secparam),\\ \ \\ \widetilde{C} \leftarrow \obf(1^{\secparam},C), \\ \ \\ x \leftarrow \adversary(1^{\secparam},\widetilde{C})} \right] \leq \negl(\secparam).$$
\end{definition}

\subsubsection{Subspace Hiding Obfuscators}~\label{ssec:shO}
\noindent Another ingredient in our construction is subspace hiding obfuscation. Subspace hiding obfuscation is a notion of obfuscation introduced by Zhandry \cite{Zha19}, as a tool to build pulic-key quantum money schemes. This notion allows for obfuscating a circuit, associated with subspace $A$, that checks if an input vector belongs to this subspace $A$ or not. In terms of security, we require that the obfuscation of this circuit is indistinguishable from obfuscation of another circuit that tests membership of a larger random (and hidden) subspace containing $A$. 

\begin{definition}[\cite{Zha19}]
A subspace hiding obfuscator for a field $\mathbb{F}$ and dimensions $d_0,d_1,\secparam$ is a tuple $(\shO.\obf, \shO.\eval)$ satisfying:
\begin{itemize}
\item  $\shO.\obf(A)$: on input an efficient description of a linear subspace $A \subset \mathbb{F}^\secparam$ of dimensions $d \in \{d_0,d_1\}$ outputs an obfuscator $\shO(A)$. 
\item {\bf Correctness:} For any $A$ of dimension $d \in \{d_0,d_1\}$, it holds that
$$\prob[\forall x, \shO.\eval(\shO.\obf(A),x)=\mathbb{1}_{A}(x) :  \shO(A) \leftarrow \shO.\obf(A)] \geq 1-\negl(\secparam),$$
where: $\mathbb{1}_{A}(x)=1$ if $x \in A$ and 0, otherwise. 
\item {\bf Quantum-Security:}  Any QPT adversary $\cA$ can win the following challenge with probability at most negligibly greater than $\frac{1}{2}$.
\begin{enumerate}
    \item $\cA$ chooses a $d_0$-dimensional subspace $A \subset \mathbb{F}^\secparam$. 
    \item Challenger chooses uniformly at random a $d_1$-dimensional subspace $S \supseteq A$. It samples a random bit $b$.  If $b=0$, it sends $\widetilde{g_0} \leftarrow \shO.\obf(A)$. Otherwise, it sends $\widetilde{g_1} \leftarrow \shO.\obf(S)$
    \item $\cA$ receives $\widetilde{g_b}$ and outputs $b'$.  It wins if $b'=b$.
\end{enumerate}
\end{itemize}
\end{definition}

\paragraph{Instantiation.}  Zhandry presented a construction of subspace obfuscators from  indistinguishability obfuscation~\cite{BGIRSVY01,GGHRSW16} secure against QPT adversaries. 

\subsection{q-Simulation-Extractable Non-Interactive Zero-Knowledge}
We also use the tool of non-interactive zero-knowledge (NIZK) systems for NP for our construction. A NIZK is defined between a classical PPT prover $\prover$ and a verifier $\verify$. The goal of the prover is to convince the verifier $\verify$ to accept an instance $x$ using a witness $w$ while at the same time, not revealing any information about $w$. Moreover, any malicious prover should not be able to falsely convince the verifier to accept a NO instance. Since we allow the malicious parties to be QPT, we term this NIZK as qNIZK. 
\par We require the qNIZKs to satisfy a stronger property called simulation extractability and we call a qNIZK satisfying this stronger property to be q-simulation-extractable NIZK ($\ssnizk$). 
\par We describe the PPT algorithms of $\ssnizk$ below.   
\begin{itemize}
    \item $\crsgen(1^{\secparam})$: On input security parameter $\secparam$, it outputs the common reference string $\crs$. 
    \item $\prover(\crs,x,w)$: On input common reference string $\crs$, NP instance $x$, witness $w$, it outputs the proof $\pi$. 
    \item $\verify(\crs,x,\pi)$: On input common reference string $\crs$, instance $x$, proof $\pi$, it outputs accept or reject. This is a deterministic algorithm.
\end{itemize}
\noindent This notion is associated with the following properties. We start with the standard notion of completeness.

\begin{definition}[Completeness]
\label{def:qnizk:compl}
A non-interactive protocol $\ssnizk$ for a NP language $L$ is said to be {\bf complete} if the following holds: for every $(x,w) \in \rel(L)$, we have the following: 
$$\prob \left[ \verify(\crs,x,\pi)\ \text{accepts}\ :\ \substack{\crs \leftarrow \crsgen(1^{\secparam})\\ \ \\ \pi \leftarrow \prover(\crs,x,w)} \right] = 1$$
\end{definition}

\paragraph{q-Simulation-Extractability.} We now describe the simulation-extractability property. Suppose there exists an adversary who upon receiving many proofs $\pi_1,\ldots,\pi_q$ on all YES instances $x_1,\ldots,x_q$, can produce a proof  $\pi'$ on instance $x'$ such that: (a) $x'$ is different from all the instances $x_1,\ldots,x_q$ and, (b) $\pi'$ is accepting with probability $\varepsilon$. Then, this notion guarantees the existence of two efficient algorithms $\simr_1$ and $\simr_2$ 
such that all the proofs $\pi_1,\ldots,\pi_q$, are now simulated by $\simr_1$, and $\simr_2$ can extract a valid witness for $x'$ from $(x',\pi')$ produced by the adversary with probability negligibly close to $\varepsilon$.     

\begin{definition}[q-Simulation-Extractability]
A non-interactive protocol $\ssnizk$ for a language $L$ is said to satisfy {\bf  q-simulation-extractability}  if there exists a non-uniform QPT adversary $\adversary=(\adversary_1,\adversary_2)$ such that the following holds:
     $$\prob \left[ \substack{\verify(\crs,x',\pi')\ \text{accepts}\\ \ \\ \bigwedge\\ \ \\ \left(\forall i\in [q], \left(x_i,w_i \right) \in \rel(\lang) \right)\\ \ \\ \bigwedge\\ \ \\ \left( \forall i \in [q],\ x' \neq x_i  \right)} \ :\ \substack{\crs \leftarrow \crsgen(1^{\secparam}),\\ \ \\ \left( \{(x_i,w_i)\}_{i \in [q]},\st_{\adversary} \right) \leftarrow \adversary_1(\crs)\\ \ \\  \forall i \in [q], \ \pi_i \leftarrow \prover(\crs,\td,x_i)\\ \ \\ (x',\pi') \leftarrow \adversary_2(\st_{\adversary},\pi_1,\ldots,\pi_q)} \right] = \varepsilon$$
    Then there exists QPT algorithms $\fksetup$ and   $\simr=(\simr_1,\simr_2)$ such that the following holds: 
    $$\left|\prob \left[ \substack{\verify(\crs,x',\pi')\ \text{accepts}\\ \ \\ \bigwedge\\ \ \\ \left(\forall i\in [q], \left(x_i,w_i \right) \in \rel(\lang) \right)\\ \ \\ \bigwedge\\ \ \\ (x',w') \in \rel(L)\\ \ \\ \bigwedge\\ \ \\ \left( \forall i \in [q],\ x' \neq x_i  \right)} \ :\ \substack{(\crs,\td) \leftarrow \fksetup(1^{\secparam}),\\ \ \\ \left( \{(x_i,w_i)\}_{i \in [q]},\st_{\adversary} \right) \leftarrow \adversary_1(\crs)\\ \ \\   (\pi_1,\ldots,\pi_{q},\st_{\simr}) \leftarrow \simr_1 \left( \crs,\td,\{x_i\}_{i \in [q]} \right)\\ \ \\ (x',\pi') \leftarrow \adversary_2(\st_{\adversary},\pi_1,\ldots,\pi_q)\\ \ \\ w' \leftarrow \simr_2(\st_{\simr},x',\pi')} \right]\ -\ \varepsilon \right|\leq \negl(\secparam)$$
\noindent We call a non-interactive argument system satisfying q-simulation-extractability property to be a qseNIZK system. 
\par If q-smulation-extractability property holds against quantum adversaries running in time $2^{\tilde{O}(T)}$ ($\tilde{O}(\cdot)$ notation suppresses additive factors in $O(\log(\secparam))$) then we say that $(\crsgen,\prover,\verify)$ is a $T$-sub-exponential qseNIZK system. 
\end{definition}


\begin{remark}
The definition as stated above is weaker compared to other definitions of simulation-extractability considered in the literature. For instance, we can consider general adversaries who also can obtain simulated proofs for false statements which is disallowed in the above setting. Nonetheless, the definition considered above is sufficient for our application. 
\end{remark}

\paragraph{Instantiation of qseNIZKs.} In the classical setting, simulation-extractable NIZKs can be obtained by generically~\cite{Sah99,DeS+01} combining a traditional NIZK (satisfying completeness, soundness and zero-knowledge) with a public-key encryption scheme satisfying CCA2 security. We observe that the same transformation can be ported to the quantum setting as well, by suitably instantiating the underlying primitives to be quantum-secure. These primitives in turn can be instantiated from QLWE. Thus, we can obtain a q-simulation-extractable NIZK from QLWE. 
\par For our construction of SSL, it turns out that we need a q-simulation-extractable NIZK that is secure against quantum adversaries running in sub-exponential time. Fortunately, we can still adapt the same transformation but instead instantiating the underlying primitives to be sub-exponentially secure. \par Before we formalize this theorem, we first state the necessary preliminary background.

\begin{definition}[q-Non-Interactive Zero-Knowledge]
A non-interactive system $(\crsgen,\prover,\verify)$ defined for a NP language $\lang$ is said to be {\bf q-non-interactive zero-knowledge (qNIZK)} if it satisfies Definition~\ref{def:qnizk:compl} and additionally, satisfies the following properties: 
\begin{itemize}
    \item Adaptive Soundness: For any  malicious QPT prover $\prover^*$, the following holds: 
    $$\prob \left[ \substack{\verify(\crs,x,\pi)\  \text{accepts}\\ \ \\ \bigwedge\\ \ \\ x' \notin \lang} \ \ :\ \substack{\crs \leftarrow \crsgen\left(1^{\secparam} \right)\\ \ \\ (x,\pi) \leftarrow \prover^*(\crs)} \right] \leq \negl(\secparam)$$
    \item Adaptive (Multi-Theorem) Zero-knowledge: For any QPT verifier $\verify^*$, there exists two QPT algorithms $\fkgen$ and simulator $\simr$, such that the following holds: 
    $$\Bigg| \prob\left[ \substack{1 \leftarrow \verify^*\left(\st,\{\pi\}_{i \in [q]} \right)\\ \ \\ \bigwedge\\ \ \\ \forall i\in [q],\ (x_i,w_i) \in \rel(\lang)} \ :\ \substack{\crs \leftarrow \crsgen(1^{\secparam})\\ \ \\ \left(\{(x_i,w_i)\}_{i \in [q]},\st \right) \leftarrow \verify^*(\crs)\\ \ \\ \forall i \in [q],\ \pi_i \leftarrow \prover(\crs,x_i,w_i)} \right]$$
    $$\ \ \ \ \ \ \ \ \ \ \ \ \ \ \ \ \ \ \ - \prob\left[ \substack{1 \leftarrow \verify^*\left(\st,\{\pi\}_{i \in [q]} \right)\\ \ \\ \bigwedge\\ \ \\ \forall i\in [q],\ (x_i,w_i) \in \rel(\lang)} \ :\ \substack{(\crs,\td) \leftarrow \fkgen(1^{\secparam})\\ \ \\ \left(\{(x_i,w_i)\}_{i \in [q]},\st \right) \leftarrow \verify^*(\crs)\\ \ \\ \{\pi_i\}_{i \in [q]} \leftarrow \simr(\crs,\td,\{x_i\}_{i \in [q]})} \right] \Bigg| \leq \negl(\secparam) $$
\end{itemize}
If both adaptive soundness and adaptive multi-theorem zero-knowledge holds against quantum adversaries running in time $2^{\tilde{O}(T)}$ then we say that $(\crsgen,\prover,\verify)$ is a $T$-sub-exponential qNIZK. 
\end{definition}

\begin{remark}
q-simulation-extractable NIZKs imply qNIZKs since simulation-extractability implies both soundness and zero-knowledge properties.  
\end{remark}

\begin{definition}[q-CCA2-secure PKE]
A public-encryption scheme $(\setup,\enc,\dec)$ (defined below) is said to satify {\bf q-CCA2-security} if every QPT adversary $\cA$ wins in $\expt_{\cA}$ (defined below) only with negligible probability. 
\begin{itemize}
    \item $\setup(1^{\secparam})$: On input security parameter $\secparam$, output a public key $\pk$ and a decryption key $\sk$. 
    \item $\enc(\pk,x)$: On input public-key $\pk$, message $x$, output a ciphertext $\ct$. 
    \item $\dec(\sk,\ct)$: On input decryption key $\sk$, ciphertext $\ct$, output $y$. 
\end{itemize}
For any $x \in \{0,1\}^{\poly(\secparam)}$, we have $\dec(\sk,\enc(\pk,x))=x$. \\

\noindent \underline{$\expt_{\cA}(1^{\secparam},b)$}: 
\begin{itemize}
    \item Challenger generates $\setup(1^{\secparam})$ to obtain $(\pk,\sk)$. It sends $\pk$ to $\cA$. 
    \item $\cA$ has (classical) access to a decryption oracle that on input $\ct$, outputs $\dec(\sk,\ct)$. It can make polynomially many queries. 
    \item $\cA$ then submits $(x_0,x_1)$ to the challenger which then returns $\ct^* \leftarrow \enc(\pk,x_b)$.
    \item $\cA$ is then given access to the same oracle as before. The only restriction on $\cA$ is that it cannot query $\ct^*$.
    \item Output $b'$ where the output of $\cA$ is $b'$. 
\end{itemize}
$\cA$ wins in $\expt_{\cA}$ with probability $\mu(\secparam)$ if $\prob\left[ b=b'\ : \substack{b \xleftarrow{\$} \{0,1\}\\ \ \\ \expt_{\cA}(1^{\secparam})} \right]=\frac{1}{2} + \mu(\secparam)$.  
\par If the above q-CCA2 security holds against quantum adveraries running in time $2^{\tilde{O}(T)}$  then we say that $(\setup,\enc,\dec)$ is a  $T$-sub-exponential  q-CCA2-secure PKE scheme. 
\end{definition}

\begin{remark}
One could also consider the setting when the CCA2 adversary has superposition access to the oracle. However, for our construction, it suffices to consider the setting when the adversary only has classical access to the oracle.  
\end{remark}

\noindent Consider the following lemma. 

\begin{lemma}
\label{lem:qsenizks}
Consider a language $\lang_{\ell} \in NP$ such that every $x \in \lang_{\ell}$ is such that $|x|=\ell$.  
\par Under the $\ell$-sub-exponential QLWE assumption, there exists a q-simulation-extractable NIZKs for $\lang_{\ell}$ satisfying perfect completeness. 
\end{lemma}
\begin{proof}
We first state the following proposition that shows how to generically construct a q-simulation-extractable NIZK from qNIZK and a CCA2-secure public-key encryption scheme.

\begin{proposition}
Consider a language $\lang_{\ell} \in NP$ such that every $x \in \lang_{\ell}$ is such that $|x|=\ell$. 
\par Assuming $\ell$-sub-exponential qNIZKs for NP and $\ell$-sub-exponential q-CCA2-secure PKE schemes, there exists a $\ell$-sub-exponential qseNIZK system for $\lang_{\ell}$.  
\end{proposition}
\begin{proof}

Let $\qpke$ be a $\ell$-sub-exponential qCCA2-secure PKE scheme. Let $\qnizk$ be a $\ell$-sub-exponential qNIZK for the following relation. 
$$\rel_{\qnizk}=\left\{ \left((\pk,\ct_w,x),\ (w,r_w) \right)\ :\ \left( (x,w) \in \rel(\lang_{\ell}) \bigwedge \  \ct_w = \enc(\pk,(x,w);r_w) \right)  \right\}$$

\noindent We present the construction (quantum analogue of~\cite{Sah99,DeS+01}) of q-simulation-extractable NIZK for $\lang_{\ell}$ below. 

\begin{itemize}
    \item $\crsgen(1^{\secparam})$: On input security parameter $\secparam$, \begin{itemize}
        \item Compute $\qnizk.\crsgen(1^{\secparam_1})$ to obtain $\qnizk.\crs$, where $\secparam_1=\poly(\secparam,\ell)$ is chosen such that $\qnizk$ is a $\ell$-sub-exponential q-non-interactive zero-knowledge argument system. 
        \item Compute $\pke.\setup(1^{\secparam_2})$ to obtain $(\pk,\sk)$, where $\secparam_2=\poly(\secparam,\ell)$ is chosen such that $\qpke$ is a $\ell$-sub-exponential q-CCA2-secure PKE scheme. 
    \end{itemize}
    Output $\crs=(\pk,\qnizk.\crs)$. 
    \item $\prover(\crs,x,w)$: On input common reference string $\crs$, instance $x$, witness $w$, 
    \begin{itemize}
        \item Parse $\crs$ as $(\pk,\qnizk.\crs)$. 
        \item Compute $\ct_w \leftarrow \pke.\enc(\pk,(x,w);r_w)$, where $r_w \xleftarrow{\$} \{0,1\}^{\poly(\secparam)}$.
        \item Compute $\qnizk.\pi \leftarrow \qnizk.\prover(\qnizk.\crs,(\pk,\ct_w,x),(w,r_w))$. 
    \end{itemize}
    \noindent Output $\pi=\left( \qnizk.\pi,\ct_w \right)$. 
    \item $\verify(\crs,x,\pi)$: On input common reference string $\crs$, NP instance $x$, proof $\pi$, 
    \begin{itemize}
        \item Parse $\crs$ as $(\pk,\ct,\qnizk.\crs)$.
        \item Output $\qnizk.\verify \left( \qnizk.\crs,(\pk,\ct_w,x),\pi \right)$. 
    \end{itemize}
\end{itemize}

\noindent We prove that the above argument system satisfies q-simulation-extractability. We describe the algorithms $\fksetup$ and $\simr=(\simr_1.\simr_2)$ below. Let $\qnizk.\fkgen$ and $\qnizk.\simr$ be the QPT algorithms associated with the zero-knowledge property of $\qnizk$.  \\

\noindent \underline{$\fksetup(1^{\secparam})$}: Compute $(\qnizk.\crs,\tau) \leftarrow \qnizk.\fkgen \left(1^{\secparam} \right)$. Compute $(\pk,\sk) \leftarrow \pke.\setup(1^{\secparam})$. Output $\crs=\left( \qnizk.\crs,\pk,\ct \right)$ and $\td=(\tau,\sk)$. \\

\noindent \underline{$\simr_1 \left( \crs,\td,\{x_i\}_{i \in [q]} \right)$}: Compute $\qnizk.\simr \left( \qnizk.\crs,\tau,(\pk,\ct,x_i) \right)$ to obtain $\qnizk.\pi_i$, for every $i \in [q]$. Output $\left\{\qnizk.\pi_1,\ldots,\qnizk.\pi_q\right\}$ and $\st=\left(\td,\crs,\left( \left\{ x_i \right\}_{i \in [q]} \right) \right)$. \\

\noindent \underline{$\simr_2 \left(\st,x',\pi' \right)$}: On input $\st=\left(\td=(\tau,\sk),\crs,\left( \left\{ x_i \right\}_{i \in [q]} \right) \right)$, instance $x'$, proof $\pi'=(\qnizk.\pi',\ct'_w)$, compute $\dec(\sk,\ct'_{w'})$ to obtain $w'$. Output $w'$.  \\

\noindent Suppose $\cA$ be a quantum adversary running in time $2^{\widetilde{O}(\ell)}$ such that the following holds:  $$\prob \left[ \substack{\verify(\crs,x',\pi')\ \text{accepts}\\ \ \\ \bigwedge\\ \ \\ \left(\forall i\in [q],\ \left(x_i,w_i \right) \in \rel(\lang) \right)\\ \ \\ \bigwedge\\ \ \\ \left( \forall i \in [q],\ x' \neq x_i  \right)} \ :\ \substack{\crs \leftarrow \crsgen(1^{\secparam}),\\ \ \\ \left( \{(x_i,w_i)\}_{i \in [q]},\st_{\adversary} \right) \leftarrow \adversary_1(\crs)\\ \ \\  \forall i \in [q], \ \pi_i \leftarrow \prover(\crs,\td,x_i)\\ \ \\ (x',\pi') \leftarrow \adversary_2(\st_{\adversary},\pi_1,\ldots,\pi_q)} \right] = \varepsilon$$
\noindent Let $\delta$ be such that the following holds: 
 $$\prob \left[ \substack{\verify(\crs,x',\pi')\ \text{accepts}\\ \ \\ \bigwedge\\ \ \\ \left(\forall i\in [q], \left(x_i,w_i \right) \in \rel(\lang) \right)\\ \ \\ \bigwedge\\ \ \\ (x',w') \in \rel(L)\\ \ \\ \bigwedge\\ \ \\ \left( \forall i \in [q],\ x' \neq x_i  \right)} \ :\ \substack{(\crs,\td) \leftarrow \fksetup(1^{\secparam}),\\ \ \\ \left( \{(x_i,w_i)\}_{i \in [q]},\st_{\adversary} \right) \leftarrow \adversary_1(\crs)\\ \ \\   (\pi_1,\ldots,\pi_{q},\st_{\simr}) \leftarrow \simr_1 \left( \crs,\td,\{x_i\}_{i \in [q]} \right)\\ \ \\ (x',\pi') \leftarrow \adversary_2(\st_{\adversary},\pi_1,\ldots,\pi_q)\\ \ \\ w' \leftarrow \simr_2(\st_{\simr},x',\pi')} \right]\  = \delta$$
We prove using a standard hybrid argument that $|\delta-\varepsilon| \leq \negl(\secparam)$. \\

\noindent \underline{$\hybrid_1$}: $\cA$ is given $\pi_1,\ldots,\pi_q$, where $\pi_i \leftarrow \prover(\crs,x_i,w_i)$. Let $(x',\pi')$ is the output of $\cA$ and parse $\pi'=(\qnizk.\pi',\ct'_w)$. Decrypt $\ct'_w$ using $\sk$ to obtain $(x^*,w')$. 
\par From the adaptive soundness of $\qnizk$, the probability that $(x',w') \in \rel(\lang_{\ell})$ and $x^* = x'$ is negligibly close to $\varepsilon$. \\ 

\noindent \underline{$\hybrid_{2}$}: $\cA$ is given $\pi_1,\ldots,\pi_q$, where the proofs are generated as follows: first compute $(\qnizk.\pi_1,\allowbreak \ldots,\allowbreak \qnizk.\pi_{q}) \leftarrow \qnizk.\simr(\crs,\td,\{x_i\}_{i \in [q]})$, where $(\crs,\td) \leftarrow \qnizk.\fkgen(1^{\secparam})$. Then compute $\ct_{w_i} \leftarrow \enc(\pk,(x_i,w_i))$ for every $i \in [q]$. Set $\pi_i=(\qnizk.\pi_i,\ct_{w_i})$. The rest of this hybrid is defined as in $\hybrid_1$. 
\par From the adaptive zero-knowledge property of $\qnizk$, the probability that $(x',w') \in \rel(\lang_{\ell})$ and $x^* = x'$ in the hybrid $\hybrid_{2.j}$ is still negligibly close to $\varepsilon$. \\

\noindent \underline{$\hybrid_3$}: This hybrid is defined similar to the previous hybrid except that $\ct_{w_i} \leftarrow \enc(\pk,0)$, for every $i \in [q]$.  
\par From the previous hybrids, it follows that $\ct'_w \neq \ct_{w_i}$, for all $i\in [q]$ with probability negligibly close to $\varepsilon$; this follows from the fact that $\qpke$ is perfectly correct and the fact that $x^* = x'$ holds with probability negligibly close to $\varepsilon$. Thus, we can invoke  q-CCA2-security of $\pke$, the probability that $(x',w') \in \rel(\lang_{\ell})$ is still negligibly close to $\varepsilon$.  \\

\noindent But note that $\hybrid_3$ corresponds to the simulated experiment and thus we just showed that the probability that we can recover $w'$ such that $(x',w') \in \rel(\lang_{\ell})$ is negligibly close to $\varepsilon$. 

\end{proof}

\noindent The primitives in the above proposition can be instantiated from sub-exponential QLWE by starting with  existing LWE-based constructions of the above primitive and suitably setting the parameters of the underlying LWE assumption. We state the following propositions without proof. 

\begin{proposition}[\cite{PS19}]
Assuming $\ell$-sub-exponential QLWE (Section~\ref{sec:prelims:lwe}), there exists a $\ell$-sub-exponential qNIZK for NP.
\end{proposition}

\begin{remark}
To be precise, the work of~\cite{PS19} constructs a NIZK system satisfying adaptive multi-theorem zero-knowledge and non-adaptive soundness. However, non-adaptive soundness implies adaptive soundness using complexity leveraging; the reduction incurs a security loss of $2^{\ell}$. 
\end{remark}

\begin{proposition}[~\cite{PW11}]
Assuming $\ell$-sub-exponential QLWE (Section~\ref{sec:prelims:lwe}), there exists a $\ell$-sub-exponential q-CCA2-secure PKE scheme. 
\end{proposition}

\end{proof}
}
\newcommand{\key}{\mathrm{sk}}
\newcommand{\num}{q}
\section{Secure Software Leasing (SSL)}
\label{sec:SSL}

\noindent We present the definition of secure software leasing schemes. A secure software leasing (SSL) scheme for a class of circuits $\cktclass=\{\cktclass_{\secparam}\}_{\secparam \in \mathbb{N}}$ consists of the following QPT algorithms. 
\begin{itemize}\setlength\itemsep{1em}
    \item {\bf Private-key Generation,  $\gen(1^{\secparam})$}: On input security parameter $\secparam$, outputs a private key $\key$.  
    \item {\bf Software Lessor, $\lessor\left(\key,C\right)$}: On input the private key $\key$ and a $\poly(n)$-sized classical circuit $C \in \cktclass_{\secparam}$, with input length $n$ and output length $m$, outputs a quantum state $\rho_C$. 
    \item {\bf Evaluation, $\run(\rho_C,x)$}: On input the quantum state $\rho_C$ and an input $x \in \{0,1\}^{n}$, outputs $y$, and some state $\rho'_{C,x}$.
    \item {\bf Check of Returned Software, $\return\left( \key,\rho_C^* \right)$}: On input the private key $\key$ and the state $\rho_C^*$, it checks if $\rho_C^*$ is a valid leased state and if so it outputs 1, else it outputs 0. 
\end{itemize}

\paragraph{Setup.} In this work, we only consider SSL schemes in the setup model. In this model, all the lessors in the world have access to a common reference string generated using a PPT algorithm $\setup$. The difference between $\setup$ and $\gen$ is that $\setup$ is run by a trusted third party whose output is used by all the lessors while $\gen$ is executed by each lessor separately. We note that our impossibility result rules out SSL schemes for all quantum unlearnable class of circuits even in the setup model. 
\par We define this notion below. 

\begin{definition}[SSL with Setup]
A secure software leasing scheme $(\gen, \lessor, \run, \return)$ is said to be in the common reference string (CRS) model if additionally, it has an algorithm $\setup$ that on input $1^{\secparam}$ outputs a string $\crs$. 
\par Moreover, the algorithm $\gen$ now takes as input $\crs$ instead of $1^{\secparam}$ and $\run$ additionally takes as input $\crs$.
\end{definition}

\noindent We require that a SSL scheme, in the setup model, satisfies the following properties. 

\begin{definition}[Correctness]
A SSL scheme $(\setup,\gen,\lessor,\run,\return)$ for $\cktclass=\{\cktclass_{\secparam}\}_{\secparam \in \mathbb{N}}$ is $\varepsilon$-correct if for all $C \in \cktclass_{\secparam}$, with input length $n$, the following two properties holds for some negligible function $\varepsilon$: \begin{itemize} 
\item Correctness of Run: 
\[{\Pr}\left[\forall x \in \{0,1\}^n,\ y=C(x) \ :\ \substack{ \crs \leftarrow \setup\left( 1^{\secparam} \right),\\ \ \\ \key \leftarrow \gen(\crs),\\ \ \\ \rho_C \leftarrow \lessor(\sk,C)\\ \ \\ \left(\rho_{C,x}',y \right) \leftarrow \run \left( \crs,\rho_C,x \right)} \right]\geq 1-\varepsilon
\]
\item Correctness of Check: 
\[{\Pr}\left[\return \left( \key,\rho_C \right)=1\ :\ \substack{ \crs \leftarrow \setup\left( 1^{\secparam} \right),\\ \ \\ \key \leftarrow \gen(\crs) \\ \ \\ \rho_C \leftarrow \lessor(\key,C)} \right]\geq 1-\varepsilon
\]
\end{itemize}  
\end{definition}


\paragraph{Reusability.} A desirable property of a SSL scheme is reusability: the lessee should be able to repeatedly execute $\run$ on multiple inputs. A SSL scheme does not necessarily guarantee reusability; for instance, $\run$ could destroy the state after executing it just once. But fortunately, we can transform this scheme into another scheme that satisfies reusability.   

We define reusability formally. 

\begin{definition}(Reusability)
A SSL scheme $(\setup,\gen,\lessor,\run,\return)$ for $\cktclass=\{\cktclass_{\secparam}\}_{\secparam \in \mathbb{N}}$ is said to be reusable if for all $C\in \cktclass$ and for all $x\in\{0,1\}^n$, $$\trD{\rho_{C,x}'}{\rho_C} \leq \negl(\secparam).$$
\end{definition}

\noindent Note that the above requirement $\trD{\rho_{C,x}'}{\rho_C} \leq \negl(\secparam)$ would guarantee that an evaluator can evaluate the leased state on multiple inputs; on each input, the original leased state is only disturbed a little which means that the resulting state can be reused for evaluation on other inputs.


\par The following proposition states that any SSL scheme can be converted into one that is reusable. 

\begin{proposition}
Let $(\setup,\gen,\lessor,\run,\return)$ be any SSL scheme (not necessarily satisfying the reusability condition). Then, there is a QPT algorithm $\run'$ such that $(\setup,\gen,\lessor,\run',\return)$ is a reusable SSL scheme.
\end{proposition}

\begin{proof}
For any $C\in \cktclass$ and for any $x\in\{0,1\}^n$, we have that $\run(\crs,\rho_C,x)$ outputs $C(x)$ with probability $1-\varepsilon$.  By the Almost As Good As New Lemma (Lemma~\ref{clm:ASGAN}),there is a way to implement $\run$ such that it is possible to obtain $C(x)$, and then recover a state $\widetilde{\rho_C}$ satisfying $\trD{\widetilde{\rho_C}}{\rho_C} \leq \sqrt{\varepsilon}$. We let $\run'$ be this operation.
\end{proof}

\noindent Thus, it suffices to just focus on the correctness property when constructing a SSL scheme.  

\subsection{Security}

Our notion intends to capture the different scenarios discussed in the introduction. In particular, we want to capture the security guarantee that given an authorized (valid) copy $\rho_C$, no pirate can output two authorized copies. We will assume that these valid copies contain a quantum state and a classical string. The $\run$ algorithm expects valid copies to have this form; without loss of generality, the classical part can always be measured before executing $\run$.  

\subsubsection{Finite-Term Lessor Security}
 We require the following security guarantee: suppose a QPT adversary (pirate) receives a leased copy of $C$ generated using $\lessor$; denote this by $\rho_C$. We require that the pirate cannot produce a bipartite state $\sigma^*$ on registers $\Reg_1$ and $\Reg_2$, such that $\sigma_1^*:=\tr_2[\sigma^*]$ passes the verification by $\return$, and the resulting \textit{post-measurement} state on $\Reg_2$, which we denote by $P_2(\sigma^*)$, still computes $C$ by $\run(P_2(\sigma^*),x)=C(x)$.

Before formally stating the definition, let us fix some notation.  We will use the following notation for the state that the pirate keeps after the initial copy has been returned and verified. If the pirate outputs the bipartite state $\sigma^*$, then we will write 
$$P_2(\sk, \sigma^*) \propto \tr_1 \left[ \Pi_{1} [\return(\sk,\cdot)_1 \otimes I_2  \left(\sigma^*\right)]\right] $$
for the state that the pirate keeps \textit{after} the first register has been returned and verified. Here, $\Pi_1$ denotes projecting the output of $\return$ onto $1$, and where $\return(\sk,\cdot)_1\otimes I_2(\sigma^*)$ denotes applying the $\return$ QPT onto the first register, and the identity on the second register of $\sigma^*$. In other words, $P_2(\sk, \sigma^*)$ is used to denote the post-measurement state on $\Reg_2$ conditioned on $\return(\sk,\cdot)$ accepting on $\Reg_1$. 

\begin{definition}[Finite-Term Perfect Lessor Security]
We say that a SSL scheme $(\setup,\gen,\lessor,\run,\allowbreak \return)$ for a class of circuits $\cktclass=\{\cktclass_{\secparam}\}_{\secparam \in \mathbb{N}}$ is said to satisfy {\bf  $(\beta,\gamma,\distrc)$-perfect finite-term lessor security}, with respect to a distribution $\distrc$ on $\cktclass$, if for every QPT adversary $\adversary$ (pirate) that outputs a bipartite (possibly entangled) quantum state on two registers, $\Reg_1$ and $\Reg_2$, the following holds: 
$$\underset{}{\prob}\left[ \substack{  \return\left(\key,\sigma^*_1\right)=1 \\ \ \\ \bigwedge\\ \ \\ \forall x,\ \prob\left[ \run(\crs,P_2(\sk, \sigma^*),x) = C(x) \right] \geq \beta}\ :\ \substack{\crs \leftarrow \setup\left(1^{\secparam} \right),\\ \ \\C \leftarrow \distrc(\secparam),\\ \ \\ \key \leftarrow \gen(\crs),\\ \ \\ \rho_{C} \leftarrow \lessor(\key,C),\\ \ \\ \sigma^* \leftarrow \adversary(\crs,\rho_C) \\ \ \\ \sigma^*_1 = \tr_2[\sigma^*]} \right] \leq \gamma$$

\end{definition}

\begin{remark}
The reason why we use the word perfect here is because we require $\run(P_2(\sigma^*),x)=C(x)$ to hold with probability at least $\beta$ on {\em every} input $x$. Note that $\run$ is not necessarily deterministic (for instance, it could perform measurements) and thus we allow it to output the incorrect value with some probability. 
\end{remark}

\subsection{Infinite-Term Lessor Security}

\newcommand{\dist}{\mathcal{D}}

\noindent In the infinite-term lease case, we want the following security notion: given $(\sigma_1^*, \sigma_2^*)$ generated by a pirate $\cA(\rho_C)$, guarantees that if one copy satisfies the correctness, $$\forall x \Pr[\run(\crs,\sigma_1^*,x)=C(x)] \geq \beta $$
for some non-negligible $\beta$, then after successfully evaluating $C(x)$ using $\sigma_1^*$ on any input $x^*$, it should be the case that the resulting state on the second register, which we will denote by $\cE_{x^*}^{(2)}(\sigma^*)$, cannot also satisfy
$$ \forall x \Pr[\run(\crs,\cE_{x^*}^{(2)}(\sigma^*),x)=C(x)] \geq \beta. $$
In other words, if one of the copies has already been succesful in computing $C$ in $\run$, then there will be inputs in which the second copy cannot evaluate $C$ with better than negligible probability. 
\par This security notion would rule out the following scenario.  Eve gets a copy of $\rho_C$ and gives $\sigma_1^*$ to Alice and $\sigma_2^*$ to Bob. Alice now chooses an input $x_A$, and Bob an input $x_B$. It cannot be the case that for all inputs $(x_A,x_B)$ they choose, they will compute $(C(x_A), C(x_B))$ with non-negligible probability.

\begin{definition}[Infinite-term Perfect Lessor Security]
 We say that a SSL scheme $(\setup, \gen, \allowbreak \lessor, \allowbreak \run, \return)$ for a class of circuits $\cktclass = \{C_\secparam\}_{\secparam\in \mathbb{N}}$ is said to be $(\gamma, \beta, \distrc)$-\textbf{infinite-term perfect lessor secure}, with respect to a distribution $\distrc$, if for every QPT adversary $\adversary$ (pirate) that outputs a bipartite (possibly entangled) quantum state on two registers, $\Reg_1$ and $\Reg_2$, the following holds:
 $$\underset{}{\prob}\left[  \forall x, \left( \substack{\prob\left[ (\run(\crs,x,\sigma_1^*) = C(x)\right] \geq \beta \\ \ \\ \bigwedge \\ \ \\  \forall x', \prob \left[\run(\crs, x', \cE_{x}^{(2)}(\sigma^*)) = C(x')  \right] \geq \beta  } \right)\ :\ \substack{\crs \leftarrow \setup\left(1^{\secparam} \right),\\ \ \\C \leftarrow \distrc(\secparam),\\ \ \\ \key \leftarrow \gen(\crs),\\ \ \\ \rho_{C} \leftarrow \lessor(\key,C),\\ \ \\ \sigma^* \leftarrow \adversary(\crs,\rho_C) \\ \ \\ \sigma_1^* = \tr_2[\sigma^*]} \right] \leq \gamma.$$
\end{definition}

 \begin{remark}
 Both finite and infinite-term security can be extended to the case where the pirate is given multiple copies, $\rho_C^{\otimes m}$, where $\rho_C$ is the output of $\lessor$ on $C$. In the finite-term case, we require the following: if a pirate outputs $m+1$ copies and moreover, the $m$ initial copies are returned and succesfully checked, computing $\run$ on the remaining copy (that the pirate did not return) will not be functionally equivalent to the circuit $C$. In the infinite-term case, the pirate cannot output $m+1$ copies where $\run$ on each of the $m+1$ copies can be used to successfully compute $C$. 
 \end{remark}

\newcommand{\distrf}{\mathcal{D}_{\fclass}}
\newcommand{\newbfckt}{\mathbf{C}}
\newcommand{\newckt}{D}
\section{Impossibility of SSL}\label{sec:imp}
To prove the impossibility of SSL, we first construct {\em de-quantumizable} class of circuits. 

\subsection{De-Quantumizable Circuits: Definition}

A de-quantumizable class of circuits $\cktclass$ is a class of circuits for which there is a QPT algorithm that given any quantum circuit with the same functionality as $C \in \cktclass$, it finds a (possibly different) classical circuit $C'\in \cktclass$ with the same functionality as $C$. Of course if $\cktclass$ is learnable, then it could be possible to just observe the input-output behavior of the quantum circuit to find such a $C'$. To make this notion meaningful, we additionally impose the requirement that $\cktclass$ needs to be quantum unlearnable; given only  oracle access to $C$, any quantum algorithm can find a circuit (possibly a quantum circuit and an auxiliary input state $\rho$) with the same functionality as $C$ with only negligible probability. 
\begin{definition}
We say that a collection of QPT algorithms, $\{U_C,\rho_C\}_{C \in \cktclass}$, computes $\cktclass$ if for any $C\in \cktclass$, with input length $n$ and output length $m$,  $\rho_C$ is a $\poly(n)$-qubits auxiliary state, and $U_C$ a QPT algorithm satisfying that for all $x \in \{0,1\}^n$, $$\prob[U_C(\rho_C,x)=C(x)] \geq 1-\negl(\secparam),$$ where the probability is over the measurement outcomes of $U_C$. We also refer to $(U_C,\rho_C)$ as an efficient quantum implementation of $C$.
A class of classical circuits $\cktclass$, associated with a distribution $\distrc$, is said to be {\em de-quantumizable} if the following holds:
\begin{itemize}
    \item {\bf Efficient de-quantumization}:  There is a QPT algorithm $\cB$ such that, for any $\{U_C,\rho_C\}_{C \in \cktclass}$ that computes $\cktclass$, the following holds:
    
    $$\prob\left[ \substack{C'\in \cktclass \\ \bigwedge \\ \forall x\in \{0,1\}^n, C(x)=C'(x)}\ :\ \substack{C \leftarrow \distrc \\ \ \\ C'(x)\leftarrow \cB(U_C,\rho_C)} \right] \geq 1-\negl(\secparam)$$
    \item {\bf $\nu$-Quantum Unlearnability}: For any QPT adversary $\cA$, the following holds: 
    $$\prob\left[ \forall x, \prob[U^*(\rho^*,x)=C(x)]\geq \nu \ :\ \substack{C \leftarrow \distrc \\ (U^*,\rho^*) \leftarrow \cA^{C(\cdot)}(1^\secparam)} \right] \leq \negl(\secparam)$$
\end{itemize}
\end{definition}

\begin{remark}\label{rmk:reuse} By the Almost As Good As New Lemma (Lemma~\ref{clm:ASGAN}),
we can assume that the QPT algorithm $U_C$ also output a state $\rho'_{C,x}$ that is negligibly close in trace distance to $\rho_C$, i.e. for all $C\in \cktclass$ and $x \in \{0,1\}^n$ it holds that
$$\prob[U_C(\rho_C,x)=(\rho_{C,x}',C(x))] \geq 1-\negl(\secparam)$$
and $\trD{\rho_{C,x}'}{\rho_C}\leq \negl(\secparam)$.
\end{remark}

\begin{remark}
We emphasize that the efficient de-quantumization property requires that the circuit $C'$ output by the adversary should be in the same circuit class $\cktclass$.  
\end{remark}

\begin{remark}
We can relax the unlearnability condition in the above definition to instead have a distribution over the inputs and have the guarantee that the adversary has to output a circuit $(U^*,\rho^*)$ such that it agrees with $C$ only on inputs drawn from this distribution. Our impossibility result will also rule out this relaxed unlearnability condition; however, for simplicity of exposition, we consider the unlearnability condition stated in the above definition. 
\end{remark}


From the above definition, we can see why a de-quantumizable class $\cktclass$ cannot be copy-protected, as there is a QPT $\cB$ that takes any $(U_C,\rho_C)$ efficiently computing $C$, and outputs a functionally equivalent {\em classical} circuit $C'$, which can be copied. In the following theorem we will show that if every circuit $C \in \cktclass$ have a unique representation in $\cktclass$, then it is also not possible to have SSL for this circuit class. To see why we need an additional condition, lets consider a QPT pirate $\cA$ that wants to break SSL given $(\run,\rho_C)$ computing $C \in \cktclass$. Then, $\cA$ can run $\cB$ to obtain a circuit $C' \in \cktclass$, but in the proccess it could have destroyed $\rho_C$, hence it wouldn't be able to return the initial copy. If $\cB$ takes as input $(\run,\rho_C)$ and outputs a {\em fixed} $C'$ with probability neglibly close to $1$, then by the Almost As Good As New Lemma, it could uncompute and recover $\rho_C$. The definition of de-quantumizable class does not guarantee that $\cB$ will output a fixed circuit $C'$, unless each circuit in the family has a unique representation in $\cktclass$.  If each circuit has a unique representation, the pirate would obtain $C'=C$ with probability neglibly close to $1$, and uncompute to recover $\rho_C$. At this point, the pirate can generate its own leasing keys $\sk'$, and run $\lessor(\sk',C')$ to obtain a valid leased state $\rho'_{C'}$. The pirate was able to generate a new valid leased state for $C$, while preserving the initial copy $\rho_C$, which it can later return to the lessor.

\begin{theorem}\label{thm:formal}
Let $(\cktclass,\distrc)$ be a de-quantumizable class of circuits in which every circuit in the support of $\distrc$ has a unique representation in $\cktclass$. Then there is no SSL scheme $(\setup,\gen,\allowbreak \lessor,\allowbreak \run,\return)$ (in CRS model) for $\cktclass$ satisfying $\varepsilon$-correctness and $( \beta,\gamma,\distrc)$-perfect finite-term lessor security for any negligible $\gamma$, and any $\beta \leq (1-\varepsilon)$. 
\end{theorem}

\begin{proof}

Consider the QPT algorithm $\cA$ (pirate) that is given $\rho_C \leftarrow \lessor(\key,C)$ for some $C \leftarrow \distrc$. The pirate will run $\cB$, the QPT that de-quantumizes $(\cktclass,\distrc)$, on input $(\run,\rho_C)$ to obtain a functionally equivalent circuit $C'\in \cktclass$. Because $C$ has a unique representation in $\cktclass$, we have $C'=C$. Since this succeeds with probability neglibly close to $1$, by the Almost As Good As New Lemma~\ref{clm:ASGAN}, it can all be done in a way such that it is possible to obtain $C$ and to recover a state $\widetilde{\rho_C}$ satisfying $\trD{\widetilde{\rho_C}}{\rho_C} \leq \negl(\secparam)$. At this point, the pirate generates its own key $\key'\leftarrow \gen(\crs)$, and prepares $\rho'_{C} \leftarrow \lessor(\key',C)$. It outputs $\widetilde{\rho_C} \otimes  \rho'_{C}$.

 This means that $\rho'_{C}$ is a valid leased state and by correctness of the SSL scheme, 
\[{\Pr}\left[\forall x \in \{0,1\}^n,\  \run \left( \crs,\rho'_{C},x \right)=C(x)\ :\ \substack{ \crs \leftarrow \setup\left( 1^{\secparam} \right),\\ \ \\ \key' \leftarrow \gen(\crs),\\ \ \\ \rho'_{C} \leftarrow \lessor(\sk',C)} \right]\geq 1-\varepsilon
\]

\noindent Furthermore, since $\trD{\widetilde{\rho_C}}{\rho_C} \leq \negl(\secparam)$, the probability that $\widetilde{\rho_C}$ passes the return check is neglibly close to $1$. Putting these together, we have

$$\prob\left[ \substack{  \return\left(\key,\widetilde{\rho_C}\right)=1 \\ \ \\ \bigwedge\\ \ \\ \forall x,\ \prob\left[ \run(\crs,\rho'_{C},x) = C(x) \right] \geq 1-\varepsilon}\ :\ \substack{\crs \leftarrow \setup\left(1^{\secparam} \right),\\ \ \\C \leftarrow \distrc(\secparam),\\ \ \\ \key \leftarrow \gen(\crs),\\ \ \\ \rho_{C} \leftarrow \lessor(\key,C),\\ \ \\ \widetilde{\rho_C}\otimes \rho'_{C} \leftarrow \adversary(\crs,\rho_C)} \right] \geq 1-\negl(\secparam) $$
\end{proof}

\subsection{De-quantumizable Circuit Class: Construction}
\noindent All that remains in the proof of impossibility of SSL is the construction of a de-quantumizable circuits class $(\cktclass, \distrc)$ in which every circuit in the support of $\distrc$ has a unique representation in $\cktclass$. We begin with an overview of the construction.

\paragraph{Constructing de-quantumizable Circuits: Challenges.}  

\noindent The starting point is the seminal work of Barak et al. \cite{BGIRSVY01}, who demonstrated a class of functions, where each function is associated with a secret key $k$, such that: (a) {\em Non-black-box secret extraction}: given non-black-box access to any classical circuit implementation of this function, the key can be efficiently recovered, (b) {\em Classical Unlearnability of secrets}: but given black-box access to this circuit, any classical adversary who can only make polynomially many queries to the oracle cannot recover the key.  
\par While the result of Barak et al. has the ingredients suitable for us, it falls short in many respects:
\begin{itemize}
    \item The proof of non-black-box secret extraction crucially relies upon the fact that we are only given a classical obfuscated circuit. In fact there are inherent difficulties that we face in adapting Barak et al. to the quantum setting; see~\cite{alagic2016quantum}. 
    
    \item As is the case with many black-box extraction techniques, the proof of Barak et al. involves evaluating the obfuscated circuit multiple times in order to recover the secret. As is typically the case with quantum settings, evaluating the same circuit again and again is not always easy -- the reason being that evaluating a circuit once could potentially destroy the state thus rendering it impossible to run it again. 
    \item Barak et al. only guarantees extraction of secrets given black-box access to the classical circuit implementation of the function. However, our requirement is qualitatively different: given a quantum implementation of the classical circuit, we need to find a (possible different) classical circuit with the same functionality.
    \item Barak et al.'s unlearnability result only ruled out adversaries who make classical queries to the oracle. On the other hand, we need to argue unlearnability against QPT adversaries who can perform superposition queries to the oracle. 
\end{itemize}
\noindent Nonetheless, we show that the techniques introduced in a simplified version of Barak\footnote{See~\cite{BP13} for a description of this simplified version.} can be suitably adapted for our purpose by using two tools: quantum fully homomorphic encryption (QFHE) and lockable obfuscation. Combining QFHE and lockable obfuscation for the purpose of secret extraction has been recently used in a completely different context, that of building zero-knowledge protocols ~\cite{BS19,AL19} (and in classical setting was first studied by~\cite{BKP19}).

\paragraph{Construction.} We  present the construction of de-quantumizable circuits.

\begin{theorem}\label{thm:impossibility} Assuming the quantum hardness of learning with errors (QLWE), and assuming that there is a $\qfhe$ that supports evaluation of arbitrary polynomial-sized quantum circuits (see \ref{sec:prelims:qfhe}), and has the following two properties: (a)   ciphertexts have classical plaintexts have classical descriptions and, (b) classical ciphertexts can be decrypted using a classical circuit,
\par there exists a de-quantumizable class of circuits $(\cktclass, \distrc)$. 
\end{theorem}
\begin{proof}
\noindent We define a de-quantumizable class of circuits $\cktclass=\{\cktclass_{\secparam}\}_{\secparam \in \mathbb{N}}$, where every circuit in $\cktclass_{\secparam}$ is defined as follows: \\ 

\noindent {$\underline{C_{a,b,r,\pk,\cO}(x)}$:}
\begin{enumerate}
    \item If $x = 0\cdots0$, output $\qfhe.\enc\left(\pk,a;r\right)|\cO|\pk$.
    \item Else if $x=a$, output $b$.
    \item Otherwise, output $0\cdots0$
\end{enumerate}
\noindent We will suitably pad with zeroes such that all the inputs (resp., outputs) are of the same length $n$ (resp., of the same length $m$). 
\par Let $\distrc(\secparam)$ be the distribution that outputs a circuit from $\cktclass_{\secparam}$ by sampling $a,b,r \xleftarrow{\$} \{0,1\}^{\secparam}$, then computing $(\pk,\sk)\leftarrow \qfhe.\gen(1^\secparam)$, and finally computing an obfuscation $\cO \leftarrow \lobf(\newbfckt[\allowbreak \qfhe.\dec(\sk,\cdot),b,\allowbreak (\sk|r)])$, where $\newbfckt$ is a compute-and-compare circuit. 


We show that with respect to this distribution: (a) $\cktclass$ is quantum unlearnable (Proposition~\ref{prop:quan:unlearn}) and, (b) $\cktclass$ is efficiently de-quantumizable (Proposition~\ref{prop:eff:dequant}). 

\begin{proposition}
\label{prop:quan:unlearn} For any non-negligible $\nu$, the circuit class 
$\cktclass$ is $\nu$-quantum unlearnable with respect to $\distrc$.
\end{proposition}


\submversion{
We provide a proof of the above proposition in Section~\ref{sec:unlearnable}. 
}

\fullversion{
\begin{proof}

\noindent We first rule out QPT adversaries, who given black-box access to the circuit, can find the secret key $\sk$ with non-negligible probability. Once we rule out this type of adversaries, we then show how to reduce a QPT adversary who breaks the quantum unlearnability property of the de-quantumizable class of circuits to one who finds the secret key $\sk$; thus completing the proof. 
\begin{claim} \label{clm:unlearn}
For all QPT $\cA$ with oracle access to $C_{a,b,r,\pk,\cO}(\cdot)$ (where the adversary is allowed to make superposition queries), we have
$$\underset{(a,b,r,\pk,\cO) \leftarrow \distrc}{\prob}\left[\sk \leftarrow \cA^{C_{a,b,r,\pk,\cO}} \left(1^\secparam \right) \right] \leq \negl(\secparam)$$
\end{claim}
\begin{proof} 
\noindent Towards proving this, we make some simplifying assumptions; this is only for simplicity of exposition and they are without loss of generality.\\

\noindent {\em Simplifying Assumptions.} Consider the following oracle $O_{a,b,r,\pk,\cO}$:
$$ O_{a,b,r,\pk,\cO} \ket{x}\ket{z} = \left\{  \begin{array}{lcl} \ket{x}\ket{z \oplus C_{a,b,r,\pk,\cO}(x)}, & \text{if }x \neq 0 \cdots0\\
& \\
\ket{x}\ket{z}, & \text{if }x=0\cdots0
\end{array} \right. $$
The first simplifying assumption is that the adversary $\cA$ is given access to the oracle $O_{a,b,r,\pk,\cO}$, instead of the oracle $C_{a,b,r,\pk,\cO}$. In addition, $\cA$ is given $\enc(\pk,a;r)$, $\pk$, and $\cO$ as auxiliary input. 
\par The second simplifying assumption is that $\cA$ is given some auxiliary state $\ket{\xi}$, and that it only performs computational basis measurements right before outputting (i.e. $\cA$ works with purified states). \\

\noindent {\em Overview.} Our proof follows the adversary method proof technique~\cite{ambainis2002quantum}. We prove this by induction on the number of queries. We show that after every query the following invariant is maintained: the state of the adversary has little amplitude over $a$. More precisely, we argue that the state of the adversary after the $t^{th}$ query, is neglibly close to the state just before the $t^{th}$ query, denoted by $\ket{\psi^t}$. After the adversary obtains the response to the $t^{th}$ query, it then applies a unitary operation to obtain the state $\ket{\psi^{t+1}}$, which is the state of the adversary just before the $(t+1)^{th}$ query. This observation implies that there is another state $\ket{\phi^{t+1}}$ that: (a) is close to $\ket{\psi^{t+1}}$ (here, we use the inductive hypothesis that $\ket{\phi^{t}}$ is close to $\ket{\psi^{t}}$) and, (b) can be prepared without querying the oracle at all. \\

\noindent Let $U_i$ denote the unitary that $\cA$ performs right before its $i^{th}$ query, and let $\bfA,\sfX,$ and $\sfY$ denote the private, oracle input, and oracle output registers of $\cA$, respectively. 
  
  Just before the $t^{th}$ query, we denote the state of the adversary to be:

$$\ket{\psi^t}:=U_t O \cdots O U_1 \ket{\psi^0}$$
where $\ket{\psi^0} = \ket{\xi}\ket{\enc(\pk,a;r),\cO,\pk}\ket{0\cdots0}_{\sfX}\ket{0\cdots0}_{\sfY}$ is the initial state of the adversary. Let $\Pi_a = (\ket{a}\bra{a})_\sfX \otimes I_{\sfY,\sfA}$.\\ 

 Note that any $\cA$ that outputs $\sk$ with non-negligible probability can also query the oracle on a state $\ket{\psi}$ satisfying $\tr[\Pi_{a}\ket{\psi}\bra{\psi}] \geq \nonnegl(\secparam)$ with non-negligible probability. Since $\cA$ outputs $\sk$ with non-negligible probability, it can decrypt $\enc(\pk,a;r)$, to find $a$ and then query the oracle on $a$. In other words, if there is an adversary $\cA$ that finds $\sk$ with non-negligible probability, then there is an adversary that at some point queries the oracle with a state $\ket{\psi}$ $\tr[\Pi_{a}\ket{\psi}\bra{\psi}] \geq \nonnegl(\secparam)$ also with non-negligible probability. \\

\noindent Hence, it suffices to show that for any adversary $\cA$ that makes at most $T=\poly(\secparam)$ queries to the oracle, it holds that
$$\prob[\forall j,  \tr[\Pi_{a} \ket{\psi^j}\bra{\psi^j}] \leq \negl(\secparam)] \geq 1-\negl(\secparam).$$

\noindent This would then imply that $\cA$ cannot output $\sk$ with non-negligible probability, thus proving Claim~\ref{clm:unlearn}.

Towards proving the above statement, consider the following claim that states that if $\cA$ has not queried  the oracle with a state that has large overlap with $\Pi_{a}$, then its next query will also not have large overlap with $\Pi_{a}$. 
\begin{claim}[No Good Progress] Let $T$ be any polynomial in $\secparam$.
Suppose for all $t<T$, the following holds:
$$\tr \left[ \Pi_{a} \ket{\psi^t}\bra{\psi^t} \right] \leq \negl(\secparam) $$
\noindent Then, $\prob[\tr[\Pi_{a}\ket{\psi^T}\bra{\psi^T}]\leq \negl(\secparam)] \geq 1-\negl(\secparam)$.
\end{claim}
\begin{proof}
For all $j$, let $\ket{\phi^j}=U_jU_{j-1}...U_1\ket{\psi^0}$. 

We will proceed by induction on $T$. Our base case is $T=1$ (just before the first query to the oracle); that is, $\ket{\psi^1}=\ket{\phi^1}$. Suppose the following holds: $$\prob[\tr[\Pi_{a}\ket{\psi^1}\bra{\psi^1}]\geq \nonnegl(\secparam)]\geq \nonnegl(\secparam).$$

The first step is to argue that if $\cA$ can prepare a state such that $\tr[\Pi_a \ket{\psi_1}\bra{\psi_1}] \geq \nonnegl(\secparam)$ given $\enc(\pk,a;r), \pk$ and $\cO \leftarrow \lobf(\newbfckt \left[ \qfhe.\dec(\sk,\cdot),b,(\sk|r) \right])$ without querying the oracle, then it can also prepare a state with large overlap with $\Pi_a$ if its given the simulator of the lockable obfuscation instead. We will use $\cA$ (specifically, the first unitary that $\cA$ applies, $U_1$) to construct an adversary $\cB$ that breaks the security of lockable obfuscation. $\cB$ is given $a$, $\enc(\pk,a;r)$, $\pk$ and $\cO$ as well as auxiliary state $\ket{\xi}$.  It the prepares $\ket{\psi_{1,\cO}}= U_1 \ket{\xi}\ket{\enc(\pk,a;r),\cO,\pk}\allowbreak \ket{0\cdots0}_{\sfX}\ket{0\cdots0}_{\sfY}$, and measures in computational basis. If the output of this measurement is $a$, it outputs $1$; otherwise, it outputs $0$.

Consider the following hybrids.\\

$\hyb_1$ In this hybrid, $\cB$ is given $a$, $\enc(\pk,a;r),\pk, \cO\leftarrow \lobf(\newbfckt[\qfhe.\dec(\sk,\cdot),b,(\sk|r)])$. \\

$\hyb_2$: In this hybrid, $\cB$ is given $a$, $\enc(\pk,a;r),\pk$ and $\cO \leftarrow \simulator(1^\secparam)$.\\

Since the lock $b$ is chosen uniformly at random, by security of lockable obfuscation, the probability that $\cB$ outputs $1$ in the first hybrid is negligibly close to the probability that $\cB$ outputs $1$ in the second hybrid. This means that if $\tr[\Pi_a \ket{\psi_{1,\cO}}\bra{\psi_{1,\cO}}] \geq \nonnegl(\secparam)$ with non-negligible probability when $\cO \leftarrow \lobf(\newbfckt[\qfhe.\dec(\sk,\cdot),b,(\sk|r)])$, then this still holds when $\cO \leftarrow \simulator(1^\secparam)$. \par But we show that if $\tr[\Pi_a \ket{\psi_{1,\cO}}\bra{\psi_{1,\cO}}] \geq \nonnegl(\secparam)$, when $\cO$ is generated as $\cO \leftarrow \simulator(1^\secparam)$, then QFHE is insecure. 
\begin{itemize}
    \item Consider the following QFHE adversary who is given $\ket{\xi}$ as auxiliary information, and chooses two messages $m_0=0\cdots0$ and $m_1=a$, where $a$ is sampled uniformly at random from $\{0,1\}^{\secparam}$. It sends $(m_0,m_1)$ to the challenger. 
    \item The challenger of QFHE then generates $\ct_d = \enc(\pk,m_d)$, for some bit $d \in \{0,1\}$ and sends it to the QFHE adversary.
    \item The QFHE adversary computes $\cO \leftarrow \simulator(1^\secparam)$.
    \item The QFHE adversary then prepares the state $\ket{\psi_d} = U_1 \left( \ket{\xi}\ket{\ct_d, \cO, \pk}\ket{0\cdots0}_{\sfX}\ket{0\cdots0}_{\sfY}\right)$ and measures register $\sfX$ in the computational basis.
\end{itemize}

\noindent If $d=0$, the probability that the QFHE adversary obtains $a$ as outcome is negligible; since $a$ is independent of $U_1$, $\pk$, $\ket{\xi}$, and $\cO$. But from our hypothesis ($\prob[\tr[\Pi_{a}\ket{\psi^1}\bra{\psi^1}]\geq \nonnegl(\secparam)]\geq \nonnegl(\secparam)$), the probability that the QFHE adversary obtains $a$ as outcome is non-negligible for the case when $d=1$. This contradicts the security of QFHE as the adversary can use $a$ to distinguish between these two cases. 

\noindent To prove the induction hypothesis, suppose that for all $t<T$, the following two conditions hold:
\begin{enumerate}
    \item $\tr[\Pi_{a} \ket{\psi^t}\bra{\psi^t}] \leq \negl(\secparam)$
    \item $|\braket{\phi^t}{\psi^t}|= 1-\delta_t$ 
\end{enumerate}
for some negligible $\delta_1,...,\delta_{T-1}$. We can write
$$|\braket{\phi^T}{\psi^T}| = |\bra{\phi^{T-1}}O\ket{\psi^{T-1}}|$$

By hypothesis (2) above, we have $\ket{\phi^{T-1}}=(1-\delta_{T-1}) e^{i \alpha}\ket{\psi^{T-1}}+\sqrt{2\delta_{T-1}-\delta_{T-1}^2}\ket{\widetilde{\psi}^{T-1}}$, here $\alpha$ is some phase, and $\ket{\widetilde{\psi}^{T-1}}$ is some state orthogonal to $\ket{\psi^{T-1}}$. Then
\begin{align*}
|\braket{\phi^T}{\psi^T}| &= |(1-\delta_{T-1}) e^{i \alpha}\bra{\psi^{T-1}}O\ket{\psi^{T-1}} + \sqrt{2\delta_{T-1}-\delta_{T-1}^2} \bra{\widetilde{\psi}^{T-1}}O\ket{\psi^{T-1}}| \\
&\geq |(1-\delta_{T-1})e^{i\alpha} \bra{\psi^{T-1}}O\ket{\psi^{T-1}}| - \sqrt{2\delta_{T-1}-\delta_{T-1}^2}\\
&\geq (1-\delta_{T-1}) |\bra{\psi^{T-1}}O\ket{\psi^{T-1}}| - \sqrt{2\delta_{T-1}-\delta_{T-1}^2}
\end{align*}

\noindent By hypothesis (1) above, and since the oracle acts non-trivially only on $a$, we have 
$|\bra{\psi^{T-1}}O\ket{\psi^{T-1}}| \geq 1-\negl(\secparam)$, which gives us
$$|\braket{\phi^T}{\psi^T}| \geq 1-\negl(\secparam). $$

Now we want to show that $\tr[\Pi_{a} \ket{\psi^{T}}\bra{\psi^{T}}] \leq \negl(\secparam)$. This follows from the security of lockable obfuscation and QFHE similarly to $T=1$ case.  Since $|\braket{\phi^T}{\psi^T}|\geq 1-\negl(\secparam)$, we have that 
$$\tr[\Pi_{a} \ket{\phi^{T}}\bra{\phi^{T}}] \leq \negl(\secparam) \implies \tr[\Pi_{a} \ket{\psi^{T}}\bra{\psi^{T}}] \leq \negl(\secparam).$$ 

\noindent From a similar argument to the $T=1$ case but using $U_T U_{T-1} \cdots U_1$ instead of just $U_1$, we have that  $\prob[\tr[\Pi_{a} \ket{\phi^{T}}\bra{\phi^{T}}]\leq \negl(\secparam)]\geq 1-\negl(\secparam)$.
\end{proof}


\noindent Let $E_i$ denote the event that $\tr[\Pi_{a} \ket{\psi^i}\bra{\psi^i}] \leq \negl(\secparam)$. Let $p_T$ be the probability that $\tr[\Pi_{a} \ket{\psi^t}\bra{\psi^t}] \leq \negl(\secparam)$ for all the queries $t\leq T$. Using the previous claim, we have that

\begin{align*}
    p_T &= \overset{T}{\underset{t=1}{\prod}}\prob[E_t|\forall j<t, E_j] \\
    &\geq (1-\negl(\secparam))^T \\
    &\geq (1-T \cdot  \negl(\secparam))
\end{align*}


\end{proof}

Suppose that there is a QPT $\cB$ that can learn $\cktclass$ with respect to $\distrc$ with non-negligible probability $\delta$. In other words, for all inputs $x$, 

$$\prob \left[ U(\rho,x)=C_{a,b,r,\pk,\cO}(x) : \substack{C_{a,b,r,\pk,\cO}\leftarrow \distrc \\ (U,\rho) \leftarrow \cB^{C_{a,b,r,\pk,\cO}}(1^\secparam)} \right] = \delta $$

\noindent We use $\cB^{C_{a,b,r,\pk,\cO}}$ to construct a QPT $\cA^{C_{a,b,r,\pk,\cO}}$ that can find $\sk$ with probability neglibly close to $\delta$, contradicting Claim~\ref{clm:unlearn}.  To do this, $\cA$ first prepares $(U,\rho) \leftarrow \cB^{C_{a,b,r,\pk,\cO}}(1^\secparam)$. Then, $\cA^{C_{a,b,r,\pk,\cO}}$ queries the oracle on input $0\cdots0$, obtaining $\ct_1 = \qfhe.\enc(\pk,a;r)$ along with $\pk$ and $\cO=\lobf(\newbfckt[\qfhe.\dec(\sk,\cdot),b,(\sk|r)])$. Finally, it homomorphically computes $\ct_2 \leftarrow \qfhe.\eval(U(\rho,\cdot), \ct_1)$. Then it computes $\sk'|r' = \cO(\ct_2)$, and outputs $\sk'$.
\par By the correctness of the $\qfhe$ and because $U(\rho,a)=b$ holds with probability $\delta$, we have that  $\qfhe.\dec_{\sk}(\ct_2)=b$ with probability negligibly close to $\delta$. By correctness of lockable obfuscation $\cO(\ct_2)$ will output the right message $\sk$. This means that output of $\cA$ is $\sk$ with probability negligibly close to $\delta$.\\

\end{proof}
}

\begin{proposition}
\label{prop:eff:dequant} $(\cktclass, \distrc)$ is efficiently de-quantumizable.
\end{proposition}

\begin{proof}
\noindent We will start with an overview of the proof. \\

\noindent \textit{Overview}: Given a quantum circuit $(U_C, \rho_C)$ that computes $C_{a,b,r,\pk,\cO}(\cdot)$, first compute on the input $x=0 \cdots 0$ to obtain $\qfhe.\enc(\pk,a;r)|\cO|\pk$. We then homomorphically evaluate the quantum circuit on $\qfhe.\enc(\pk,a;r)$ to obtain $\qfhe.\enc(\pk,b')$, where $b'$ is the output of the quantum circuit on input $a$; this is part where we crucially use the fact that we are given $(U_C, \rho_C)$ and not just black-box access to the functionality computing $(U_C, \rho_C)$. But  $b'$ is nothing but $b$! Given QFHE encryption of $b$, we can then use the lockable obfuscation to recover $\sk$; since the lockable obfuscation on input a valid encryption of $b$ outputs $\sk$. Using $\sk$ we can then recover the original circuit $C_{a,b,r,\pk,\cO}(\cdot)$. Formal details follow. \\

\noindent For any $C \in \cktclass$, let $(U_C, \rho_C)$ be any QPT algorithm (with auxiliary state $\rho_C$) satisfying that for all $x \in \{0,1\}^n$, $$\prob\left[ U_C(\rho_C,x)=\left(\rho_{C,x}',C(x)\right) \right]\geq 1-\negl(\secparam),$$
where the probability is over the measurement outcomes of $U_C$, and $\rho_{C,x}'$ is neglibly close in trace distance to $\rho_C$ (see Remark~\ref{rmk:reuse}). We will show how to constuct a QPT $\cB$ to de-quantumize $(\cktclass,\distrc)$.

$\cB$ will perform a QFHE evaluation, which we describe here. Given $\qfhe.\enc(\pk,x)$, we want to homomorphically evaluate $C(x)$ to obtain $\qfhe.\enc(\pk,C(x))$. To do this, first prepare $\qfhe.\enc(\pk,\rho_C,x)$, then evaluate $U_C$ homomorphically to obtain the following:
$$\qfhe.\enc(\pk,\rho_{C,x}', C(x)) = \qfhe.\enc(\pk,\rho_{C,x}')\big|\qfhe.\enc(\pk,C(x)) $$  

Consider the following QPT algorithm $\cB$ that is given $(U_C,\rho_C)$ for any $C \in \cktclass$. 

\paragraph{$\cB(U_C, \rho_C)$:}
\begin{enumerate}
    \item Compute $(\rho', \ct_1|\cO'|\pk') \leftarrow U_C(\rho_C,0\cdots0)$. 
    \item Compute $\sigma|\ct_2 \leftarrow \qfhe.\eval( U_C(\rho',\cdot),\ct_1)$
    \item Compute $\sk'|r' \leftarrow \cO(\ct_2)$
    \item Compute $a'\leftarrow \qfhe.\dec(\sk',\ct_1)$, $b'\leftarrow \qfhe.\dec(\sk',\ct_2)$.
    \item Output $C_{a',b',r',\pk',\cO'}$.
\end{enumerate}

\noindent We claim that with probability negligibly close to $1$, $(a',b',r',\pk',\cO')=(a,b,r,\pk,\cO)$ when $C:=C_{a,b,r,\pk,\cO}\leftarrow \distrc$. This would finish our proof.

Lets analyze the outputs of $\cB$ step-by-step.
\begin{itemize} 
\item After Step (1), with probability neglibibly close to $1$, we have that $\ct_1 = \qfhe.\enc(\pk,a;r)$ , $\pk'=\pk$, and $\cO' = \cO \leftarrow \lobf(\newbfckt[\qfhe.\dec(\sk,\cdot),b,(\sk|r)])$. Furthermore, we have that $\rho'$ is negligibly close in trace distance to $\rho_C$.
\item Conditioned on Step (1) computing $C(0\cdots0)$ correctly, we have that $\qfhe.\eval(\allowbreak U_C(\rho',.), \ct_1)$ computes correctly with probability negligibly close to $1$. This is because $\trD{\rho'}{\rho_C} \leq \negl(\secparam)$, and by correctness of both QFHE and $(U_C,\rho_C)$. Conditioned on $\ct_1 = \qfhe.\enc(\pk,a;r)$,  when Step (2) evaluates correctly, we have $\ct_2 = \qfhe.\enc(\pk,C(a))=\qfhe.\enc(\pk,b)$ 

\item Conditioned on $\ct_2 = \qfhe.\enc(\pk,b)$, by correctness of lockable obfuscation, we have that $\cO(\ct_2)$ outputs $\sk|r$. Furthermore, by correctness of QFHE, decryption is correct: $\qfhe.\dec(\sk,\ct_1)$ outputs $a$ with probability neglibly close to $1$, and $\qfhe.\dec(\sk,\ct_2)$ outputs $b$ with probability neglibly close to $1$.
\end{itemize}
With probability negligibly close to $1$, we have shown that $(a',b',r',\pk',\cO')=(a,b,r,\pk, \cO)$.

Note that it is also possible to recover $\rho''$ that is neglibly close in trace distance to $\rho_C$. This is because $\sigma = \qfhe.\enc(\pk,\rho'')$ for some $\rho''$ satisfying $\trD{\rho''}{\rho_C}$.  Once $\sk'=\sk$ has been recovered, it is possible to also decrypt $\sigma$ and obtain $\rho''$. To summarize, we have shown a QPT $\cB$ satisfying

$$\prob[\cB(U_C,\rho_C)= (\rho'',C) \ : \ C \leftarrow \distrc] \geq 1-\negl(\secparam) $$

where $\trD{\rho''}{\rho_C} \leq \negl(\secparam)$.
\end{proof}

\end{proof}

\paragraph{Implications to Copy-Protection.} We have constructed a class $\cktclass$ and an associated distribution $\distrc$ that is efficient de-quantumizable. In particular, this means that there is no copy-protection for $\cktclass$.  If for all inputs $x$, there is a QPT $(U_C, \rho_C)$ to compute $U_C(\rho_C,x)=C(x)$ with probability $1-\varepsilon$ for some negligible $\varepsilon$, then  it is possible to find, with probability close to $1$, a circuit $C'$ that computes the same functionality as $C$. We also proved that $(\cktclass, \distrc)$ is quantum unlearnable. We summarize the result in the following corollary,

\begin{corollary}
There is $(\cktclass, \distrc)$ that is quantum unlearnable, but $\cktclass$ cannot be copy-protected against $\distrc$. Specifically, for any $C \leftarrow \distrc$ with input length $n$, and for any QPT algorithm $(U_C,\rho_C)$ satisfying that for all $x \in \{0,1\}^n$, $$\prob[U_C(\rho_C,x)=C(x)] \geq 1-\varepsilon $$
for some negligible $\varepsilon$, there is a QPT algorithm (pirate) that outputs a circuit $C'$, satisfying $C'(x)=C(x)$ for all $x \in \{0,1\}^n$, with probability negligibly close to $1$. 
\end{corollary}

\paragraph{Further Discussion.}  Notice that in our proof that $\cktclass$ is efficient de-quantumizable, we just need to compute $U_C(\rho_C, x)$ at two different points $x_1 = 0\cdots0$ and $x_2=a$, where the evaluation at $x_2$ is done homomorphically. This means that any scheme that lets a user evaluate a circuit $C$ at least 2 times (for 2 possibly different inputs) with non-negligible probability cannot be copy-protected. Such a user would be able to find all the parameters of the circuit, $(a,b,r,\pk,\cO)$, succesfully with non-negligible probability, hence it can prepare as many copies of a functionally equivalent circuit $C'$. 

In our proof, we make use of the fact that $(U_C,\rho_C)$ evaluates correctly with probability close to $1$. This is in order to ensure that the pirate can indeed evaluate at $2$ points by uncomputing after it computes $C(0\cdots0)$.  Since any copy-protection scheme can be amplified to have correctness neglibly close to $1$ by providing multiple copies of the copy-protected states, our result also rules out copy-protection for non-negligible correctness parameter $\varepsilon$, as long as the correctness of $(U_C,\rho_C)$ can be amplified to neglibily close to $1$ by providing $\rho_C^{\otimes k}$ for some $k=\poly(\secparam)$. 

\paragraph{Impossibility of  Quantum VBB with single uncloneable state.} Our techniques also rule out the possibility of quantum VBB for classical circuits. In particular, this rules the possibility of quantum VBB for classical circuits with the obfucated circuit being a single uncloneable state, thus resolving an open problem by Alagic and Fefferman~\cite{alagic2016quantum}.

\begin{proposition}
Assuming the quantum hardness of learning with errors and assuming that there is a QFHE satisfying the properties described in Theorem~\ref{thm:impossibility},
\par there exists a circuit class $\cktclass$ such that any quantum VBB for $\cktclass$ is insecure.
\end{proposition}
\begin{proof}
We construct a circuit class $\cktclass=\{\cktclass_{\secparam}\}_{\secparam \in \mathbb{N}}$, where every circuit in $\cktclass_{\secparam}$ is of the form $C_{a,b,r,\pk,\cO}$ defined in the proof of Theorem~\ref{thm:impossibility}. 
\par Given any quantum VBB of $C_{a,b,r,\pk,\cO}$, there exists an adversary $\cA$ that recovers $b$ and outputs the first bit of $b$. The adversary $\cA$ follows steps 1-4 of $\cB$ defined in the proof of Proposition~\ref{prop:eff:dequant} and then outputs the first bit of $b'$. In the same proof, we showed that the probability that $b'=b$ is negligibly close to 1 and thus, the probability it outputs the first bit of $b$ is negligibly close to 1. 
\par On the other hand, any QPT simulator $\simr$ with superposition access to $C_{a,b,r,\pk,\cO}$ can recover $b$ with probability negligibly close to $1/2$. To prove this, we rely upon the proof of Claim~\ref{clm:unlearn}. We will start with the same simplifying assumptions as made in the proof of Claim 46. Suppose $T$ is the number of superposition queries made by $\simr$ to $C_{a,b,r,\pk,\cO}$. Let $\ket{\psi^0}$ is the initial state of $\simr$ and more generally, let $\ket{\psi^t}$ be the state of $\simr$ after $t$ queries, for $t \leq T$. 
\par We define an alternate QPT simulator $\simr'$ which predicts the first bit of $b$ with probability negligibly close to $\simr$. Before we describe $\simr'$, we give the necessary preliminary background. Define $\ket{\phi^t}=U_{t} U_{t-1} \cdots U_1 \ket{\psi^0}$. We proved the following claim. 

\begin{claim}
$|\braket{\phi^t}{\psi^t}|=1-\delta_{t}$ for every $t \in [T]$.
\end{claim}

\noindent $\simr'$ starts with the initial state $\ket{\psi^0}$. It then computes $\ket{\phi^T}$. If $U$ is a unitary matrix $\simr$ applies on $\ket{\psi^T}$ followed by a measurement of a register ${\bf D}$ then $\simr'$ also performs $U$  on $\ket{\phi^T}$ followed by a measurement of ${\bf D}$. By the above claim, it then follows that the probability that $\simr'$ outputs 1 is negligibly close to the probability that $\simr$ outputs 1. But the probability that $\simr'$ predicts the first bit of $b$ is $1/2$. Thus, the probability that $\simr$ predicts the first bit of $b$ is negligibly close to $1/2$. 
\end{proof}

\fullversion{

\section{qIHO for Compute-and-Compute Circuits}
\noindent To complement the impossibility result, we present a construction of SSL for a subclass of evasive circuits. Specifically, the construction works for circuit classes that have q-Input-Hiding obfuscators. In the following section, we show that there are q-Input-Hiding obfuscators for Compute-and-Compare circuits.


 Barak et al.~\cite{BBCKP14} present a construction of input-hiding obfuscators  secure against classical PPT adversaries; however, it is unclear whether their construction is secure against QPT adversaries. Instead we present a construction of input-hiding obfuscators (for a class of circuits different from the ones considered in~\cite{BBCKP14}) from QLWE. Specifically, we show how to construct a q-input-hiding obfuscator for compute-and-compare circuits $\cktclass_{\cnc}$ with respect to a distribution $\distr_{\cktclass}$ defined in Definition~\ref{def:distr:cnc}. 

\begin{lemma}[qIHO for  Compute-and-Compare Circuits]
Consider a class of compute-and-compare  circuits $\cktclass_{\cnc}$ associated with a distribution $\distrc$ (Definition~\ref{def:distr:cnc}). Assuming QLWE, 
there exists qIHO for $\cktclass_{\cnc}$. 
\end{lemma}
\begin{proof}
\noindent We prove this in two steps: we first construct a qIHO for the class of point functions and then we use this to build qIHO for  compute-and-compare class of circuits. \\

\noindent {\em \underline{qIHO for point functions}}: To prove this, we use a theorem due to~\cite{BBCKP14} that states that an average-case VBB for circuits with only polynomially many accepting points is already an input-hiding obfuscator for the same class of circuits; their same proof also holds in the quantum setting. Any q-average-case VBB for circuits with only polynomially many accepting points is already a qIHO. As a special case, we have a qIHO for point functions from q-average-case VBB for point functions. Moreover, we can instantiate q-average-case VBB for point functions from QLWE  and thus, we have qIHO for point functions from QLWE. 
\par We describe the formal details below. First, we recall the definition of average-case VBB. 

\begin{definition}[q-Average-Case Virtual Black-Box Obfuscation (VBB)]
Consider a class of circuits $\cktclass=\{\cktclass_{\secparam}\}_{\secparam \in \mathbb{N}}$ associated with a distribution $\distrc$. We say that $(\obf,\eval)$ is said to be a {\bf q-average-case virtual black-box obfucsator} for $\cktclass$ if it holds that for every QPT adversary $\adversary$, there exists a QPT simulator $\simr$ such that for every $\secparam \in \mathbb{N}$, the following holds for every non-uniform QPT distinguisher $D$: 
$$\left| {\prob} \left[ 1 \leftarrow D \left(\widetilde{C}\right)\ :\ \substack{C \leftarrow \distrc \left( \secparam \right),\\ \ \\ \widetilde{C} \leftarrow \obf\left(1^{\secparam},C \right)} \right] - \prob \left[1 \leftarrow D\left( \widetilde{C} \right)\ :\ \widetilde{C} \leftarrow \simr\left( 1^{\secparam} \right) \right] \right| \leq \negl(\secparam),$$
\end{definition}

\noindent We consider a quantum analogue of a proposition proven in~\cite{BBCKP14}. We omit the proof details since this is identical to the proof provided by~\cite{BBCKP14} albeit in the quantum setting. 

\begin{proposition}
Consider a class of evasive circuits $\cktclass=\{\cktclass_{\secparam}\}_{\secparam \in \mathbb{N}}$ associated with a distribution $\distrc$ such that each circuit $C \in \cktclass_{\secparam}$ has polynomially many accepting points. 
\par Assuming q-average-case virtual black-box obfuscation for $\cktclass$, there is a qIHO for $\cktclass$.
\end{proposition}
\noindent As a special case, we have qIHO for point functions (defined below) assuming q-average-case VBB for point functions. Moreover, q-average-case VBB for point functions can be instantiated from QLWE (see for example~\cite{WZ17,GKW17}). Thus, we have the following proposition.

\begin{proposition}[q-Input-Hiding Obfuscator for Point Functions]
\label{prop:polwe}
Consider the class of circuits $\cktclass=\{\cktclass_{\secparam}\}_{\secparam \in \mathbb{N}}$ defined as follows: every circuit $C \in \cktclass$, is associated with $x$ such that it outputs 1 on $x$ and 0 on all other points. 
\par Assuming QLWE, there is a qIHO for $\cktclass$.  
\end{proposition}

\noindent {\em \underline{qIHO for  compute-and-compare circuits from qIHO for point functions}}: We now show how to construct qIHO for  compute-and-compare circuits $\cktclass_{\cnc}$, associated with distribution $\distr_{\cnc}$ (Definition~\ref{def:distr:cnc}), from qIHO for point functions. Denote $\po.\qiho$ to be a qIHO for point functions $\newcktclass=\{\newcktclass_{\secparam}\}_{\secparam \in \mathbb{N}}$ associated with distribution $\distr_{\smpo}$, where $\distr_{\smpo}$ is a marginal distribution of $\distr_{\cnc}$ on $\{\alpha\}$. We construct qIHO for compute-and-compare circuits below; we denote this by $\cnc.\qiho$. \\

\noindent \underline{$\cnc.\qiho.\obf \left( 1^{\secparam}, \lockC[C,\alpha] \right)$}: It takes as input security parameter $\secparam$, compute-and-compare circuit $\lockC[C,\alpha]$, associated with lock $\alpha$. Compute $\po.\qiho(1^{\secparam},G_{\alpha} \in \newcktclass_{\secparam})$ to obtain $\widetilde{G_{\alpha}}$. Output $\widetilde{\lockC}= \left( C,\widetilde{G_{\alpha}(\cdot)} \right)$. \\

\noindent \underline{$\cnc.\qiho.\eval \left(\widetilde{\lockC},x\right)$}: On input obfuscated circuit $\widetilde{\lockC}=\left( C,\widetilde{G_{\alpha}} \right)$, 
input $x$, do the following: 
\begin{itemize}
    \item Compute $C(x)$ to obtain $\alpha'$. 
    \item Compute $\po.\eval\left(\widetilde{G_{\alpha}},\alpha' \right)$ to obtain $b$.
    \item Output $b$. 
\end{itemize}

\begin{claim}
\label{clm:potocnc}
Assuming $\po.\qiho$ is an input-hiding obfuscator for $\newcktclass$ associated with $\distr_{\smpo}$, $\cnc.\qiho$ is an input-hiding obfuscator for $\cktclass$ associated with $\distr_{\cnc}$. 
\end{claim}
\begin{proof}
Suppose there exists a QPT adversary $\adversary$ such that the following holds: 
$$\left| \prob \left[ \widetilde{\lockC}(x)=1\ :\ \substack{ \lockC[C,\alpha] \leftarrow \distr_{\cnc}(\secparam),\\ \ \\ \widetilde{\lockC} \leftarrow \cnc.\qiho(1^{\secparam},\lockC[C,\alpha]),\\ \ \\ x \leftarrow \adversary\left( 1^{\secparam},\widetilde{\lockC} \right)} \right] \right| = \delta  $$
Our first observation is that $\prob\left[ C(x)=\alpha\ \big|\ \widetilde{C}(x)=1  \right] = 1$. \noindent Using this, we can construct another adversary $\adversary'$ that violates the input-hiding property of $\po.\qiho$. On input $\widetilde{G_{\alpha}(\cdot)}$, $\adversary'$ computes $\adversary\left(\widetilde{\lockC}=\left(C,\widetilde{G_{\alpha}(\cdot)}\right) \right)$; denote the output to be $x$. Finally, $\adversary'$ outputs $\alpha'=C(x)$. 
\par From the above observations, it holds that $\adversary'$ breaks the input-hiding property of $\po.\qiho$ with probability $\delta$. From the security of $\po.\qiho$, we have that $\delta=\negl(\secparam)$ and thus the proof of the claim follows.  

\end{proof} 

\noindent \underline{\em Conclusion}: Combining Claim~\ref{clm:potocnc} and Proposition~\ref{prop:polwe}, we have $\qiho$ for  compute-and-compare circuits from QLWE.

\end{proof}

}

\section{Main Construction}
\label{sec:const}
In this section, we present the main construction of SSL satisfying infinite-term perfect lessor security. \submversion{We first start by describing the class of circuits of interest.

}

\fullversion{
\par Let $\cktclass=\{\cktclass_{\secparam}\}$ be the class of $\search$-searchable circuits associated with SSL. We denote  $\sz(\secparam)=\poly(\secparam)$ to be the maximum size of all circuits in $\cktclass_{\secparam}$. And let $\distr_{\cktclass}$ be the distribution associated with $\cktclass$. 
}

 \fullversion{
 \paragraph{Ingredients.}
\begin{enumerate}
    \item q-Input-hiding obfuscator $\iho=(\iho.\obf,\iho.\eval)$ for $\cktclass$. 
        \item Subspace hiding obfuscation
    $\shO=(\shO.\obf,\shO.\eval)$. The field associated with $\shO$ is $\Zq$ and the dimensions will be clear below. 
    \item q-simulation-extractable non-interactive zero-knowledge system  $\ssnizk=(\crsgen,\prover,\verify)$ for NP with sub-exponential security as guaranteed in Lemma~\ref{lem:qsenizks}.
\end{enumerate}
}

\fullversion{
\paragraph{Construction.} We describe the scheme of SSL below. 
}

\submversion{
\subsection{Construction}
}
\noindent We describe the scheme of SSL below. \submversion{ We encourage the reader to look at the overview of the construction presented in Section~\ref{sec:overview:cons} before reading the formal details below.}

\begin{itemize}

    \item $\setup(1^{\secparam})$: Compute $\crs \leftarrow \crsgen \left( 1^{\secparam_1} \right)$, where $\secparam_1=\secparam+n$ and $n$ is the input length of the circuit. Output $\crs$.
    \ \\
    \item $\gen(\crs)$: On input common reference string $\crs$, choose a random $\frac{\secparam}{2}$-dimensional subspace $A \subset \Zq^{\secparam}$. Set $\key = A$. 
    \ \\
    \item $\lessor(\key=A, C)$: On input secret key $\key$, circuit $C \in \cktclass_{\secparam}$, with input length $n$, 
    \begin{enumerate}
        \item Prepare the state $\ket{A} = \frac{1}{\sqrt{q^{\secparam/2}}} \underset{a \in A}{\sum} \ket{a}$. 
        \item Compute $\widetilde{C}\leftarrow \iho.\obf(C;r_o)$
        \item Compute $\widetilde{g} \leftarrow \shO(A;r_{A})$. 
        \item Compute $\widetilde{g_{\bot}} \leftarrow \shO(A^{\bot};r_{A^{\perp}})$.
        \item Let $x = \search(C)$; that is, $x$ is an accepting point of $C$. 
        \item Let $L$ be the NP language defined by the following NP relation. 
        $$\rel_L:=\left\{\left(\left(\widetilde{g},\widetilde{g_{\perp}},\widetilde{C} \right),\ \left(A,r_o,r_A,r_{A^\perp},C,x\right)\right) \Bigg|\ \substack{\widetilde{g} = \shO(A;r_A) \\ \widetilde{g_{\perp}} = \shO(A^\perp;r_{A^\perp}) \\ \widetilde{C} = \iho.\obf(C;r_o),\\ C(x)=1}\right\}. $$
        Compute $\pi \leftarrow \prover \left(\crs,\left(\widetilde{g},\widetilde{g_{\bot}},\widetilde{C} \right),\left( A,r_o,r_A,r_{A^\perp},C,x \right) \right) $
        \item Output $\rho_C=\ket{\Phi_C}\bra{\Phi_C}= \left( \ket{A}\bra{A},\widetilde{g},\widetilde{g_{\bot}},\widetilde{C},\pi \right)$.  
    \end{enumerate}
    \ \\
    \item $\run(\crs,\rho_C,x)$:
    \begin{enumerate}
        \item Parse $\rho_C$ as $\left( \rho,\widetilde{g},\widetilde{g_{\bot}},\widetilde{C},\pi\right)$. In particular, measure the last 4 registers. \\
        {\em Note: This lets us assume that the input to those registers is just classical, since anyone about to perform $\run$ might as well measure those registers themselves.}
        \item We denote the operation $\shO.\eval(\widetilde{g}, \ket{x}\ket{y}) = \ket{x}\ket{y \oplus \mathbb{1}_A(x)}$ by $\widetilde{g}[\ket{x}\ket{y}]$, where $\mathbb{1}_A(x)$ is an indicator function that checks membership in $A$. Compute $\widetilde{g}[\rho \otimes \ket{0}\bra{0}]$ and measure the second register.  Let $a$ denote the outcome bit, and let $\rho'$ be the post-measurement state.
        \item As above, we denote the operation $\shO.\eval(\widetilde{g_{\bot}}, \ket{x}\ket{y}) = \ket{x}\ket{y \oplus \mathbb{1}_A(x )}$ by $\widetilde{g_{\bot}}[\ket{x}\ket{y}]$. Compute $\widetilde{g_{\bot}}[\ft \rho' \ft^{\dagger} \otimes \ket{0}\bra{0}]$ and measure the second register. Let $b$ denote the outcome bit.
        
        {\em Note: in Step 2 and 3, $\run$ is projecting $\rho$ onto $\ket{A}\bra{A}$ if $a=1$ and $b=1$.}
        
        \item Afterwards, perform the Fourier Transform again on the first register of the post-measurement state, let $\rho''$ be the resulting state.
        
        \item Compute $c \leftarrow \verify\left(\crs, \left(\widetilde{g},\widetilde{g_{\bot}},\widetilde{C} \right),\pi\right)$
        \item If either $a=0$ or $b=0$ or $c=0$, reject and output $\bot$.
        \item Compute $y \leftarrow \iho.\eval \left(\widetilde{C},x \right)$.
        \item Output $\left(\rho'',\widetilde{g},\widetilde{g_{\bot}},\widetilde{C},\pi\right)$ and $y$.
    \end{enumerate}
    \item $\return(\key=A,\rho_C)$:
    \begin{enumerate}
        \item Parse $\rho_C$ as $\left( \rho,\widetilde{g},\widetilde{g_{\bot}},\widetilde{C},\pi\right)$.
        \item Perform the measurement $\{\ket{A}\bra{A},I-\ket{A}\bra{A}\}$ on $\rho$. If the measurement outcome corresponds to $\ket{A}\bra{A}$, output $1$. Otherwise, output $0$.
    \end{enumerate}
\end{itemize}

\begin{lemma} [Overwhelming probability of perfect correctness] The above scheme  satisfies $\epsilon=\negl(\secparam)$ correctness.
\end{lemma}
\begin{proof}
We first argue that the correctness of $\run$ holds. Since $\qiho$ is perfectly correct, it suffices to show that $\run$ will not output $\bot$. For this to happen, we need to show that $a,b,c=1$. Since $\widetilde{g}=\shO(A)$,  $\widetilde{g_\perp}=\shO(A^\perp)$, and the input state is $\ket{A}\bra{A}$, then $a=1$ and $b=1$ with probability negligibly close to $1$ by correctness of $\shO$.  If $\pi$ is a correct proof, then by perfect correctness of $\ssnizk$, we have that $\prob[c=1]=1$.   \\
\\ To see that the correctness of $\return$ also holds, note that the leased state is $\rho=\ket{A}\bra{A}$, which will pass the check with probability $1$.

\end{proof}

\begin{lemma}
\label{lem:mainconstruction}
Fix $\beta=\mu(\secparam)$, where $\mu(\secparam)$ is any non-negligible function. Assuming the security of $\qiho,\ssnizk$ and $\shO$, the above scheme satisfies  $(\beta,\gamma,\distrc)$-infinite-term perfect lessor security, where $\gamma$ is a negligible function. 
\end{lemma}

\submversion{
\noindent The proof of the above lemma is presented in Section~\ref{sec:mainproof}. 
}

\fullversion{
\begin{proof}
For any QPT adversary $\cA$, define the following event. \\

\noindent \underline{$\realevent(1^{\secparam})$}: 
\begin{itemize}
    \item $\crs \leftarrow \setup\left( 1^{\secparam} \right)$, 
    \item $\key \leftarrow \gen(\crs)$,
    \item $C \leftarrow \distrc(\secparam)$,
    \item $\left(\rho_C=\left( \ket{A}\bra{A},\widetilde{g},\widetilde{g_{\bot}},\widetilde{C},\pi \right)\right)  \leftarrow \lessor\left(\key,r \right)$ 
    \item $\rho^*=\left( \widetilde{C}^{(1)},\widetilde{g}^{(1)},\widetilde{g_{\bot}}^{(1)},\pi^{(1)},\widetilde{C}^{(2)},\widetilde{g}^{(2)},\widetilde{g_{\bot}}^{(2)},\pi^{(2)},\sigma^*  \right) \leftarrow \adversary\left(\crs, \rho_C \right)$\\
    {\em That is, $\adversary$ outputs two copies; the classical part in the first copy is $\left( \widetilde{C}^{(1)},\widetilde{g}^{(1)},\widetilde{g_{\bot}}^{(1)},\pi^{(1)} \right)$ and the classical part in the second copy is $\left( \widetilde{C}^{(2)},\widetilde{g}^{(2)},\widetilde{g_{\bot}}^{(2)},\pi^{(2)} \right)$. Moreover, it outputs a single density matrix $\sigma^*$ associated with two registers $\Reg_1$ and $\Reg_2$; the state in $\Reg_1$ is associated with the first copy and the state in $\Reg_2$ is associated with the second.   }
    \item $ \sigma^*_1 = \tr_2[\sigma^*]$
    \item $\rho_C^{(1)} = \left(\sigma^*_1, \widetilde{C}^{(1)},\widetilde{g}^{(1)},\widetilde{g_{\bot}}^{(1)},\pi^{(1)} \right) \bigwedge \rho_C^{(2)} = \left( \Pi_2( \sigma^*),\widetilde{C}^{(2)},\widetilde{g}^{(2)},\widetilde{g_{\bot}}^{(2)},\pi^{(2)} \right)$ where
    $$\Pi_2(\sigma^*) = \frac{\tr_1 \left[(\Pi_{(\widetilde{g^{(1)}},{\widetilde{g^{(1)}_\perp}})}\otimes I) \sigma^*\right]}{\tr \left[(\Pi_{(\widetilde{g^{(1)}},{\widetilde{g^{(1)}_\perp}})}\otimes I) \sigma^*\right]} $$
    and where $\Pi_{(\widetilde{g^{(1)}},{\widetilde{g^{(1)}_\perp}})}$ is the projection onto the subspace obfuscated by $(\widetilde{g^{(1)}}, \widetilde{g^{(1)}_\perp})$. In other words, $\Pi_2(\sigma^*)$ is the quantum state on register 2 conditioned on $\run$ not outputting $\bot$ when applied to register $1$.
\end{itemize}

\noindent To prove the lemma, we need to prove the following: 

 $$\underset{}{\prob}\left[ \substack{ \forall x, \prob\left[ (\run(\crs,x,\sigma_1^*) = C(x)\right] \geq \beta \\ \ \\ \bigwedge \\ \ \\  \forall x,x', \prob \left[\run(\crs, x', \cE_{x}^{(2)}(\sigma^*)) = C(x')  \right] \geq \beta }\ :\ \realevent\left( 1^{\secparam} \right) \right] \leq \gamma.$$
 
 \noindent Note that for all $x$, $\cE_x^{(2)}(\sigma^*) = \Pi_2(\sigma^*)$, since the only quantum operation that $\run$ performs is projecting the first register of $\sigma^*$ onto the subspace corresponding to $\widetilde{g}^{(1)}$.
\noindent Consider the following: 
\begin{itemize}
    \item Define $\gamma_1$ as follows: 
    $$\prob \left[ \substack{ \forall x, \prob\left[ (\run(\crs,x,\sigma_1^*) = C(x)\right] \geq \beta \\ \ \\ \bigwedge \\ \ \\  \forall x', \prob \left[\run(\crs, x', \Pi_2(\sigma^*)) = C(x')  \right] \geq \beta \\ \ \\ \bigwedge \\ \ \\  \left(\widetilde{C},\widetilde{g},\widetilde{g
_{\bot}}\right) =  \left(\widetilde{C}^{(1)},\widetilde{g}^{(1)},\widetilde{g_{\bot}}^{(1)} \right) \\ \ \\ \bigwedge \\ \ \\ \left(\widetilde{C},\widetilde{g},\widetilde{g
_{\bot}}\right) =  \left(\widetilde{C}^{(2)},\widetilde{g}^{(2)},\widetilde{g_{\bot}}^{(2)} \right) }\  :\ \realevent\left( 1^{\secparam} \right) \right] = \gamma_1$$
    \item The other possible case is the case where at least one of the copies $\left(\widetilde{C}^{(1)},\widetilde{g}^{(1)},\widetilde{g_{\bot}}^{(1)} \right)$ or $\left(\widetilde{C}^{(2)},\widetilde{g}^{(2)},\widetilde{g_{\bot}}^{(2)} \right)$ is not equal to the corresponding resgisters of the original copy. Without loss of generality, we will assume that the the second copy is not the same. Define $\gamma_2$ as follows: 
     $$\prob \left[ \substack{ \forall x, \prob\left[ (\run(\crs,x,\sigma_1^*) = C(x)\right] \geq \beta \\ \ \\ \bigwedge \\ \ \\  \forall x', \prob \left[\run(\crs, x', \Pi_2(\sigma^*)) = C(x')  \right] \geq \beta \\ \ \\ \bigwedge \\ \ \\ \left(\widetilde{C},\widetilde{g},\widetilde{g
_{\bot}}\right) \neq  \left(\widetilde{C}^{(2)},\widetilde{g}^{(2)},\widetilde{g_{\bot}}^{(2)}\right) }\ :\ \realevent\left( 1^{\secparam} \right) \right] = \gamma_2$$
\end{itemize}
Note that $\gamma=\gamma_1+\gamma_2$. In the next two propositions, we prove that both $\gamma_1$ and $\gamma_2$ are negligible which will complete the proof of the lemma.

\begin{proposition}
$\gamma_1\leq \negl(\secparam)$
\end{proposition}
\begin{proof}

The run algorithm first projects $\sigma^*$ into $\ket{A}^{\otimes 2}$, and outputs $\bot$ if $\sigma^*$ is not $(\ket{A}\bra{A})^{\otimes 2}$.
Suppose that $\bra{A} \sigma_1^* \ket{A}$ is negligible, then $\run$ will output $\bot$ on the first register with probability negligibly close to $1$, and we would have $\gamma_1$ negligible as desired. 

On the contrary, suppose that $\bra{A} \sigma_1^* \ket{A}$ is non-negligible, and we have that $$\Pi_2(\sigma^*)=\frac{\tr_1\left[(\ket{A}\bra{A}\otimes I)\sigma^*\right]}{\tr\left[(\ket{A}\bra{A}\otimes I)\sigma^*\right]}$$ 
i.e. the state in the second register after $\run$ succesfully projects $\sigma_1^*$ onto $\ket{A}\bra{A}$.

We will prove the following claim, which implies that at least one of the two copies will output $\bot$ under $\run$ with probability neglibly close to $1$.
\begin{claim}\label{clm:notsame}
$\bra{A} \Pi_2(\sigma^*) \ket{A} \leq \negl(\secparam)$
\end{claim}
\begin{proof}
Suppose not.  Then, we can use $\cA$ to break quantum no-cloning.  Specifically, Zhandry \cite{Zha19} showed that no QPT algorithm on input $(\ket{A},\widetilde{g}:=\shO(A),\widetilde{g_\perp}:=\shO(A^\perp))$ can prepare the state $\ket{A}^{\otimes 2}$ with non-negligible probability. We will show that $\cA$ allows us to do exactly this if $\bra{A} \Pi_2(\sigma^*)\ket{A}$ is non-negligible.

Consider the following adversary $\cB'$. It runs $\cA$ and then projects the output of $\cA$ onto $\left( \ket{A}\bra{A} \right)^{\otimes 2}$; the output of the projection is the output of $\cB'$. 

\paragraph{$\cB'(C)$:}
\begin{enumerate}
    \item Compute $\crs,\key$ as in the construction
    \item Compute $\rho_C \leftarrow \lessor(\key,C)$. Let $\rho_C=\left( \ket{A}\bra{A},\widetilde{g},\widetilde{g_{\bot}},\widetilde{C},\pi \right)$. 
    \item Compute $\cA(\crs,\rho_C)$ to obtain $\left( \widetilde{C}^{(1)},\widetilde{g}^{(1)},\widetilde{g_{\bot}}^{(1)},\pi^{(1)},\widetilde{C}^{(2)},\widetilde{g}^{(2)},\widetilde{g_{\bot}}^{(2)},\pi^{(2)},\sigma^*\right)$.
    \item Then, project $\sigma^*$ onto $(\ket{A}\bra{A})^{\otimes 2}$ by using $\widetilde{g}$ and $\widetilde{g_\perp}$. Let $m$ be the outcome of this projection, so $m=1$ means that the post measured state is $(\ket{A}\bra{A})^{\otimes 2}$. 
    \item Output $m$.
\end{enumerate}

The projection $(\ket{A}\bra{A})^{\otimes 2}$ can be done by first projecting the first register onto $\ket{A}\bra{A}$ and then the second register. Conditioned on the first register not outputting $\bot$, means that $\sigma^*_1$ is succesfully projected onto $\ket{A}\bra{A}$. By our assumption that $\bra{A} \sigma^*_1 \ket{A}$ is non-negligible, this will happen with non-negligible probability.  Conditioned on this being the case, if $\bra{A} \Pi_2( \sigma^*) \ket{A}$ is non-negligible, then projecting the second register onto $\ket{A}\bra{A}$ will also succeed with non-negligible probability.  This means that $m=1$ with non-negligible probability.

 Consider the following adversary. It follows the same steps as $\cB'$ except in preparing the states $\ket{A}$ and computing obfuscations $\widetilde{g}$, $\widetilde{g_{\bot}}$; it gets these quantities as input. Moreover, it simulates the proof $\pi$ instead of computing the proof using the honest prover. This is because unlike $\cB'$, the adversary $\cB$ does not have the randomness used in computing $\widetilde{g}$ and $\widetilde{g_{\bot}}$ and hence cannot compute the proof $\pi$ honestly. 

\paragraph{$\cB(\ket{A},\widetilde{g},\widetilde{g_\perp})$:}
\begin{enumerate}
    \item Sample randomness $r_o$ and compute $\tilde{C} \leftarrow \qiho.\obf(C;r_o)$.
    \item Let $\fksetup$ and $\simr$ be associated with the simulation-extractability propety of $\ssnizk$. Compute $(\widetilde{\crs},\td) \leftarrow \fksetup(1^{\secparam})$. 
    \item Compute $(\pi,\st) \leftarrow \simr \left( \widetilde{\crs},\td,\left(\widetilde{g},\widetilde{g_\perp},\widetilde{C} \right) \right) $
    \item Let $\rho_C= (\ket{A}\bra{A},\widetilde{g}, \widetilde{g_\perp}, \widetilde{C},\pi)$
    \item Run $\cA(\widetilde{\crs},\rho_C)$ to obtain $\left(\widetilde{C}^{(1)},\widetilde{g}^{(1)},\widetilde{g_{\bot}}^{(1)},\pi^{(1)},\widetilde{C}^{(2)},\widetilde{g}^{(2)},\widetilde{g_{\bot}}^{(2)},\pi^{(2)},\widetilde{\sigma^*}\right)$.
    \item Then, project $\widetilde{\sigma^*}$ onto $(\ket{A}\bra{A})^{\otimes 2}$ by using $\widetilde{g}$ and $\widetilde{g_\perp}$. Let $m$ be the outcome of this projection, so $m=1$ means that the post measured state is $(\ket{A}\bra{A})^{\otimes 2}$.
    \item Output $m$.
\end{enumerate}

\noindent Note that from the q-simulation-extractability property\footnote{We don't need the full-fledged capability of q-simulation-extractability to argue this part; we only need q-zero-knowledge property which is implied by q-simulation-extractability.} of $\ssnizk$, it follows that the probability that $\cB$ outputs 1 is negligibly close to the probability that $\cB'$ outputs $1$ because everything else is sampled from the same distribution.  This implies that $\cB$ on input $(\ket{A},\widetilde{g},\widetilde{g_\perp})$ outputs $\ket{A}^{\otimes 2}$ with non-negligible probability, contradicting~\cite{Zha19}.

\end{proof}

At this point, we want to show that if $\left( \widetilde{g^{(2)}},\widetilde{g^{(2)}_\perp} \right) = \left( \widetilde{g},\widetilde{g_\perp} \right)$, and $\bra{A}\Pi_2(\sigma^*)\ket{A} \leq \negl(\secparam)$, then the probability that $\run(\crs, \Pi_2(\sigma^*),x)$ evaluates $C$ correctly is negligible.

By correctness of $\shO$, we have $$\prob[\forall x \text{ } \widetilde{g^{(2)}}(x) = \mathbb{1}_A(x)] \geq 1-\negl(\secparam)$$ $$\prob[\forall x \text{ } \widetilde{g_\perp^{(2)}}(x) = \mathbb{1}_{A^\perp}(x)] \geq 1-\negl(\secparam)$$ 

This means that with probability negligibly close to $1$, the first thing that the $\run$ algorithm does on input $\rho_C^{(2)}=(\Pi_2(\sigma^*), \widetilde{g^{(2)}}, \widetilde{g^{(2)}_\perp}, \widetilde{C},\pi)$ is to measure $\{\ket{A}\bra{A},I-\ket{A}\bra{A}\}$ on $\Pi_2( \sigma^*)$. If $I-\ket{A}\bra{A}$ is obtained, then the $\run$ algorithm will output $\bot$. By Claim~\ref{clm:notsame}, the probability that this happens is neglibly close to 1. Formally, when $\widetilde{g}$ and $\widetilde{g_\perp}$ are subspace obfuscations of $A$ and $A^\perp$ respectively, the check $a=1$ and $b=1$ performed by the $\run$ algorithm is a projection onto $\ket{A}\bra{A}$. 

\begin{align*}
\prob[a=1,b=1] &= \tr[\ft^\dagger \Pi_{A^\perp} \ft \Pi_A \Pi_2( \sigma^*)] \\
&= \tr[\ket{A}\bra{A} \Pi_2( \sigma^*)] \\
&= \bra{A}\Pi_2( \sigma^*)\ket{A}\\
&\leq \negl(\secparam)
\end{align*}
where $\Pi_A = \underset{a \in A}{\sum} \ket{a}\bra{a}$ and $\Pi_{A^\perp} = \underset{a \in A^\perp}{\sum} \ket{a}\bra{a}$.
From this, we have that $\prob[\run(\crs,\rho_C^{(2)},x)=\bot] \geq 1-\negl(\secparam)$, and we have
$\prob[\run(\crs,\rho_C^{(2)},x)=C(x)] \leq \negl(\secparam)$ with probability neglibly close to $1$.

This finishes our proof that if $\beta$ is non-negligible, then $\gamma_1\leq \negl(\secparam)$.
\end{proof}
\begin{proposition}
$\gamma_2\leq\negl(\secparam)$.  
\end{proposition}
\begin{proof}
We consider the following hybrid process. \\

\noindent \underline{$\hybprocess_1(1^{\secparam})$}: 
\begin{itemize}
     \item $\left( \widetilde{\crs},\td \right) \leftarrow \fksetup\left( 1^{\secparam} \right)$,
    \item $\key \leftarrow \gen(\crs)$,
    \item $C \leftarrow \distrc(\secparam)$,
    \item Sample a random $\frac{\secparam}{2}$-dimensionall sub-space $A \subset \Zq^{\secparam}$. Prepare the state $\ket{A}=\frac{1}{\sqrt{q^{\secparam/2}}} \sum_{a \in A} \ket{a}$.
    \item Compute $\widetilde{g} \leftarrow \shO \left(A;r_A\right)$, 
    \item Compute $\widetilde{g_{\bot}} \leftarrow \shO \left(A^{\perp};r_{A^{\perp}}\right)$,
    \item Compute $\widetilde{C}\leftarrow \iho.\obf(C;r_o)$
    \item $(\pi,\st) \leftarrow \simr_1 \left(\crs, \td,\left(\widetilde{g}, \widetilde{g_{\bot}},\widetilde{C} \right) \right) $
    \item Set $\rho_C=\left( \ket{A}\bra{A},\widetilde{g},\widetilde{g_{\bot}},\widetilde{C},\pi \right)$. 
    \item $\left( \widetilde{C}^{(1)},\widetilde{g}^{(1)},\widetilde{g_{\bot}}^{(1)},\pi^{(1)},\widetilde{C}^{(2)},\widetilde{g}^{(2)},\widetilde{g_{\bot}}^{(2)},\pi^{(2)},\sigma^*  \right) \leftarrow \adversary\left(\crs, \rho_C \right)$
    \item Set $ \sigma^*_1 = \tr_2[\sigma^*]$
    \item Set $\rho_C^{(1)} = \left( \sigma^*_1,\widetilde{C}^{(1)},\widetilde{g}^{(1)},\widetilde{g_{\bot}}^{(1)},\pi^{(1)} \right)$ and $ \rho_C^{(2)} = \left( \Pi_2( \sigma^*),\widetilde{C}^{(2)},\widetilde{g}^{(2)},\widetilde{g_{\bot}}^{(2)},\pi^{(2)} \right)$
    \item $\left(A^*,r_o^*,r_A^*,r_{A^{\perp}}^*,C^*,x^* \right) \leftarrow \simr_2\left(\st,\left(\widetilde{g}^{(2)},\widetilde{g_{\bot}}^{(2)},\widetilde{C}^{(2)} \right),\ \pi^{(2)} \right)$. 
\end{itemize}

\noindent The proof of the following claim follows from the q-simulation-extractactability property of $\ssnizk$. 

\begin{claim}
\label{clm:senizk}

Assuming that $\ssnizk$ satisfies q-simulation extractability property secure against QPT adversaries running in time $2^{n}$, we have: 
$$\prob \left[ \substack{\left(\left( \widetilde{g}^{(2)},\widetilde{g_{\bot}}^{(2)},\widetilde{C}^{(2)} \right),\ \left(A^*,r_o^*,r_A^*,r_{A^{\perp}}^*,C^*,x^* \right)\right) \in \rel(L)  \\ \ \\ \bigwedge\\ \ \\  \forall x',\  \prob\left[\run\left(\crs,\rho^{(2)},x' \right) = C(x')  \right]  \geq \beta \\ \ \\ \bigwedge \\ \ \\ \left(\widetilde{C},\widetilde{g},\widetilde{g_{\bot}} \right)  \neq \left( \widetilde{C}^{(2)},\widetilde{g}^{(2)},\widetilde{g_{\bot}}^{(2)} \right)}\ :\ \hybprocess_1 \left(1^{\secparam} \right) \right] = \delta_1$$ 
Then, $|\delta_1-\gamma_2| \leq \negl(\secparam)$. 
\end{claim}
\begin{remark}
Note that $n$ is smaller than the length of the NP instance and thus, we can invoke the sub-exponential security of the seNIZK system guaranteed in Lemma~\ref{lem:qsenizks}. 
\end{remark}
\begin{proof}[Proof of Claim~\ref{clm:senizk}]
Consider the following $\ssnizk$ adversary $\reduction$: 
\begin{itemize}
    \item It gets as input $\crs$.
    \item It samples and computes $(C,A,\widetilde{g},\widetilde{g_{\bot}},\widetilde{C})$ as described in $\hybprocess_1(1^{\secparam})$. It sends the following instance-witness pair to the challenger of seNIZK: $$\left(\left(C,A,\widetilde{g},\widetilde{g_{\bot}},\widetilde{C}\right),\ ((A,r_o,r_A,r_{A^{\bot}},C,x) \right),$$ 
    where $r_o,r_A,r_{A^{\perp}}$ is, respectively, the randomness used to compute obfuscations  $\widetilde{g}$, $\widetilde{g_{\perp}}$ and $ \widetilde{C}$.   
    \item The challenger returns back $\pi$. 
    \item $\reduction$ then sends $\left(\ket{A},\widetilde{g},\widetilde{g_{\perp}},\widetilde{C},\pi \right)$ to $\cA$. 
    \item $\cA$ then outputs $\left( \widetilde{C}^{(1)},\widetilde{g}^{(1)},\widetilde{g_{\bot}}^{(1)},\pi^{(1)},\widetilde{C}^{(2)},\widetilde{g}^{(2)},\widetilde{g_{\bot}}^{(2)},\pi^{(2)},\sigma^*  \right)$. 
     \item $\reduction$ sets $ \sigma^*_1 = \tr_2[\sigma^*]$. 
    \item Finally, $\reduction$ performs the following checks: 
    \begin{itemize}
        \item {\em Verify if the classical parts are different}: Check if $\left( \widetilde{C},\widetilde{g},\widetilde{g_{\bot}} \right) = \left( \widetilde{C}^{(2)},\widetilde{g}^{(2)},\widetilde{g_{\bot}}^{(2)} \right)$; if so output $\bot$, otherwise continue. 
        \item {\em Verify if second copy computes $C$}: If the measurement above does not output $\bot$, set $ \rho_C^{(2)} = \left( \Pi_2(\sigma^*),\widetilde{C}^{(2)},\widetilde{g}^{(2)},\widetilde{g_{\bot}}^{(2)},\pi^{(2)}\right)$. For every $x$, check if $\widetilde{C}^{(2)}(x)=C(x)$. If for any $x$, the check fails, output $\bot$. {\em // Note that this step takes time $2^{O(n+\log(n))}$.}

    \end{itemize}
    \item Output $\left(\left(\widetilde{g}^{(2)},\widetilde{g_{\bot}}^{(2)},\widetilde{C}^{(2)} \right),\pi^{(2)} \right)$.  
\end{itemize}
\noindent Note that $\reduction$ is a valid $\ssnizk$ adversary: it produces a proof on an instance different from one for which it obtained a proof (either real or simulated) and moreover, the proof produced by $\reduction$ (conditioned on not $\bot$) is an accepting proof.
\par If $\reduction$ gets as input honest CRS and honestly generated proof $\pi$ then this corresponds to $\process_1(1^{\secparam})$ and if $\reduction$ gets as input simulated CRS and simulated proof $\pi$ then this corresponds to $\hybprocess_1(1^{\secparam})$.

Thus, from the security of q-simulation-extractable NIZKs, we have that $|\gamma_2 - \delta_1| \leq \negl(\secparam)$. 
\end{proof}

\noindent We first prove the following claim. 

\begin{claim}
$$\left(\left(\left( \widetilde{g}^{(2)},\widetilde{g_{\bot}}^{(2)},\widetilde{C}^{(2)} \right),\ \left(A^*,r_o^*,r_A^*,r_{A^{\perp}}^*,C^*,x^* \right)\right) \in \rel(L) \ \bigwedge\ \forall x,\  \prob\left[\run\left(\crs,\rho_C^{(2)},x \right) = C(x)  \right]  \geq \beta \right)$$
$$\Longrightarrow C(x^*)=1,$$
\end{claim}
\begin{proof}
We first claim that $\forall x,\  \prob\left[\run\left(\crs,\rho_C^{(2)},x \right) = C(x)  \right]  \geq \beta$ implies that $\widetilde{C}^{(2)} \equiv C$, where $\equiv$ denotes functional equivalence. Suppose not. Let $x'$ be an input such that $\widetilde{C}^{(2)}(x') \neq C(x')$ then this means that $\run(\crs,\rho_C^{(2)},x')$ {\em always} outputs a value different from  $C(x')$; follows from the description of $\run$. This means that $\prob[ \run\left(\crs,\rho_C^{(2)},x' \right) = C(x')]=0$, contradicting the hypothesis. \par Moreover, $\left(\left( \widetilde{g}^{(2)},\widetilde{g_{\bot}}^{(2)},\widetilde{C}^{(2)} \right),\ \left(A^*,r_o^*,r_A^*,r_{A^{\perp}}^*,C^*,x^* \right)\right) \in \rel(L)$ implies that $\widetilde{C}^{(2)}=\qiho(1^{\secparam},C^*;r_{o}
^*)$ and $C^*(x^*)=1$. Furthermore, perfect correctness of $\qiho$ implies that $\widetilde{C}^{(2)} \equiv C^*$.  
\par So far we have concluded that $\widetilde{C}^{(2)} \equiv C$,  $\widetilde{C}^{(2)} \equiv C^*$ and $C^*(x^*)=1$. Combining all of them together, we have $C(x^*)=1$.    

\end{proof}

\noindent Consider the following inequalities. 

\begin{eqnarray*}
\delta_1 & = & \prob \left[ \substack{\left(\left( \widetilde{g}^{(2)},\widetilde{g_{\bot}}^{(2)},\widetilde{C}^{(2)} \right),\ \left(A^*,r_o^*,r_A^*,r_{A^{\perp}}^*,C^*,x^* \right)\right) \in \rel(L)\\ \ \\ \bigwedge\\ \ \\ \forall x',\  \prob\left[\run\left(\crs,\rho_C^{(2)},x' \right) = C(x')  \right]  \geq \beta \\ \ \\ \bigwedge \\ \ \\ \left(\widetilde{C},\widetilde{g},\widetilde{g_{\bot}} \right)  \neq \left( \widetilde{C}^{(2)},\widetilde{g}^{(2)},\widetilde{g_{\bot}}^{(2)} \right)}\ :\ \hybprocess_1 \right]\\ 
& & \\
& & \\
& = & \prob \left[ \substack{C(x^*)=1\\ \ \\ \bigwedge\\ \ \\  \left(\widetilde{C},\widetilde{g},\widetilde{g_{\bot}} \right)  \neq \left( \widetilde{C}^{(2)},\widetilde{g}^{(2)},\widetilde{g_{\bot}}^{(2)} \right)}\ :\ \hybprocess_1 \right] \\ 
& & \\
& \leq & \prob \left[ C\left( x^* \right)=1\ :\ \hybprocess_1 \right]
\end{eqnarray*}
\noindent Let $\prob \left[ C\left( x^* \right)=1\ :\ \hybprocess_1 \right]=\delta_2$.

\begin{claim}
Assuming the q-input-hiding property of $\iho$, we have $\delta_2 \leq \negl(\secparam)$ 
\end{claim}
\begin{proof}
Suppose $\delta_2$ is not negligible. Then we construct a QPT adversary $\reduction$ that violates the q-input-hiding property of $\iho$, thus arriving at a contradiction. 
\par $\reduction$ now takes as input $\widetilde{C}$ (an input-hiding obfuscator of $C$), computes $\left( \widetilde{\crs},\td \right) \leftarrow \fksetup \left( 1^{\secparam} \right)$ and then computes $\rho_C=\left( \ket{A},\widetilde{g},\widetilde{g_{\bot}},\widetilde{C},\pi \right)$ as computed in $\hybprocess_1$. It sends $\left(\widetilde{\crs},\rho_C \right)$ to $\adversary$ who responds with $\left( \widetilde{C}^{(1)},\widetilde{g}^{(1)},\widetilde{g_{\bot}}^{(1)},\pi^{(1)},\widetilde{C}^{(2)},\widetilde{g}^{(2)},\widetilde{g_{\bot}}^{(2)},\pi^{(2)},\sigma^*  \right)$. Compute $\left(A^*,r_o^*,r_A^*,r_{A^{\perp}}^*,C^*,x^* \right)$ by generating $ \simr_2(\st,(\widetilde{g}^{(2)},\widetilde{g_{\bot}}^{(2)},\allowbreak \widetilde{C}^{(2)} ),\allowbreak \pi^{(2)} )$, where $\st$ is as defined in $\hybprocess_1$. Output $x^*$.
\par Thus, $\reduction$ violates the q-input-hiding property of $\iho$ with probability $\delta_2$ and thus $\delta_2$ has to be negligible.
\end{proof}

\noindent Combining the above observations, we have that $\gamma_2 \leq \negl(\secparam)$ for some negligible function $\negl$. This completes the proof.

\end{proof}

\end{proof}
}



\bibliographystyle{plain}
\bibliography{crypto}

\begin{thebibliography}{10}
\bibitem{revenue}
How microsoft corporation makes most of its money.
\end{thebibliography}

\newpage

\submversion{
\section*{Supplementary Material}

\newcommand{\qnizk}{\mathsf{qNIZK}}
\newcommand{\pke}{\mathsf{qPKE}}
\newcommand{\fkgen}{\fksetup}
\newcommand{\expt}{\mathsf{Expt}}
\newcommand{\evnt}{\mathsf{Process}}
\newcommand{\view}{\mathsf{View}}
\section{Preliminaries}
\label{sec:prelims}
We assume that the reader is familiar with basic cryptographic notions such as negligible functions and computational indistinguishability (see~\cite{Gol01}). 
\par The security parameter is denoted by $\secparam$ and we denote $\negl(\secparam)$ to be a negligible function in $\secparam$. We denote (classical) computational indistiguishability of two distributions $\distr_0$ and $\distr_1$ by $\distr_0 \approx_{c,\varepsilon} \distr_1$. In the case when $\varepsilon$ is negligible, we drop $\varepsilon$ from this notation. 

\subsection{Quantum}
\label{ssec:notation}

For completeness, we present some of the basic quantum definitions, for more details see \cite{nielsen2002quantum}.
\paragraph{Quantum states and channels.} Let $\cH$ be any finite Hilbert space, and let $L(\cH):=\{\cE:\cH \rightarrow \cH \}$ be the set of all linear operators from $\cH$ to itself (or endomorphism). Quantum states over $\cH$ are the positive semidefinite operators in $L(\cH)$ that have unit trace, we call these density matrices, and use the notation $\rho$ or $\sigma$ to stand for density matrices when possible. Quantum channels or quantum operations acting on quantum states over $\cH$ are completely positive trace preserving (CPTP) linear maps from $L(\cH)$ to $L(\cH')$ where $\cH'$ is any other finite dimensional Hilbert space. 
We use the trace distance, denoted by $\trD{\rho}{\sigma}$, as our distance measure on quantum states, 
$$\trD{\rho}{\sigma} = \frac{1}{2} \tr \left[\sqrt{\left(\rho-\sigma\right)^\dagger\left(\rho-\sigma\right)}\right] $$

A state over $\cH=\mathbb{C}^2$ is called a qubit. For any $n \in \mathbb{N}$, we refer to the quantum states over $\cH = (\mathbb{C}^2)^{\otimes n}$ as $n$-qubit quantum states. To perform a standard basis measurement on a qubit means projecting the qubit into $\{\ket{0},\ket{1}\}$. A quantum register is a collection of qubits. A classical register is a quantum register that is only able to store qubits in the computational basis.

A unitary quantum circuit is a sequence of unitary operations (unitary gates) acting on a fixed number of qubits. Measurements in the standard basis can be performed at the end of the unitary circuit. A (general) quantum circuit is a unitary quantum circuit with $2$ additional operations: $(1)$ a gate that adds an ancilla qubit to the system, and $(2)$ a gate that discards (trace-out) a qubit from the system. A quantum polynomial-time algorithm (QPT) is a uniform collection of quantum circuits $\{C_n\}_{n \in \mathbb{N}}$. We always assume that the QPT adversaries are non-uniform -- a QPT adversary $\cA$ acting on $n$ qubits could be given a quantum auxiliary state with $\poly(n)$ qubits.

\paragraph{Quantum Computational Indistinguishability.}

When we talk about quantum distinguishers, we need the following definitions, which we take from \cite{Wat09}.
\begin{definition}[Indistinguishable collections of states] Let $I$ be an infinite subset $I \subset \{0,1\}^*$, let $p : \mathbb{N} \rightarrow \mathbb{N}$ be a polynomially bounded function, and let $\rho_{x}$ and $\sigma_x$ be $p(|x|)$-qubit states. We say that $\{\rho_{x}\}_{x \in I}$ and $\{\sigma_x\}_{x\in I}$ are \textbf{quantum computationally indistinguishable collections of quantum states} if for every QPT $\cE$ that outputs a single bit, any polynomially bounded  $q:\mathbb{N}\rightarrow \mathbb{N}$, and any auxiliary $q(|x|)$-qubits state $\nu$, and for all $x \in I$, we have that
$$\left|\Pr\left[\cE(\rho_x\otimes \nu)=1\right]-\Pr\left[\cE(\sigma_x \otimes \nu)=1\right]\right| \leq \epsilon(|x|) $$
for some function $\epsilon:\mathbb{N}\rightarrow [0,1]$.  We use the following notation 
$$\rho_x \approx_{Q,\epsilon} \sigma_x$$
and we ignore the $\epsilon$ when it is understood that it is a negligible function.
\end{definition}

\paragraph{Quantum Fourier Transform and Subspaces.} Our main construction uses the same type of quantum states (superpositions over linear subspaces) considered by\cite{aaronson2012quantum,Zha19} in the context of constructing quantum money. 

We recall some key facts from these works relevant to our construction. 
Consider the field $\Zq^{\secparam}$ where $q\geq2$,and let $\ft$ denote the quantum fourier transfrom over $\Zq^{\secparam}$.

For any linear subspace $A$, let $A^\perp$ denote its orthogonal (dual) subspace, $$A^\perp = \{v \in \Zq^\secparam| \langle v, a\rangle = 0 \}.$$

\noindent Let $\ket{A} = \frac{1}{\sqrt{|A|}} \underset{a\in A}{\sum}\ket{a}$.  The quantum fourier Transform, $\ft$, does the following:
$$\ft \ket{A} = \ket{A^\perp}. $$
Since $(A^\perp)^\perp = A$, we also have $\ft \ket{A^\perp} = \ket{A}$.

Let $\Pi_{A}=\underset{a \in A}{\sum} \ket{a}\bra{a}$, then as shown in Lemma 21 of~\cite{aaronson2012quantum},
$$\ft (\Pi_{A^\perp}) \ft \Pi_{A} = \ket{A}\bra{A}.$$

\paragraph{{\bf Gentle Measurement Lemma/Almost As Good As New Lemma}.} We use the Almost As Good As New Lemma \cite{aaronson2004limitations} (also referred to as gentle measurement lemma in the literature), restated here verbatim from \cite{aaronson2016complexity}.

\begin{lemma}[Almost As Good As New]\label{clm:ASGAN} Let $\rho$ be a mixed state acting on $\mathbb{C}^d$. Let $U$ be a unitary and $(\Pi_0,\Pi_1=1-\Pi_0)$ be projectors all acting on $\mathbb{C}^d \otimes \mathbb{C}^d$. We interpret $(U,\Pi_0,\Pi_1)$ as a measurement performed by appending an acillary system of dimension $d'$ in the state $\ket{0}\bra{0}$, applying $U$ and then performing the projective measurement $\{\Pi_0,\Pi_1\}$ on the larger system. Assuming that the outcome corresponding to $\Pi_0$ has probability $1-\varepsilon$, i.e., $\tr[\Pi_0(U\rho \otimes \ket{0}\bra{0}U^\dagger)]=1-\varepsilon$, we have $$\trD{\rho}{\widetilde{\rho}} \leq \sqrt{\varepsilon} ,$$
where $\widetilde{\rho}$ is state after performing the measurement and then undoing the unitary $U$ and tracing out the ancillary system: $$\widetilde{\rho} = \tr_{d'}\left(U^\dagger \left( \Pi_0U \left( \rho \otimes \ket{0}\bra{0} \right)U^\dagger \Pi_0 + \Pi_1U \left( \rho \otimes \ket{0}\bra{0} \right)U^\dagger \Pi_1\right) U \right) $$
\end{lemma}

We use this Lemma to argue that whenever a QPT algorithm $\cA$ on input $\rho$, outputs a particular bit string $z$ with probability $1-\varepsilon$, then $\cA$ can be performed in a way that also lets us recover the initial state. In particular, given the QPT description for $\cA$, we can implement $\cA$ with an acillary system, a unitary, and only measuring in the computational basis after the unitary has been applied, similarly to Lemma~\ref{clm:ASGAN}. Then, it is possible to uncompute in order to also obtain $\widetilde{\rho}$.

\paragraph{Notation about Quantum-Secure Classical Primitives.} For a classical primitive X, we use the notation q-X to denote the fact that we assume X to be secure against QPT adversaries. 

\subsection{Learning with Errors}
\label{sec:prelims:lwe}

\noindent We consider the decisional learning with errors (LWE) problem, introduced by Regev~\cite{Reg05}. We define this problem formally below. 

\begin{quote}
    The problem $(n,m,q,\chi)$-LWE, where $n,m,q \in \mathbb{N}$ and  $\chi$ is a distribution supported over $\mathbb{Z}$, is to distinguish between the distributions $(\bfA,\bfA \bfs + \bfe)$ and $(\bfA,\bfu)$, where $\bfA \xleftarrow{\$} \mathbb{Z}_q^{m \times n},\bfs \xleftarrow{\$} \mathbb{Z}_q^{n \times 1},\bfe \xleftarrow{\$} \chi^{m \times 1}$ and $\bfu \leftarrow \mathbb{Z}_q^{m \times 1}$.
\end{quote}

\noindent The above problem believed to be hard against classical PPT algorithms -- also referred to as LWE assumption -- has had many powerful applications in cryptography. In this work, we conjecture the above problem to be hard even against QPT algorithms; this conjecture referred to as QLWE assumption has been useful in the constructions of interesting primitives such as quantum fully-homomorphic encryption~\cite{mahadev2018classical,Bra18}. We refer to this assumption as QLWE assumption. 

\begin{quote}
    {\bf QLWE assumption}: This assumption is parameterized by $\secparam$. Let $n=\poly(\secparam)$, $m=\poly(n \cdot \log(q))$ and $\chi$ be a discrete Gaussian distribution\footnote{Refer~\cite{Bra18} for a definition of discrete Gaussian distribution.} with parameter $\alpha q > 0$, where $\alpha$ can set to be any non-negative number.  
    \par Any QPT distinguisher (even given access to polynomial-sized advice state) can solve $(n,m,q,\chi)$-LWE only with  probability $\negl(\secparam)$, for some negligible function $\negl$. 
\end{quote}

\begin{remark}
We drop the notation $\secparam$ from the description of the assumption when it is clear. 
\end{remark}

\noindent $(n,m,q,\chi)$-LWE is shown~\cite{Reg05,PRS17} to be as hard as approximating shortest independent vector problem (SIVP) to within a factor of $\gamma=\tilde{O}(n/\alpha)$ (where $\alpha$ is defined above). The best known quantum algorithms for this problem run in time $2^{\tilde{O}(n/\log(\gamma))}$. 
\par For our construction of SSL, we require a stronger version of QLWE that is secure even against sub-exponential quantum adversaries. We state this assumption formally below. 

\begin{quote}
    {\bf $T$-Sub-exponential QLWE Assumption}: This assumption is parameterized by $\secparam$ and time $T$. Let $n=T+\poly(\secparam)$, $m=\poly(n \cdot \log(q))$ and $\chi$ be a discrete Gaussian distribution with parameter $\alpha q > 0$, where $\alpha$ can set to be any non-negative number.  
    \par Any quantum distinguisher (even given access to polynomial-sized advice state) running in time $2^{\widetilde{O}(T)}$ can solve $(n,m,q,\chi)$-LWE only with probability $\negl(\secparam)$, for some negligible function $\negl$. 
\end{quote}


\subsection{Quantum Fully Homomorphic Encryption}
\label{sec:prelims:qfhe}

A fully homomorphic encryption scheme allows for publicly evaluating an encryption of $x$ using a function $f$ to obtain an encryption of $f(x)$. Traditionally $f$ has been modeled as classical circuits but in this work, we consider the setting when $f$ is modeled as quantum circuits and when the messages are quantum states. This notion is referred to as quantum fully homomorphic encryption (QFHE). We state our definition verbatim from \cite{broadbent2015quantum}.

\begin{definition} Let $\cM$ be the Hilbert space associated with the message space (plaintexts), $\cC$ be the Hilbert space associated with the ciphertexts, and $\cR_{evk}$ be the Hilbert space associated with the evaluation key. A quantum fully homomorphic encryption scheme is a tuple of QPT algorithms  $\qfhe=(\gen,\enc,\dec,\allowbreak \eval)$ satisfying 
\begin{itemize}
    \item $\qfhe.\gen(1^\secparam)$: outputs a a public and a secret key, $(\pk,\sk)$, as well as a quantum state $\rho_{evk}$, which can serve as an evaluation key. 
    \item $\qfhe.\enc(\pk,\cdot):L(\cM)\rightarrow L(\cC)$:  takes as input a state $\rho$ and outputs a ciphertext $\sigma$
    \item $\qfhe.\dec(\sk,\cdot):L(\cC)\rightarrow L(\cM)$: takes a quantum ciphertext $\sigma$, and outputs a qubit $\rho$ in the message space $L(\cM)$.
    \item $\qfhe.\eval(\cE, \cdot ):L(\cR_{evk}\otimes \cC^{\otimes n})\rightarrow L(\cC^{\otimes m})$: takes as input a quantum circuit $\cE: L(\cM^{\otimes n}) \rightarrow L(\cM^{\otimes m})$, and a ciphertext in $L(\cC^{\otimes n})$ and outputs a ciphertext in $L(\cC^{\otimes m})$, possibly consuming the evaluation key $\rho_{evk}$ in the proccess.

\end{itemize}
\end{definition}
\noindent Semantic security and compactness are defined analogously to the classical setting, and we defer to~\cite{broadbent2015quantum} for a definition.
\noindent For the impossibility result, we require a $\qfhe$ scheme where ciphertexts of classical plaintexts are also classical. Given any $x \in \{0,1\}$, we want $\qfhe.\enc_\pk(\ket{x}\bra{x})$ to be a computational basis state $\ket{z}\bra{z}$ for some $z \in \{0,1\}^l$ (here, $l$ is the length of ciphertexts for 1-bit messages). In this case, we write $\qfhe.\enc_\pk(x)$. We also want the same to be true for evaluated ciphertexts, i.e. if $\cE(\ket{x}\bra{x})=\ket{y}\bra{y}$ for some $x \in \{0,1\}^n$ and $y \in \{0,1\}^m$, then 
$$\qfhe.\enc_\pk(y) \leftarrow \qfhe.\eval(\rho_{evk}, \cE, \qfhe.\enc_\pk(x)) $$
is a classical ciphertext of $y$.

We also need the property that classical ciphertexts can be efficiently decrypted with a classical algorithm.

\paragraph{Instantiation.} The works of~\cite{mahadev2018classical,Bra18} give lattice-based candidates for quantum fully homomorphic encryption schemes; we currently do not know how to base this on learning with errors alone\footnote{Brakerski~\cite{Bra18} remarks that the security of their candidate can be based on a circular security assumption that is also used to argue the security of existing constructions of unbounded depth multi-key FHE~\cite{CM15,MW16,PS16,BP16}.}. The desirable property required from the quantum FHE schemes, that classical messages have classical ciphertexts, is satisfied by both candidates~\cite{mahadev2018classical,Bra18}.

\subsection{Obfuscation}
In this work, we use different notions of cryptographic obfucation. We review all the required notions below, but first we recall the functionality of obfuscation. 

\begin{definition}[Functionality of Obfuscation]
Consider a class of circuits $\cktclass$. An obfuscator $\obfuscator$ consists of two PPT algorithms $\obf$ and $\eval$ such that the following holds: for every $\secparam \in \mathbb{N}$, circuit $C \in \cktclass$, $x \in \{0,1\}^{\poly(\secparam)}$, we have $C(x) \leftarrow \eval(\widetilde{C},x)$ where $\widetilde{C} \leftarrow \obf(1^{\secparam},C)$.  
\end{definition}

\subsubsection{Lockable Obfuscation}
\label{ssec:prelims:lobfs}

\noindent In the impossibility result, we will make use of program obfuscation schemes that are (i) defined for compute-and-compare circuits and, (ii) satisfy distributional virtual black box security notion~\cite{BGIRSVY01}. Such obfuscation schemes were first introduced by~\cite{WZ17,GKW17} and are called lockable obfuscation schemes. We recall their definition, adapted to quantum security, below.

\begin{definition}[Quantum-Secure Lockable Obfuscation]
An obfuscation scheme $(\lobf,\leval)$ for a class of circuits $\cktclass$ is said to be a \textbf{quantum-secure lockable obfuscation scheme} if the following properties are satisfied: 
\begin{itemize}
    \item It satisfies the functionality of obfuscation. 
    \item {\bf Compute-and-compare circuits}: Each circuit $\lockC$ in $\cktclass$ is parameterized by strings $\alpha \in \{0,1\}^{\poly(\secparam)},\beta \in \{0,1\}^{\poly(\secparam)}$ and a poly-sized circuit $C$ such that on every input $x$, $\lockC(x)$ outputs $\beta$ if and only if $C(x)=\alpha$. 
    \item {\bf Security}: For every polynomial-sized circuit $C$, string $\beta \in \{0,1\}^{\poly(\secparam)}$,for every QPT adversary $\adversary$ there exists a QPT simulator $\simulator$ such that the following holds: sample $\alpha \xleftarrow{\$}  \{0,1\}^{\poly(\secparam)}$,
    $$\left\{ \lobf \left( 1^{\secparam},\lockC \right) \right\} \approx_{Q,\varepsilon} \left\{\simulator\left(1^{\secparam},1^{|C|} \right) \right\},$$
    where $\lockC$ is a circuit parameterized by $C,\alpha,\beta$ with $\varepsilon \leq \frac{1}{2^{|\alpha|}}$.  
\end{itemize}

\end{definition}

\paragraph{Instantiation.} The works of~\cite{WZ17,GKW17,GKVW19} construct a lockable obfuscation scheme based on polynomial-security of learning with errors (see Section~\ref{sec:prelims:lwe}). Since learning with errors is conjectured to be hard against QPT algorithms, the obfuscation schemes of~\cite{WZ17,GKW17,GKVW19} are also secure against QPT algorithms. Thus, we have the following theorem. 

\begin{theorem}[\cite{GKW17,WZ17,GKVW19}]
Assuming quantum hardness of learning with errors, there exists a quantum-secure lockable obfuscation scheme. 
\end{theorem}

\subsubsection{q-Input-Hiding Obfuscators}
\noindent One of the main tools used in our construction is q-input-hiding obfuscators. The notion of input-hiding obfuscators was first defined in the classical setting by Barak et al.~\cite{BBCKP14}. We adopt the same notion except that we require the security of the primitive to hold against QPT adversaries. \par The notion of q-input-hiding obfuscators states that given an obfuscated circuit, it should be infeasible for a QPT adversary to find an accepting input; that is, an input on which the circuit outputs 1. Note that this notion is only meaningful for the class of evasive circuits. 

The definition below is suitably adapted from Barak et al.~\cite{BBCKP14}; in particular, our security should hold against QPT adversaries. 

\begin{definition}[q-Input-Hiding Obfuscators~\cite{BBCKP14}]
An obfuscator $\qiho=(\obf,\eval)$ for a class of circuits associated with distribution $\distrc$ is {\bf q-input-hiding} if for every non-uniform QPT adversary $\adversary$, for every sufficiently large $\secparam \in \mathbb{N}$,
$$\prob \left[C(x)=1\ :\ \substack{C \leftarrow \distrc(\secparam),\\ \ \\ \widetilde{C} \leftarrow \obf(1^{\secparam},C), \\ \ \\ x \leftarrow \adversary(1^{\secparam},\widetilde{C})} \right] \leq \negl(\secparam).$$
\end{definition}

\subsubsection{Subspace Hiding Obfuscators}~\label{ssec:shO}
\noindent Another ingredient in our construction is subspace hiding obfuscation. Subspace hiding obfuscation is a notion of obfuscation introduced by Zhandry \cite{Zha19}, as a tool to build pulic-key quantum money schemes. This notion allows for obfuscating a circuit, associated with subspace $A$, that checks if an input vector belongs to this subspace $A$ or not. In terms of security, we require that the obfuscation of this circuit is indistinguishable from obfuscation of another circuit that tests membership of a larger random (and hidden) subspace containing $A$. 

\begin{definition}[\cite{Zha19}]
A subspace hiding obfuscator for a field $\mathbb{F}$ and dimensions $d_0,d_1,\secparam$ is a tuple $(\shO.\obf, \shO.\eval)$ satisfying:
\begin{itemize}
\item  $\shO.\obf(A)$: on input an efficient description of a linear subspace $A \subset \mathbb{F}^\secparam$ of dimensions $d \in \{d_0,d_1\}$ outputs an obfuscator $\shO(A)$. 
\item {\bf Correctness:} For any $A$ of dimension $d \in \{d_0,d_1\}$, it holds that
$$\prob[\forall x, \shO.\eval(\shO.\obf(A),x)=\mathbb{1}_{A}(x) :  \shO(A) \leftarrow \shO.\obf(A)] \geq 1-\negl(\secparam),$$
where: $\mathbb{1}_{A}(x)=1$ if $x \in A$ and 0, otherwise. 
\item {\bf Quantum-Security:}  Any QPT adversary $\cA$ can win the following challenge with probability at most negligibly greater than $\frac{1}{2}$.
\begin{enumerate}
    \item $\cA$ chooses a $d_0$-dimensional subspace $A \subset \mathbb{F}^\secparam$. 
    \item Challenger chooses uniformly at random a $d_1$-dimensional subspace $S \supseteq A$. It samples a random bit $b$.  If $b=0$, it sends $\widetilde{g_0} \leftarrow \shO.\obf(A)$. Otherwise, it sends $\widetilde{g_1} \leftarrow \shO.\obf(S)$
    \item $\cA$ receives $\widetilde{g_b}$ and outputs $b'$.  It wins if $b'=b$.
\end{enumerate}
\end{itemize}
\end{definition}

\paragraph{Instantiation.}  Zhandry presented a construction of subspace obfuscators from  indistinguishability obfuscation~\cite{BGIRSVY01,GGHRSW16} secure against QPT adversaries. 

\subsection{q-Simulation-Extractable Non-Interactive Zero-Knowledge}
We also use the tool of non-interactive zero-knowledge (NIZK) systems for NP for our construction. A NIZK is defined between a classical PPT prover $\prover$ and a verifier $\verify$. The goal of the prover is to convince the verifier $\verify$ to accept an instance $x$ using a witness $w$ while at the same time, not revealing any information about $w$. Moreover, any malicious prover should not be able to falsely convince the verifier to accept a NO instance. Since we allow the malicious parties to be QPT, we term this NIZK as qNIZK. 
\par We require the qNIZKs to satisfy a stronger property called simulation extractability and we call a qNIZK satisfying this stronger property to be q-simulation-extractable NIZK ($\ssnizk$). 
\par We describe the PPT algorithms of $\ssnizk$ below.   
\begin{itemize}
    \item $\crsgen(1^{\secparam})$: On input security parameter $\secparam$, it outputs the common reference string $\crs$. 
    \item $\prover(\crs,x,w)$: On input common reference string $\crs$, NP instance $x$, witness $w$, it outputs the proof $\pi$. 
    \item $\verify(\crs,x,\pi)$: On input common reference string $\crs$, instance $x$, proof $\pi$, it outputs 1 (accept) or 0 (reject). This is a deterministic algorithm.
\end{itemize}
\noindent This notion is associated with the following properties. We start with the standard notion of completeness.

\begin{definition}[Completeness]
\label{def:qnizk:compl}
A non-interactive protocol $\ssnizk$ for a NP language $L$ is said to be {\bf complete} if the following holds: for every $(x,w) \in \rel(L)$, we have the following: 
$$\prob \left[ \verify(\crs,x,\pi)\ \text{accepts}\ :\ \substack{\crs \leftarrow \crsgen(1^{\secparam})\\ \ \\ \pi \leftarrow \prover(\crs,x,w)} \right] = 1$$
\end{definition}

\paragraph{q-Simulation-Extractability.} We now describe the simulation-extractability property. Suppose there exists an adversary who upon receiving many proofs $\pi_1,\ldots,\pi_q$ on all YES instances $x_1,\ldots,x_q$, can produce a proof  $\pi'$ on instance $x'$ such that: (a) $x'$ is different from all the instances $x_1,\ldots,x_q$ and, (b) $\pi'$ is accepting with probability $\varepsilon$. Then, this notion guarantees the existence of two efficient algorithms $\simr_1$ and $\simr_2$ 
such that all the proofs $\pi_1,\ldots,\pi_q$, are now simulated by $\simr_1$, and $\simr_2$ can extract a valid witness for $x'$ from $(x',\pi')$ produced by the adversary with probability negligibly close to $\varepsilon$.     

\begin{definition}[q-Simulation-Extractability]
A non-interactive protocol $\ssnizk$ for a language $L$ is said to satisfy {\bf  q-simulation-extractability}  if there exists a non-uniform QPT adversary $\adversary=(\adversary_1,\adversary_2)$ such that the following holds:
     $$\prob \left[ \substack{\verify(\crs,x',\pi')\ \text{accepts}\\ \ \\ \bigwedge\\ \ \\ \left(\forall i\in [q], \left(x_i,w_i \right) \in \rel(\lang) \right)\\ \ \\ \bigwedge\\ \ \\ \left( \forall i \in [q],\ x' \neq x_i  \right)} \ :\ \substack{\crs \leftarrow \crsgen(1^{\secparam}),\\ \ \\ \left( \{(x_i,w_i)\}_{i \in [q]},\st_{\adversary} \right) \leftarrow \adversary_1(\crs)\\ \ \\  \forall i \in [q], \ \pi_i \leftarrow \prover(\crs,\td,x_i)\\ \ \\ (x',\pi') \leftarrow \adversary_2(\st_{\adversary},\pi_1,\ldots,\pi_q)} \right] = \varepsilon$$
    Then there exists QPT algorithms $\fksetup$ and   $\simr=(\simr_1,\simr_2)$ such that the following holds: 
    $$\left|\prob \left[ \substack{\verify(\crs,x',\pi')\ \text{accepts}\\ \ \\ \bigwedge\\ \ \\ \left(\forall i\in [q], \left(x_i,w_i \right) \in \rel(\lang) \right)\\ \ \\ \bigwedge\\ \ \\ (x',w') \in \rel(L)\\ \ \\ \bigwedge\\ \ \\ \left( \forall i \in [q],\ x' \neq x_i  \right)} \ :\ \substack{(\crs,\td) \leftarrow \fksetup(1^{\secparam}),\\ \ \\ \left( \{(x_i,w_i)\}_{i \in [q]},\st_{\adversary} \right) \leftarrow \adversary_1(\crs)\\ \ \\   (\pi_1,\ldots,\pi_{q},\st_{\simr}) \leftarrow \simr_1 \left( \crs,\td,\{x_i\}_{i \in [q]} \right)\\ \ \\ (x',\pi') \leftarrow \adversary_2(\st_{\adversary},\pi_1,\ldots,\pi_q)\\ \ \\ w' \leftarrow \simr_2(\st_{\simr},x',\pi')} \right]\ -\ \varepsilon \right|\leq \negl(\secparam)$$
\noindent We call a non-interactive argument system satisfying q-simulation-extractability property to be a qseNIZK system. 
\par If q-smulation-extractability property holds against quantum adversaries running in time $2^{\tilde{O}(T)}$ ($\tilde{O}(\cdot)$ notation suppresses additive factors in $O(\log(\secparam))$) then we say that $(\crsgen,\prover,\verify)$ is a $T$-sub-exponential qseNIZK system. 
\end{definition}


\begin{remark}
The definition as stated above is weaker compared to other definitions of simulation-extractability considered in the literature. For instance, we can consider general adversaries who also can obtain simulated proofs for false statements which is disallowed in the above setting. Nonetheless, the definition considered above is sufficient for our application. 
\end{remark}

\paragraph{Instantiation of qseNIZKs.} In the classical setting, simulation-extractable NIZKs can be obtained by generically~\cite{Sah99,DeS+01} combining a traditional NIZK (satisfying completeness, soundness and zero-knowledge) with a public-key encryption scheme satisfying CCA2 security. We observe that the same transformation can be ported to the quantum setting as well, by suitably instantiating the underlying primitives to be quantum-secure. These primitives in turn can be instantiated from QLWE. Thus, we can obtain a q-simulation-extractable NIZK from QLWE. 
\par For our construction of SSL, it turns out that we need a q-simulation-extractable NIZK that is secure against quantum adversaries running in sub-exponential time. Fortunately, we can still adapt the same transformation but instead instantiating the underlying primitives to be sub-exponentially secure. \par Before we formalize this theorem, we first state the necessary preliminary background.

\begin{definition}[q-Non-Interactive Zero-Knowledge]
A non-interactive system $(\crsgen,\prover,\verify)$ defined for a NP language $\lang$ is said to be {\bf q-non-interactive zero-knowledge (qNIZK)} if it satisfies Definition~\ref{def:qnizk:compl} and additionally, satisfies the following properties: 
\begin{itemize}
    \item Adaptive Soundness: For any  malicious QPT prover $\prover^*$, the following holds: 
    $$\prob \left[ \substack{\verify(\crs,x,\pi)\  \text{accepts}\\ \ \\ \bigwedge\\ \ \\ x' \notin \lang} \ \ :\ \substack{\crs \leftarrow \crsgen\left(1^{\secparam} \right)\\ \ \\ (x,\pi) \leftarrow \prover^*(\crs)} \right] \leq \negl(\secparam)$$
    \item Adaptive (Multi-Theorem) Zero-knowledge: For any QPT verifier $\verify^*$, there exists two QPT algorithms $\fkgen$ and simulator $\simr$, such that the following holds: 
    $$\Bigg| \prob\left[ \substack{1 \leftarrow \verify^*\left(\st,\{\pi\}_{i \in [q]} \right)\\ \ \\ \bigwedge\\ \ \\ \forall i\in [q],\ (x_i,w_i) \in \rel(\lang)} \ :\ \substack{\crs \leftarrow \crsgen(1^{\secparam})\\ \ \\ \left(\{(x_i,w_i)\}_{i \in [q]},\st \right) \leftarrow \verify^*(\crs)\\ \ \\ \forall i \in [q],\ \pi_i \leftarrow \prover(\crs,x_i,w_i)} \right]$$
    $$\ \ \ \ \ \ \ \ \ \ \ \ \ \ \ \ \ \ \ - \prob\left[ \substack{1 \leftarrow \verify^*\left(\st,\{\pi\}_{i \in [q]} \right)\\ \ \\ \bigwedge\\ \ \\ \forall i\in [q],\ (x_i,w_i) \in \rel(\lang)} \ :\ \substack{(\crs,\td) \leftarrow \fkgen(1^{\secparam})\\ \ \\ \left(\{(x_i,w_i)\}_{i \in [q]},\st \right) \leftarrow \verify^*(\crs)\\ \ \\ \{\pi_i\}_{i \in [q]} \leftarrow \simr(\crs,\td,\{x_i\}_{i \in [q]})} \right] \Bigg| \leq \negl(\secparam) $$
\end{itemize}
If both adaptive soundness and adaptive multi-theorem zero-knowledge holds against quantum adversaries running in time $2^{\tilde{O}(T)}$ then we say that $(\crsgen,\prover,\verify)$ is a $T$-sub-exponential qNIZK. 
\end{definition}

\begin{remark}
q-simulation-extractable NIZKs imply qNIZKs since simulation-extractability implies both soundness and zero-knowledge properties.  
\end{remark}

\begin{definition}[q-CCA2-secure PKE]
A public-encryption scheme $(\setup,\enc,\dec)$ (defined below) is said to satify {\bf q-CCA2-security} if every QPT adversary $\cA$ wins in $\expt_{\cA}$ (defined below) only with negligible probability. 
\begin{itemize}
    \item $\setup(1^{\secparam})$: On input security parameter $\secparam$, output a public key $\pk$ and a decryption key $\sk$. 
    \item $\enc(\pk,x)$: On input public-key $\pk$, message $x$, output a ciphertext $\ct$. 
    \item $\dec(\sk,\ct)$: On input decryption key $\sk$, ciphertext $\ct$, output $y$. 
\end{itemize}
For any $x \in \{0,1\}^{\poly(\secparam)}$, we have $\dec(\sk,\enc(\pk,x))=x$. \\

\noindent \underline{$\expt_{\cA}(1^{\secparam},b)$}: 
\begin{itemize}
    \item Challenger generates $\setup(1^{\secparam})$ to obtain $(\pk,\sk)$. It sends $\pk$ to $\cA$. 
    \item $\cA$ has (classical) access to a decryption oracle that on input $\ct$, outputs $\dec(\sk,\ct)$. It can make polynomially many queries. 
    \item $\cA$ then submits $(x_0,x_1)$ to the challenger which then returns $\ct^* \leftarrow \enc(\pk,x_b)$.
    \item $\cA$ is then given access to the same oracle as before. The only restriction on $\cA$ is that it cannot query $\ct^*$.
    \item Output $b'$ where the output of $\cA$ is $b'$. 
\end{itemize}
$\cA$ wins in $\expt_{\cA}$ with probability $\mu(\secparam)$ if $\prob\left[ b=b'\ : \substack{b \xleftarrow{\$} \{0,1\}\\ \ \\ \expt_{\cA}(1^{\secparam})} \right]=\frac{1}{2} + \mu(\secparam)$.  
\par If the above q-CCA2 security holds against quantum adveraries running in time $2^{\tilde{O}(T)}$  then we say that $(\setup,\enc,\dec)$ is a  $T$-sub-exponential  q-CCA2-secure PKE scheme. 
\end{definition}

\begin{remark}
One could also consider the setting when the CCA2 adversary has superposition access to the oracle. However, for our construction, it suffices to consider the setting when the adversary only has classical access to the oracle.  
\end{remark}

\noindent Consider the following lemma. 

\begin{lemma}
\label{lem:qsenizks}
Consider a language $\lang_{\ell} \in NP$ such that every $x \in \lang_{\ell}$ is such that $|x|=\ell$.  
\par Under the $\ell$-sub-exponential QLWE assumption, there exists a q-simulation-extractable NIZKs for $\lang_{\ell}$ satisfying perfect completeness. 
\end{lemma}
\begin{proof}
We first state the following proposition that shows how to generically construct a q-simulation-extractable NIZK from qNIZK and a CCA2-secure public-key encryption scheme.

\begin{proposition}
Consider a language $\lang_{\ell} \in NP$ such that every $x \in \lang_{\ell}$ is such that $|x|=\ell$. 
\par Assuming $\ell$-sub-exponential qNIZKs for NP and $\ell$-sub-exponential q-CCA2-secure PKE schemes, there exists a $\ell$-sub-exponential qseNIZK system for $\lang_{\ell}$.  
\end{proposition}
\begin{proof}

Let $\qpke$ be a $\ell$-sub-exponential qCCA2-secure PKE scheme. Let $\qnizk$ be a $\ell$-sub-exponential qNIZK for the following relation. 
$$\rel_{\qnizk}=\left\{ \left((\pk,\ct_w,x),\ (w,r_w) \right)\ :\ \left( (x,w) \in \rel(\lang_{\ell}) \bigwedge \  \ct_w = \enc(\pk,(x,w);r_w) \right)  \right\}$$

\noindent We present the construction (quantum analogue of~\cite{Sah99,DeS+01}) of q-simulation-extractable NIZK for $\lang_{\ell}$ below. 

\begin{itemize}
    \item $\crsgen(1^{\secparam})$: On input security parameter $\secparam$, \begin{itemize}
        \item Compute $\qnizk.\crsgen(1^{\secparam_1})$ to obtain $\qnizk.\crs$, where $\secparam_1=\poly(\secparam,\ell)$ is chosen such that $\qnizk$ is a $\ell$-sub-exponential q-non-interactive zero-knowledge argument system. 
        \item Compute $\pke.\setup(1^{\secparam_2})$ to obtain $(\pk,\sk)$, where $\secparam_2=\poly(\secparam,\ell)$ is chosen such that $\qpke$ is a $\ell$-sub-exponential q-CCA2-secure PKE scheme. 
    \end{itemize}
    Output $\crs=(\pk,\qnizk.\crs)$. 
    \item $\prover(\crs,x,w)$: On input common reference string $\crs$, instance $x$, witness $w$, 
    \begin{itemize}
        \item Parse $\crs$ as $(\pk,\qnizk.\crs)$. 
        \item Compute $\ct_w \leftarrow \pke.\enc(\pk,(x,w);r_w)$, where $r_w \xleftarrow{\$} \{0,1\}^{\poly(\secparam)}$.
        \item Compute $\qnizk.\pi \leftarrow \qnizk.\prover(\qnizk.\crs,(\pk,\ct_w,x),(w,r_w))$. 
    \end{itemize}
    \noindent Output $\pi=\left( \qnizk.\pi,\ct_w \right)$. 
    \item $\verify(\crs,x,\pi)$: On input common reference string $\crs$, NP instance $x$, proof $\pi$, 
    \begin{itemize}
        \item Parse $\crs$ as $(\pk,\ct,\qnizk.\crs)$.
        \item Output $\qnizk.\verify \left( \qnizk.\crs,(\pk,\ct_w,x),\pi \right)$. 
    \end{itemize}
\end{itemize}

\noindent We prove that the above argument system satisfies q-simulation-extractability. We describe the algorithms $\fksetup$ and $\simr=(\simr_1.\simr_2)$ below. Let $\qnizk.\fkgen$ and $\qnizk.\simr$ be the QPT algorithms associated with the zero-knowledge property of $\qnizk$.  \\

\noindent \underline{$\fksetup(1^{\secparam})$}: Compute $(\qnizk.\crs,\tau) \leftarrow \qnizk.\fkgen \left(1^{\secparam} \right)$. Compute $(\pk,\sk) \leftarrow \pke.\setup(1^{\secparam})$. Output $\crs=\left( \qnizk.\crs,\pk,\ct \right)$ and $\td=(\tau,\sk)$. \\

\noindent \underline{$\simr_1 \left( \crs,\td,\{x_i\}_{i \in [q]} \right)$}: Compute $\qnizk.\simr \left( \qnizk.\crs,\tau,(\pk,\ct,x_i) \right)$ to obtain $\qnizk.\pi_i$, for every $i \in [q]$. Output $\left\{\qnizk.\pi_1,\ldots,\qnizk.\pi_q\right\}$ and $\st=\left(\td,\crs,\left( \left\{ x_i \right\}_{i \in [q]} \right) \right)$. \\

\noindent \underline{$\simr_2 \left(\st,x',\pi' \right)$}: On input $\st=\left(\td=(\tau,\sk),\crs,\left( \left\{ x_i \right\}_{i \in [q]} \right) \right)$, instance $x'$, proof $\pi'=(\qnizk.\pi',\ct'_w)$, compute $\dec(\sk,\ct'_{w'})$ to obtain $w'$. Output $w'$.  \\

\noindent Suppose $\cA$ be a quantum adversary running in time $2^{\widetilde{O}(\ell)}$ such that the following holds:  $$\prob \left[ \substack{\verify(\crs,x',\pi')\ \text{accepts}\\ \ \\ \bigwedge\\ \ \\ \left(\forall i\in [q],\ \left(x_i,w_i \right) \in \rel(\lang) \right)\\ \ \\ \bigwedge\\ \ \\ \left( \forall i \in [q],\ x' \neq x_i  \right)} \ :\ \substack{\crs \leftarrow \crsgen(1^{\secparam}),\\ \ \\ \left( \{(x_i,w_i)\}_{i \in [q]},\st_{\adversary} \right) \leftarrow \adversary_1(\crs)\\ \ \\  \forall i \in [q], \ \pi_i \leftarrow \prover(\crs,\td,x_i)\\ \ \\ (x',\pi') \leftarrow \adversary_2(\st_{\adversary},\pi_1,\ldots,\pi_q)} \right] = \varepsilon$$
\noindent Let $\delta$ be such that the following holds: 
 $$\prob \left[ \substack{\verify(\crs,x',\pi')\ \text{accepts}\\ \ \\ \bigwedge\\ \ \\ \left(\forall i\in [q], \left(x_i,w_i \right) \in \rel(\lang) \right)\\ \ \\ \bigwedge\\ \ \\ (x',w') \in \rel(L)\\ \ \\ \bigwedge\\ \ \\ \left( \forall i \in [q],\ x' \neq x_i  \right)} \ :\ \substack{(\crs,\td) \leftarrow \fksetup(1^{\secparam}),\\ \ \\ \left( \{(x_i,w_i)\}_{i \in [q]},\st_{\adversary} \right) \leftarrow \adversary_1(\crs)\\ \ \\   (\pi_1,\ldots,\pi_{q},\st_{\simr}) \leftarrow \simr_1 \left( \crs,\td,\{x_i\}_{i \in [q]} \right)\\ \ \\ (x',\pi') \leftarrow \adversary_2(\st_{\adversary},\pi_1,\ldots,\pi_q)\\ \ \\ w' \leftarrow \simr_2(\st_{\simr},x',\pi')} \right]\  = \delta$$
We prove using a standard hybrid argument that $|\delta-\varepsilon| \leq \negl(\secparam)$. \\

\noindent \underline{$\hybrid_1$}: $\cA$ is given $\pi_1,\ldots,\pi_q$, where $\pi_i \leftarrow \prover(\crs,x_i,w_i)$. Let $(x',\pi')$ is the output of $\cA$ and parse $\pi'=(\qnizk.\pi',\ct'_w)$. Decrypt $\ct'_w$ using $\sk$ to obtain $(x^*,w')$. 
\par From the adaptive soundness of $\qnizk$, the probability that $(x',w') \in \rel(\lang_{\ell})$ and $x^* = x'$ is negligibly close to $\varepsilon$. \\ 

\noindent \underline{$\hybrid_{2}$}: $\cA$ is given $\pi_1,\ldots,\pi_q$, where the proofs are generated as follows: first compute $(\qnizk.\pi_1,\allowbreak \ldots,\allowbreak \qnizk.\pi_{q}) \leftarrow \qnizk.\simr(\crs,\td,\{x_i\}_{i \in [q]})$, where $(\crs,\td) \leftarrow \qnizk.\fkgen(1^{\secparam})$. Then compute $\ct_{w_i} \leftarrow \enc(\pk,(x_i,w_i))$ for every $i \in [q]$. Set $\pi_i=(\qnizk.\pi_i,\ct_{w_i})$. The rest of this hybrid is defined as in $\hybrid_1$. 
\par From the adaptive zero-knowledge property of $\qnizk$, the probability that $(x',w') \in \rel(\lang_{\ell})$ and $x^* = x'$ in the hybrid $\hybrid_{2.j}$ is still negligibly close to $\varepsilon$. \\

\noindent \underline{$\hybrid_3$}: This hybrid is defined similar to the previous hybrid except that $\ct_{w_i} \leftarrow \enc(\pk,0)$, for every $i \in [q]$.  
\par From the previous hybrids, it follows that $\ct'_w \neq \ct_{w_i}$, for all $i\in [q]$ with probability negligibly close to $\varepsilon$; this follows from the fact that $\qpke$ is perfectly correct and the fact that $x^* = x'$ holds with probability negligibly close to $\varepsilon$. Thus, we can invoke  q-CCA2-security of $\pke$, the probability that $(x',w') \in \rel(\lang_{\ell})$ is still negligibly close to $\varepsilon$.  \\

\noindent But note that $\hybrid_3$ corresponds to the simulated experiment and thus we just showed that the probability that we can recover $w'$ such that $(x',w') \in \rel(\lang_{\ell})$ is negligibly close to $\varepsilon$. 

\end{proof}

\noindent The primitives in the above proposition can be instantiated from sub-exponential QLWE by starting with  existing LWE-based constructions of the above primitive and suitably setting the parameters of the underlying LWE assumption. We state the following propositions without proof. 

\begin{proposition}[\cite{PS19}]
Assuming $\ell$-sub-exponential QLWE (Section~\ref{sec:prelims:lwe}), there exists a $\ell$-sub-exponential qNIZK for NP.
\end{proposition}

\begin{remark}
To be precise, the work of~\cite{PS19} constructs a NIZK system satisfying adaptive multi-theorem zero-knowledge and non-adaptive soundness. However, non-adaptive soundness implies adaptive soundness using complexity leveraging; the reduction incurs a security loss of $2^{\ell}$. 
\end{remark}

\begin{proposition}[~\cite{PW11}]
Assuming $\ell$-sub-exponential QLWE (Section~\ref{sec:prelims:lwe}), there exists a $\ell$-sub-exponential q-CCA2-secure PKE scheme. 
\end{proposition}

\end{proof}

\section{Proof of Proposition~\ref{prop:quan:unlearn}}
\label{sec:unlearnable}
\noindent We start with the following claim. 

\begin{claim} \label{clm:unlearn}
For all QPT $\cA$ with oracle access to $C_{a,b,r,\pk,\cO}(\cdot)$ (where the adversary is allowed to make superposition queries), we have
$$\underset{(a,b,r,\pk,\cO) \leftarrow \distrc}{\prob}\left[\sk \leftarrow \cA^{C_{a,b,r,\pk,\cO}} \left(1^\secparam \right) \right] \leq \negl(\secparam)$$
\end{claim}
\begin{proof} 
\noindent Towards proving this, we make some simplifying assumptions; this is only for simplicity of exposition and they are without loss of generality.\\

\noindent {\em Simplifying Assumptions.} Consider the following oracle $O_{a,b,r,\pk,\cO}$:
$$ O_{a,b,r,\pk,\cO} \ket{x}\ket{z} = \left\{  \begin{array}{lcl} \ket{x}\ket{z \oplus C_{a,b,r,\pk,\cO}(x)}, & \text{if }x \neq 0 \cdots0\\
& \\
\ket{x}\ket{z}, & \text{if }x=0\cdots0
\end{array} \right. $$
The first simplifying assumption is that the adversary $\cA$ is given access to the oracle $O_{a,b,r,\pk,\cO}$, instead of the oracle $C_{a,b,r,\pk,\cO}$. In addition, $\cA$ is given $\enc(\pk,a;r)$, $\pk$, and $\cO$ as auxiliary input. 
\par The second simplifying assumption is that $\cA$ is given some auxiliary state $\ket{\xi}$, and that it only performs computational basis measurements right before outputting (i.e. $\cA$ works with purified states). \\

\noindent {\em Overview.} Our proof follows the adversary method proof technique~\cite{ambainis2002quantum}. We prove this by induction on the number of queries. We show that after every query the following invariant is maintained: the state of the adversary has little amplitude over $a$. More precisely, we argue that the state of the adversary after the $t^{th}$ query, is neglibly close to the state just before the $t^{th}$ query, denoted by $\ket{\psi^t}$. After the adversary obtains the response to the $t^{th}$ query, it then applies a unitary operation to obtain the state $\ket{\psi^{t+1}}$, which is the state of the adversary just before the $(t+1)^{th}$ query. This observation implies that there is another state $\ket{\phi^{t+1}}$ that: (a) is close to $\ket{\psi^{t+1}}$ (here, we use the inductive hypothesis that $\ket{\phi^{t}}$ is close to $\ket{\psi^{t}}$) and, (b) can be prepared without querying the oracle at all. \\

\noindent Let $U_i$ denote the unitary that $\cA$ performs right before its $i^{th}$ query, and let $\bfA,\sfX,$ and $\sfY$ denote the private, oracle input, and oracle output registers of $\cA$, respectively. 
  
  Just before the $t^{th}$ query, we denote the state of the adversary to be:

$$\ket{\psi^t}:=U_t O \cdots O U_1 \ket{\psi^0}$$
where $\ket{\psi^0} = \ket{\xi}\ket{\enc(\pk,a;r),\cO,\pk}\ket{0\cdots0}_{\sfX}\ket{0\cdots0}_{\sfY}$ is the initial state of the adversary. Let $\Pi_a = (\ket{a}\bra{a})_\sfX \otimes I_{\sfY,\sfA}$.\\ 

 Note that any $\cA$ that outputs $\sk$ with non-negligible probability can also query the oracle on a state $\ket{\psi}$ satisfying $\tr[\Pi_{a}\ket{\psi}\bra{\psi}] \geq \nonnegl(\secparam)$ with non-negligible probability. Since $\cA$ outputs $\sk$ with non-negligible probability, it can decrypt $\enc(\pk,a;r)$, to find $a$ and then query the oracle on $a$. In other words, if there is an adversary $\cA$ that finds $\sk$ with non-negligible probability, then there is an adversary that at some point queries the oracle with a state $\ket{\psi}$ $\tr[\Pi_{a}\ket{\psi}\bra{\psi}] \geq \nonnegl(\secparam)$ also with non-negligible probability. \\

\noindent Hence, it suffices to show that for any adversary $\cA$ that makes at most $T=\poly(\secparam)$ queries to the oracle, it holds that
$$\prob[\forall j,  \tr[\Pi_{a} \ket{\psi^j}\bra{\psi^j}] \leq \negl(\secparam)] \geq 1-\negl(\secparam).$$

\noindent This would then imply that $\cA$ cannot output $\sk$ with non-negligible probability, thus proving Claim~\ref{clm:unlearn}.

Towards proving the above statement, consider the following claim that states that if $\cA$ has not queried  the oracle with a state that has large overlap with $\Pi_{a}$, then its next query will also not have large overlap with $\Pi_{a}$. 
\begin{claim}[No Good Progress] Let $T$ be any polynomial in $\secparam$.
Suppose for all $t<T$, the following holds:
$$\tr \left[ \Pi_{a} \ket{\psi^t}\bra{\psi^t} \right] \leq \negl(\secparam) $$
\noindent Then, $\prob[\tr[\Pi_{a}\ket{\psi^T}\bra{\psi^T}]\leq \negl(\secparam)] \geq 1-\negl(\secparam)$.
\end{claim}
\begin{proof}
For all $j$, let $\ket{\phi^j}=U_jU_{j-1}...U_1\ket{\psi^0}$. 

We will proceed by induction on $T$. Our base case is $T=1$ (just before the first query to the oracle); that is, $\ket{\psi^1}=\ket{\phi^1}$. Suppose the following holds: $$\prob[\tr[\Pi_{a}\ket{\psi^1}\bra{\psi^1}]\geq \nonnegl(\secparam)]\geq \nonnegl(\secparam).$$

The first step is to argue that if $\cA$ can prepare a state such that $\tr[\Pi_a \ket{\psi_1}\bra{\psi_1}] \geq \nonnegl(\secparam)$ given $\enc(\pk,a;r), \pk$ and $\cO \leftarrow \lobf(\newbfckt \left[ \qfhe.\dec(\sk,\cdot),b,(\sk|r) \right])$ without querying the oracle, then it can also prepare a state with large overlap with $\Pi_a$ if its given the simulator of the lockable obfuscation instead. We will use $\cA$ (specifically, the first unitary that $\cA$ applies, $U_1$) to construct an adversary $\cB$ that breaks the security of lockable obfuscation. $\cB$ is given $a$, $\enc(\pk,a;r)$, $\pk$ and $\cO$ as well as auxiliary state $\ket{\xi}$.  It the prepares $\ket{\psi_{1,\cO}}= U_1 \ket{\xi}\ket{\enc(\pk,a;r),\cO,\pk}\allowbreak \ket{0\cdots0}_{\sfX}\ket{0\cdots0}_{\sfY}$, and measures in computational basis. If the output of this measurement is $a$, it outputs $1$; otherwise, it outputs $0$.

Consider the following hybrids.\\

$\hyb_1$ In this hybrid, $\cB$ is given $a$, $\enc(\pk,a;r),\pk, \cO\leftarrow \lobf(\newbfckt[\qfhe.\dec(\sk,\cdot),b,(\sk|r)])$. \\

$\hyb_2$: In this hybrid, $\cB$ is given $a$, $\enc(\pk,a;r),\pk$ and $\cO \leftarrow \simulator(1^\secparam)$.\\

Since the lock $b$ is chosen uniformly at random, by security of lockable obfuscation, the probability that $\cB$ outputs $1$ in the first hybrid is negligibly close to the probability that $\cB$ outputs $1$ in the second hybrid. This means that if $\tr[\Pi_a \ket{\psi_{1,\cO}}\bra{\psi_{1,\cO}}] \geq \nonnegl(\secparam)$ with non-negligible probability when $\cO \leftarrow \lobf(\newbfckt[\qfhe.\dec(\sk,\cdot),b,(\sk|r)])$, then this still holds when $\cO \leftarrow \simulator(1^\secparam)$. \par But we show that if $\tr[\Pi_a \ket{\psi_{1,\cO}}\bra{\psi_{1,\cO}}] \geq \nonnegl(\secparam)$, when $\cO$ is generated as $\cO \leftarrow \simulator(1^\secparam)$, then QFHE is insecure. 
\begin{itemize}
    \item Consider the following QFHE adversary who is given $\ket{\xi}$ as auxiliary information, and chooses two messages $m_0=0\cdots0$ and $m_1=a$, where $a$ is sampled uniformly at random from $\{0,1\}^{\secparam}$. It sends $(m_0,m_1)$ to the challenger. 
    \item The challenger of QFHE then generates $\ct_d = \enc(\pk,m_d)$, for some bit $d \in \{0,1\}$ and sends it to the QFHE adversary.
    \item The QFHE adversary computes $\cO \leftarrow \simulator(1^\secparam)$.
    \item The QFHE adversary then prepares the state $\ket{\psi_d} = U_1 \left( \ket{\xi}\ket{\ct_d, \cO, \pk}\ket{0\cdots0}_{\sfX}\ket{0\cdots0}_{\sfY}\right)$ and measures register $\sfX$ in the computational basis.
\end{itemize}

\noindent If $d=0$, the probability that the QFHE adversary obtains $a$ as outcome is negligible; since $a$ is independent of $U_1$, $\pk$, $\ket{\xi}$, and $\cO$. But from our hypothesis ($\prob[\tr[\Pi_{a}\ket{\psi^1}\bra{\psi^1}]\geq \nonnegl(\secparam)]\geq \nonnegl(\secparam)$), the probability that the QFHE adversary obtains $a$ as outcome is non-negligible for the case when $d=1$. This contradicts the security of QFHE as the adversary can use $a$ to distinguish between these two cases. 

\noindent To prove the induction hypothesis, suppose that for all $t<T$, the following two conditions hold:
\begin{enumerate}
    \item $\tr[\Pi_{a} \ket{\psi^t}\bra{\psi^t}] \leq \negl(\secparam)$
    \item $|\braket{\phi^t}{\psi^t}|= 1-\delta_t$ 
\end{enumerate}
for some negligible $\delta_1,...,\delta_{T-1}$. We can write
$$|\braket{\phi^T}{\psi^T}| = |\bra{\phi^{T-1}}O\ket{\psi^{T-1}}|$$

By hypothesis (2) above, we have $\ket{\phi^{T-1}}=(1-\delta_{T-1}) e^{i \alpha}\ket{\psi^{T-1}}+\sqrt{2\delta_{T-1}-\delta_{T-1}^2}\ket{\widetilde{\psi}^{T-1}}$, here $\alpha$ is some phase, and $\ket{\widetilde{\psi}^{T-1}}$ is some state orthogonal to $\ket{\psi^{T-1}}$. Then
\begin{align*}
|\braket{\phi^T}{\psi^T}| &= |(1-\delta_{T-1}) e^{i \alpha}\bra{\psi^{T-1}}O\ket{\psi^{T-1}} + \sqrt{2\delta_{T-1}-\delta_{T-1}^2} \bra{\widetilde{\psi}^{T-1}}O\ket{\psi^{T-1}}| \\
&\geq |(1-\delta_{T-1})e^{i\alpha} \bra{\psi^{T-1}}O\ket{\psi^{T-1}}| - \sqrt{2\delta_{T-1}-\delta_{T-1}^2}\\
&\geq (1-\delta_{T-1}) |\bra{\psi^{T-1}}O\ket{\psi^{T-1}}| - \sqrt{2\delta_{T-1}-\delta_{T-1}^2}
\end{align*}

\noindent By hypothesis (1) above, and since the oracle acts non-trivially only on $a$, we have 
$|\bra{\psi^{T-1}}O\ket{\psi^{T-1}}| \geq 1-\negl(\secparam)$, which gives us
$$|\braket{\phi^T}{\psi^T}| \geq 1-\negl(\secparam). $$

Now we want to show that $\tr[\Pi_{a} \ket{\psi^{T}}\bra{\psi^{T}}] \leq \negl(\secparam)$. This follows from the security of lockable obfuscation and QFHE similarly to $T=1$ case.  Since $|\braket{\phi^T}{\psi^T}|\geq 1-\negl(\secparam)$, we have that 
$$\tr[\Pi_{a} \ket{\phi^{T}}\bra{\phi^{T}}] \leq \negl(\secparam) \implies \tr[\Pi_{a} \ket{\psi^{T}}\bra{\psi^{T}}] \leq \negl(\secparam).$$ 

\noindent From a similar argument to the $T=1$ case but using $U_T U_{T-1} \cdots U_1$ instead of just $U_1$, we have that  $\prob[\tr[\Pi_{a} \ket{\phi^{T}}\bra{\phi^{T}}]\leq \negl(\secparam)]\geq 1-\negl(\secparam)$.
\end{proof}

\noindent Let $E_i$ denote the event that $\tr[\Pi_{a} \ket{\psi^i}\bra{\psi^i}] \leq \negl(\secparam)$. Let $p_T$ be the probability that $\tr[\Pi_{a} \ket{\psi^t}\bra{\psi^t}] \leq \negl(\secparam)$ for all the queries $t\leq T$. Using the previous claim, we have that

\begin{align*}
    p_T &= \overset{T}{\underset{t=1}{\prod}}\prob[E_t|\forall j<t, E_j] \\
    &\geq (1-\negl(\secparam))^T \\
    &\geq (1-T \cdot  \negl(\secparam))
\end{align*}

\end{proof}

Suppose that there is a QPT $\cB$ that can learn $\cktclass$ with respect to $\distrc$ with non-negligible probability $\delta$. In other words, for all inputs $x$, 

$$\prob \left[ U(\rho,x)=C_{a,b,r,\pk,\cO}(x) : \substack{C_{a,b,r,\pk,\cO}\leftarrow \distrc \\ (U,\rho) \leftarrow \cB^{C_{a,b,r,\pk,\cO}}(1^\secparam)} \right] = \delta $$

\noindent We use $\cB^{C_{a,b,r,\pk,\cO}}$ to construct a QPT $\cA^{C_{a,b,r,\pk,\cO}}$ that can find $\sk$ with probability neglibly close to $\delta$, contradicting Claim~\ref{clm:unlearn}.  To do this, $\cA$ first prepares $(U,\rho) \leftarrow \cB^{C_{a,b,r,\pk,\cO}}(1^\secparam)$. Then, $\cA^{C_{a,b,r,\pk,\cO}}$ queries the oracle on input $0\cdots0$, obtaining $\ct_1 = \qfhe.\enc(\pk,a;r)$ along with $\pk$ and $\cO=\lobf(\newbfckt[\qfhe.\dec(\sk,\cdot),b,(\sk|r)])$. Finally, it homomorphically computes $\ct_2 \leftarrow \qfhe.\eval(U(\rho,\cdot), \ct_1)$. Then it computes $\sk'|r' = \cO(\ct_2)$, and outputs $\sk'$.
\par By the correctness of the $\qfhe$ and because $U(\rho,a)=b$ holds with probability $\delta$, we have that  $\qfhe.\dec_{\sk}(\ct_2)=b$ with probability negligibly close to $\delta$. By correctness of lockable obfuscation $\cO(\ct_2)$ will output the right message $\sk$. This means that output of $\cA$ is $\sk$ with probability negligibly close to $\delta$.
\section{Proof of Lemma~\ref{lem:mainconstruction}}
\label{sec:mainproof}
\noindent To prove the lemma, it helps to formulate a probabilistic process. For any QPT adversary $\cA$, define the following event. \\

\noindent \underline{$\realevent(1^{\secparam})$}: 
\begin{itemize}
    \item $\crs \leftarrow \setup\left( 1^{\secparam} \right)$, 
    \item $\key \leftarrow \gen(\crs)$,
    \item $C \leftarrow \distrc(\secparam)$,
    \item $\left(\rho_C=\left( \ket{A}\bra{A},\widetilde{g},\widetilde{g_{\bot}},\widetilde{C},\pi \right)\right)  \leftarrow \lessor\left(\key,C \right)$ 
    \item $\rho^*=\left( \widetilde{C}^{(1)},\widetilde{g}^{(1)},\widetilde{g_{\bot}}^{(1)},\pi^{(1)},\widetilde{C}^{(2)},\widetilde{g}^{(2)},\widetilde{g_{\bot}}^{(2)},\pi^{(2)},\sigma^*  \right) \leftarrow \adversary\left(\crs, \rho_C \right)$\\
    {\em That is, $\adversary$ outputs two copies; the classical part in the first copy is $\left( \widetilde{C}^{(1)},\widetilde{g}^{(1)},\widetilde{g_{\bot}}^{(1)},\pi^{(1)} \right)$ and the classical part in the second copy is $\left( \widetilde{C}^{(2)},\widetilde{g}^{(2)},\widetilde{g_{\bot}}^{(2)},\pi^{(2)} \right)$. Moreover, it outputs a single density matrix $\sigma^*$ associated with two registers $\Reg_1$ and $\Reg_2$; the state in $\Reg_1$ is associated with the first copy and the state in $\Reg_2$ is associated with the second.   }
    \item $ \sigma^*_1 = \tr_2[\sigma^*]$
    \item Set $\rho_C^{(1)} = \left(\sigma^*_1, \widetilde{C}^{(1)},\widetilde{g}^{(1)},\widetilde{g_{\bot}}^{(1)},\pi^{(1)} \right)$ and $\rho_C^{(2)} = \left( \Pi_2( \sigma^*),\widetilde{C}^{(2)},\widetilde{g}^{(2)},\widetilde{g_{\bot}}^{(2)},\pi^{(2)} \right)$ where:
    $$\Pi_2(\sigma^*) = \frac{\tr_1 \left[(\Pi_{(\widetilde{g}^{(1)},{\widetilde{g_\perp}^{(1)}})}\otimes I) \sigma^*\right]}{\tr \left[(\Pi_{(\widetilde{g}^{(1)},{\widetilde{g_\perp}^{(1)}})}\otimes I) \sigma^*\right]} $$
    and where $\Pi_{(\widetilde{g}^{(1)},{\widetilde{g_\perp}^{(1)}})}$ is the projection onto the subspace obfuscated by $(\widetilde{g}^{(1)}, \widetilde{g_\perp}^{(1)})$. In other words, $\Pi_2(\sigma^*)$ is the quantum state on register 2 conditioned on $\run$ not outputting $\bot$ when applied to register $1$.
\end{itemize}

\noindent Note that to prove this lemma, it suffices to prove the following: 

 $$\underset{}{\prob}\left[ \forall x, \left( \substack{ \prob\left[ (\run(\crs,x,\sigma_1^*) = C(x)\right] \geq \beta \\ \ \\ \bigwedge \\ \ \\  \forall x', \prob \left[\run(\crs, x', \cE_{x}^{(2)}(\sigma^*)) = C(x')  \right] \geq \beta } \right)\ :\ \realevent\left( 1^{\secparam} \right) \right] \leq \gamma.$$
 
 \noindent Note that for all $x$, $\cE_x^{(2)}(\sigma^*) = \Pi_2(\sigma^*)$, since the only quantum operation that $\run$ performs is projecting the first register of $\sigma^*$ onto the subspace corresponding to $\widetilde{g}^{(1)}$.
\noindent Consider the following: 
\begin{itemize}
    \item In the first case, we have the classical parts of the two pirated copies to be the same as the classical part of the original copy. Define $\gamma_1$ as follows: 
    $$\prob \left[ \substack{ \forall x, \prob\left[ (\run(\crs,x,\sigma_1^*) = C(x)\right] \geq \beta \\ \ \\ \bigwedge \\ \ \\  \forall x', \prob \left[\run(\crs, x', \Pi_2(\sigma^*)) = C(x')  \right] \geq \beta \\ \ \\ \bigwedge \\ \ \\  \left(\widetilde{C},\widetilde{g},\widetilde{g
_{\bot}}\right) =  \left(\widetilde{C}^{(1)},\widetilde{g}^{(1)},\widetilde{g_{\bot}}^{(1)} \right) \\ \ \\ \bigwedge \\ \ \\ \left(\widetilde{C},\widetilde{g},\widetilde{g
_{\bot}}\right) =  \left(\widetilde{C}^{(2)},\widetilde{g}^{(2)},\widetilde{g_{\bot}}^{(2)} \right) }\  :\ \realevent\left( 1^{\secparam} \right) \right] = \gamma_1$$
    \item The other possible case is the case where at least one of the copies $(\widetilde{C}^{(1)},\widetilde{g}^{(1)},\widetilde{g_{\bot}}^{(1)})$ or $(\widetilde{C}^{(2)},\widetilde{g}^{(2)},\widetilde{g_{\bot}}^{(2)} )$ is not equal to the corresponding resgisters of the original copy. Define $\gamma_2$ as follows: 
     $$\prob \left[ \substack{ \forall x, \prob\left[ (\run(\crs,x,\sigma_1^*) = C(x)\right] \geq \beta \\ \ \\ \bigwedge \\ \ \\  \forall x', \prob \left[\run(\crs, x', \Pi_2(\sigma^*)) = C(x')  \right] \geq \beta \\ \ \\ \bigwedge \\ \ \\ \exists i,\  \left(\widetilde{C},\widetilde{g},\widetilde{g
_{\bot}}\right) \neq  \left(\widetilde{C}^{(i)},\widetilde{g}^{(i)},\widetilde{g_{\bot}}^{(i)}\right) }\ :\ \realevent\left( 1^{\secparam} \right) \right] = \gamma_2$$
 Henceforth, without loss of generality, we will assume that the second copy is not the same.
\end{itemize}
\noindent Note that $\gamma=\gamma_1+\gamma_2$. In the next two propositions, we prove that both $\gamma_1$ and $\gamma_2$ are negligible which will complete the proof of the lemma.   

\begin{proposition}
$\gamma_1\leq \negl(\secparam)$
\end{proposition}
\begin{proof}

The run algorithm first projects $\sigma^*$ into $\ket{A}^{\otimes 2}$, and outputs $\bot$ if $\sigma^*$ is not $(\ket{A}\bra{A})^{\otimes 2}$.
Suppose that $\bra{A} \sigma_1^* \ket{A}$ is negligible, then $\run$ will output $\bot$ on the first register with probability negligibly close to $1$, and we would have $\gamma_1$ negligible as desired. 

On the contrary, suppose that $\bra{A} \sigma_1^* \ket{A}$ is non-negligible, and we have that $$\Pi_2(\sigma^*)=\frac{\tr_1\left[(\ket{A}\bra{A}\otimes I)\sigma^*\right]}{\tr\left[(\ket{A}\bra{A}\otimes I)\sigma^*\right]}$$ 
i.e. the state in the second register after $\run$ succesfully projects $\sigma_1^*$ onto $\ket{A}\bra{A}$.

We will prove the following claim, which implies that at least one of the two copies will output $\bot$ under $\run$ with probability neglibly close to $1$.
\begin{claim}\label{clm:notsame}
$\bra{A} \Pi_2(\sigma^*) \ket{A} \leq \negl(\secparam)$
\end{claim}
\begin{proof}
Suppose not. We then show how to use $\cA$ to violate the quantum no-cloning theorem.  Specifically, we use the claim by Zhandry \cite{Zha19} who showed that there does not exist any QPT algorithm that on input $(\ket{A},\widetilde{g}:=\shO(A),\widetilde{g_\perp}:=\shO(A^\perp))$ can prepare the state $\ket{A}^{\otimes 2}$ with non-negligible probability. But in order to show this, we first need to remove information about $A$ being used elsewhere in the state; in particular, we need to remove the usage of $A$ in the generation of the seNIZK proof. \\

\noindent Consider the following adversary $\cB'$. It runs $\cA$ and then projects the output of $\cA$ onto $\left( \ket{A}\bra{A} \right)^{\otimes 2}$; the output of the projection is the output of $\cB'$. 

\paragraph{$\cB'(C)$:}
\begin{enumerate}
    \item Compute $\crs,\key$ as in the construction.
    \item Compute $\rho_C \leftarrow \lessor(\key,C)$. Let $\rho_C=\left( \ket{A}\bra{A},\widetilde{g},\widetilde{g_{\bot}},\widetilde{C},\pi \right)$. 
    \item Compute $\cA(\crs,\rho_C)$ to obtain $\left( \widetilde{C}^{(1)},\widetilde{g}^{(1)},\widetilde{g_{\bot}}^{(1)},\pi^{(1)},\widetilde{C}^{(2)},\widetilde{g}^{(2)},\widetilde{g_{\bot}}^{(2)},\pi^{(2)},\sigma^*\right)$.
    \item Then, project $\sigma^*$ onto $(\ket{A}\bra{A})^{\otimes 2}$ by using $\widetilde{g}$ and $\widetilde{g_\perp}$. Let $m$ be the outcome of this projection, so $m=1$ means that the post measured state is $(\ket{A}\bra{A})^{\otimes 2}$. 
    \item Output $m$.
\end{enumerate}

The projection $(\ket{A}\bra{A})^{\otimes 2}$ can be done by first projecting the first register onto $\ket{A}\bra{A}$ and then the second register. Conditioned on the first register not outputting $\bot$, means that $\sigma^*_1$ is succesfully projected onto $\ket{A}\bra{A}$. By our assumption that $\bra{A} \sigma^*_1 \ket{A}$ is non-negligible, this will happen with non-negligible probability.  Conditioned on this being the case, if $\bra{A} \Pi_2( \sigma^*) \ket{A}$ is non-negligible, then projecting the second register onto $\ket{A}\bra{A}$ will also succeed with non-negligible probability.  This means that $m=1$ with non-negligible probability.\\

\noindent Consider the following adversary $\cB$; later we will use $\cB$ to violate the no-cloning theorem. It follows the same steps as $\cB'$ except in preparing the states $\ket{A}$ and computing obfuscated circuits $\widetilde{g}$, $\widetilde{g_{\bot}}$; it gets these quantities as input. Moreover, it simulates the proof $\pi$ instead of computing the proof using the honest prover. This is because unlike $\cB'$, the adversary $\cB$ does not have the randomness used in computing $\widetilde{g}$ and $\widetilde{g_{\bot}}$ and hence cannot compute the proof $\pi$ honestly. 

\paragraph{$\cB(\ket{A},\widetilde{g},\widetilde{g_\perp})$:}
\begin{enumerate}
    \item Sample randomness $r_o$ and compute $\tilde{C} \leftarrow \qiho.\obf(C;r_o)$.
    \item Let $\fksetup$ and $\simr$ be associated with the simulation-extractability propety of $\ssnizk$. Compute $(\widetilde{\crs},\td) \leftarrow \fksetup(1^{\secparam})$. 
    \item Compute $(\pi,\st) \leftarrow \simr \left( \widetilde{\crs},\td,\left(\widetilde{g},\widetilde{g_\perp},\widetilde{C} \right) \right) $
    \item Let $\rho_C= (\ket{A}\bra{A},\widetilde{g}, \widetilde{g_\perp}, \widetilde{C},\pi)$
    \item Run $\cA(\widetilde{\crs},\rho_C)$ to obtain $\left(\widetilde{C}^{(1)},\widetilde{g}^{(1)},\widetilde{g_{\bot}}^{(1)},\pi^{(1)},\widetilde{C}^{(2)},\widetilde{g}^{(2)},\widetilde{g_{\bot}}^{(2)},\pi^{(2)},\widetilde{\sigma^*}\right)$.
    \item Then, project $\widetilde{\sigma^*}$ onto $(\ket{A}\bra{A})^{\otimes 2}$ by using $\widetilde{g}$ and $\widetilde{g_\perp}$. Let $m$ be the outcome of this projection, so $m=1$ means that the post measured state is $(\ket{A}\bra{A})^{\otimes 2}$.
    \item Output $m$.
\end{enumerate}

\noindent Note that from the q-simulation-extractability property\footnote{We don't need the full-fledged capability of q-simulation-extractability to argue this part; we only need q-zero-knowledge property which is implied by q-simulation-extractability.} of $\ssnizk$, it follows that the probability that $\cB$ outputs 1 is negligibly close to the probability that $\cB'$ outputs $1$ because everything else is sampled from the same distribution.  This implies that $\cB$ on input $(\ket{A},\widetilde{g},\widetilde{g_\perp})$ outputs $\ket{A}^{\otimes 2}$ with non-negligible probability, contradicting the claim by Zhandry~\cite{Zha19}. This completes the proof of the claim. \submversion{\qed}

\end{proof}

\noindent At this point, we want to show that if $\left( \widetilde{g}^{(2)},\widetilde{g}^{(2)}_\perp \right) = \left( \widetilde{g},\widetilde{g_\perp} \right)$, and $\bra{A}\Pi_2(\sigma^*)\ket{A} \leq \negl(\secparam)$, then the probability that $\run(\crs, \Pi_2(\sigma^*),x)$ evaluates $C$ correctly is negligible.

By correctness of $\shO$, we have $$\prob[\forall x \text{ } \widetilde{g}^{(2)}(x) = \mathbb{1}_A(x)] \geq 1-\negl(\secparam)$$ $$\prob[\forall x \text{ } \widetilde{g}_\perp^{(2)}(x) = \mathbb{1}_{A^\perp}(x)] \geq 1-\negl(\secparam)$$ 

This means that with probability negligibly close to $1$, the first thing that the $\run$ algorithm does on input $\rho_C^{(2)}=(\Pi_2(\sigma^*), \widetilde{g}^{(2)}, \widetilde{g}^{(2)}_\perp, \widetilde{C},\pi)$ is to measure $\{\ket{A}\bra{A},I-\ket{A}\bra{A}\}$ on $\Pi_2( \sigma^*)$. If $I-\ket{A}\bra{A}$ is obtained, then the $\run$ algorithm will output $\bot$. By Claim~\ref{clm:notsame}, the probability that this happens is neglibly close to 1. Formally, when $\widetilde{g}$ and $\widetilde{g_\perp}$ are subspace obfuscations of $A$ and $A^\perp$ respectively, the check $a=1$ and $b=1$ performed by the $\run$ algorithm is a projection onto $\ket{A}\bra{A}$. 

\begin{align*}
\prob[a=1,b=1] &= \tr[\ft^\dagger \Pi_{A^\perp} \ft \Pi_A \Pi_2( \sigma^*)] \\
&= \tr[\ket{A}\bra{A} \Pi_2( \sigma^*)] \\
&= \bra{A}\Pi_2( \sigma^*)\ket{A}\\
&\leq \negl(\secparam)
\end{align*}
where $\Pi_A = \underset{a \in A}{\sum} \ket{a}\bra{a}$ and $\Pi_{A^\perp} = \underset{a \in A^\perp}{\sum} \ket{a}\bra{a}$.
From this, we have that $\prob[\run(\crs,\rho_C^{(2)},x)=\bot] \geq 1-\negl(\secparam)$, and we have
$\prob[\run(\crs,\rho_C^{(2)},x)=C(x)] \leq \negl(\secparam)$ with probability neglibly close to $1$.

This finishes our proof that if $\beta$ is non-negligible, then $\gamma_1\leq \negl(\secparam)$.
\end{proof}
\begin{proposition}
$\gamma_2\leq\negl(\secparam)$.  
\end{proposition}
\begin{proof}
We consider the following hybrid process. \\

\noindent \underline{$\hybprocess_1(1^{\secparam})$}: 
\begin{itemize}
     \item $\left( \widetilde{\crs},\td \right) \leftarrow \fksetup\left( 1^{\secparam} \right)$,
    \item $\key \leftarrow \gen(\crs)$,
    \item $C \leftarrow \distrc(\secparam)$,
    \item Sample a random $\frac{\secparam}{2}$-dimensionall sub-space $A \subset \Zq^{\secparam}$. Prepare the state $\ket{A}=\frac{1}{\sqrt{q^{\secparam/2}}} \sum_{a \in A} \ket{a}$.
    \item Compute $\widetilde{g} \leftarrow \shO \left(A;r_A\right)$, 
    \item Compute $\widetilde{g_{\bot}} \leftarrow \shO \left(A^{\perp};r_{A^{\perp}}\right)$,
    \item Compute $\widetilde{C}\leftarrow \iho.\obf(C;r_o)$
    \item $(\pi,\st) \leftarrow \simr_1 \left(\crs, \td,\left(\widetilde{g}, \widetilde{g_{\bot}},\widetilde{C} \right) \right) $
    \item Set $\rho_C=\left( \ket{A}\bra{A},\widetilde{g},\widetilde{g_{\bot}},\widetilde{C},\pi \right)$. 
    \item $\left( \widetilde{C}^{(1)},\widetilde{g}^{(1)},\widetilde{g_{\bot}}^{(1)},\pi^{(1)},\widetilde{C}^{(2)},\widetilde{g}^{(2)},\widetilde{g_{\bot}}^{(2)},\pi^{(2)},\sigma^*  \right) \leftarrow \adversary\left(\crs, \rho_C \right)$
    \item Set $ \sigma^*_1 = \tr_2[\sigma^*]$
    \item Set $\rho_C^{(1)} = \left( \sigma^*_1,\widetilde{C}^{(1)},\widetilde{g}^{(1)},\widetilde{g_{\bot}}^{(1)},\pi^{(1)} \right)$ and $ \rho_C^{(2)} = \left( \Pi_2( \sigma^*),\widetilde{C}^{(2)},\widetilde{g}^{(2)},\widetilde{g_{\bot}}^{(2)},\pi^{(2)} \right)$
    \item $\left(A^*,r_o^*,r_A^*,r_{A^{\perp}}^*,C^*,x^* \right) \leftarrow \simr_2\left(\st,\left(\widetilde{g}^{(2)},\widetilde{g_{\bot}}^{(2)},\widetilde{C}^{(2)} \right),\ \pi^{(2)} \right)$. 
\end{itemize}

\noindent The proof of the following claim follows from the q-simulation-extractactability property of $\ssnizk$. 

\begin{claim}
\label{clm:senizk}

Assuming that $\ssnizk$ satisfies q-simulation extractability property secure against QPT adversaries running in time $2^{n}$, we have: 
$$\prob \left[ \substack{\left(\left( \widetilde{g}^{(2)},\widetilde{g_{\bot}}^{(2)},\widetilde{C}^{(2)} \right),\ \left(A^*,r_o^*,r_A^*,r_{A^{\perp}}^*,C^*,x^* \right)\right) \in \rel(L)  \\ \ \\ \bigwedge\\ \ \\  \forall x',\  \prob\left[\run\left(\crs,\rho^{(2)},x' \right) = C(x')  \right]  \geq \beta \\ \ \\ \bigwedge \\ \ \\ \exists i,\ \left(\widetilde{C},\widetilde{g},\widetilde{g_{\bot}} \right)  \neq \left( \widetilde{C}^{(i)},\widetilde{g}^{(i)},\widetilde{g_{\bot}}^{(i)} \right)}\ :\ \hybprocess_1 \left(1^{\secparam} \right) \right] = \delta_1$$ 
Then, $|\delta_1-\gamma_2| \leq \negl(\secparam)$.
\end{claim}
\begin{remark}
Note that $n$ is smaller than the length of the NP instance and thus, we can invoke the sub-exponential security of the seNIZK system guaranteed in Lemma~\ref{lem:qsenizks}. 
\end{remark}
\begin{proof}[of Claim~\ref{clm:senizk}]
Consider the following $\ssnizk$ adversary $\reduction$: 
\begin{itemize}
    \item It gets as input $\crs$.
    \item It samples and computes $(C,A,\widetilde{g},\widetilde{g_{\bot}},\widetilde{C})$ as described in $\hybprocess_1(1^{\secparam})$. It sends the following instance-witness pair to the challenger of seNIZK: $$\left(\left(\widetilde{g},\widetilde{g_{\bot}},\widetilde{C}\right),\ ((A,r_o,r_A,r_{A^{\bot}},C,x) \right),$$ 
    where $r_o,r_A,r_{A^{\perp}}$ are, respectively, the randomness used to compute the obfuscated circuits  $\widetilde{g}$, $\widetilde{g_{\perp}}$ and $ \widetilde{C}$.   
    \item The challenger returns back $\pi$. 
    \item $\reduction$ then sends $\left(\ket{A},\widetilde{g},\widetilde{g_{\perp}},\widetilde{C},\pi \right)$ to $\cA$. 
    \item $\cA$ then outputs $\left( \widetilde{C}^{(1)},\widetilde{g}^{(1)},\widetilde{g_{\bot}}^{(1)},\pi^{(1)},\widetilde{C}^{(2)},\widetilde{g}^{(2)},\widetilde{g_{\bot}}^{(2)},\pi^{(2)},\sigma^*  \right)$. 
    \item Finally, $\reduction$ performs the following checks: 
    \begin{itemize}
        \item {\em Verify if the classical parts are different}: Check if $\left( \widetilde{C},\widetilde{g},\widetilde{g_{\bot}} \right) = \left( \widetilde{C}^{(2)},\widetilde{g}^{(2)},\widetilde{g_{\bot}}^{(2)} \right)$; if so output $\bot$, otherwise continue. 
        \item {\em Verify if second copy computes $C$}: If the measurement above does not output $\bot$, set $ \rho_C^{(2)} = \left( \Pi_2(\sigma^*),\widetilde{C}^{(2)},\widetilde{g}^{(2)},\widetilde{g_{\bot}}^{(2)},\pi^{(2)}\right)$. For every $x$, check if $\widetilde{C}^{(2)}(x)=C(x)$. If for any $x$, the check fails, output $\bot$. {\em // Note that this step takes time $2^{n} \cdot \poly(\secparam)$.}

    \end{itemize}
    \item Output $\left(\left(\widetilde{g}^{(2)},\widetilde{g_{\bot}}^{(2)},\widetilde{C}^{(2)} \right),\pi^{(2)} \right)$.  
\end{itemize}
\noindent Note that $\reduction$ is a valid $\ssnizk$ adversary: it produces a proof on an instance different from the instance along with a proof (either real or simulated) it received and moreover, the proof produced by $\reduction$ (conditioned on not $\bot$) is an accepting proof.
\par If $\reduction$ gets as input honest CRS and honestly generated proof $\pi$ then this corresponds to $\process_1(1^{\secparam})$ and if $\reduction$ gets as input simulated CRS and simulated proof $\pi$ then this corresponds to $\hybprocess_1(1^{\secparam})$.

Thus, from the security of q-simulation-extractable NIZKs, we have that $|\gamma_2 - \delta_1| \leq \negl(\secparam)$. 
\end{proof}

\noindent We first prove the following claim. 

\begin{claim}
Suppose the following two conditions hold: 
\begin{itemize}
    \item $\left(\left( \widetilde{g}^{(2)},\widetilde{g_{\bot}}^{(2)},\widetilde{C}^{(2)} \right),\ \left(A^*,r_o^*,r_A^*,r_{A^{\perp}}^*,C^*,x^* \right)\right) \in \rel(L)$ 
    \item $\forall x',\  \prob\left[\run\left(\crs,\rho_C^{(2)},x' \right) = C(x')  \right]  \geq \beta$
\end{itemize}
Then it holds that $C(x^*)=1$. 
\end{claim}
\begin{proof}
We first claim that $\forall x,\  \prob\left[\run\left(\crs,\rho_C^{(2)},x \right) = C(x)  \right]  \geq \beta$ implies that $\widetilde{C}^{(2)} \equiv C$, where $\equiv$ denotes functional equivalence. Suppose not. Let $x'$ be an input such that $\widetilde{C}^{(2)}(x') \neq C(x')$ then this means that $\run(\crs,\rho_C^{(2)},x')$ {\em always} outputs a value different from  $C(x')$; this follows from the description of $\run$. This further means that $\prob[ \run\left(\crs,\rho_C^{(2)},x' \right) = C(x')]=0$, contradicting the hypothesis. \par Moreover, $\left(\left( \widetilde{g}^{(2)},\widetilde{g_{\bot}}^{(2)},\widetilde{C}^{(2)} \right),\ \left(A^*,r_o^*,r_A^*,r_{A^{\perp}}^*,C^*,x^* \right)\right) \in \rel(L)$ implies that $\widetilde{C}^{(2)}=\qiho(1^{\secparam},C^*;r_{o}
^*)$ and $C^*(x^*)=1$. Furthermore, perfect correctness of $\qiho$ implies that $\widetilde{C}^{(2)} \equiv C^*$.  
\par So far we have concluded that $\widetilde{C}^{(2)} \equiv C$,  $\widetilde{C}^{(2)} \equiv C^*$ and $C^*(x^*)=1$. Combining all of them together, we have $C(x^*)=1$. \submversion{\qed}   

\end{proof}

\noindent We claim that $\delta_1$ is upper bounded by the probability that $x^*$ is an accepting input of $C$. Consider the following inequalities. 

\begin{eqnarray*}
\delta_1 & = & \prob \left[ \substack{\left(\left( \widetilde{g}^{(2)},\widetilde{g_{\bot}}^{(2)},\widetilde{C}^{(2)} \right),\ \left(A^*,r_o^*,r_A^*,r_{A^{\perp}}^*,C^*,x^* \right)\right) \in \rel(L)\\ \ \\ \bigwedge\\ \ \\ \forall x',\  \prob\left[\run\left(\crs,\rho_C^{(2)},x' \right) = C(x')  \right]  \geq \beta \\ \ \\ \bigwedge \\ \ \\ \exists i,\ \left(\widetilde{C},\widetilde{g},\widetilde{g_{\bot}} \right)  \neq \left( \widetilde{C}^{(i)},\widetilde{g}^{(i)},\widetilde{g_{\bot}}^{(i)} \right)}\ :\ \hybprocess_1 \right]\\ 
& & \\
& & \\
& = & \prob \left[ \substack{C(x^*)=1\\ \ \\ \bigwedge\\ \ \\  \exists i,\ \left(\widetilde{C},\widetilde{g},\widetilde{g_{\bot}} \right)  \neq \left( \widetilde{C}^{(i)},\widetilde{g}^{(i)},\widetilde{g_{\bot}}^{(i)} \right)}\ :\ \hybprocess_1 \right] \\ 
& & \\
& \leq & \prob \left[ C\left( x^* \right)=1\ :\ \hybprocess_1 \right]
\end{eqnarray*}
\noindent Let $\prob \left[ C\left( x^* \right)=1\ :\ \hybprocess_1 \right]=\delta_2$.

\begin{claim}
Assuming the q-input-hiding property of $\iho$, we have $\delta_2 \leq \negl(\secparam)$ 
\end{claim}
\begin{proof}
Suppose $\delta_2$ is not negligible. Then we construct a QPT adversary $\reduction$ that violates the q-input-hiding property of $\iho$, thus arriving at a contradiction. 
\par $\reduction$ now takes as input $\widetilde{C}$ (an input-hiding obfuscator of $C$), computes $\left( \widetilde{\crs},\td \right) \leftarrow \fksetup \left( 1^{\secparam} \right)$ and then computes $\rho_C=\left( \ket{A}\bra{A},\widetilde{g},\widetilde{g_{\bot}},\widetilde{C},\pi \right)$ as computed in $\hybprocess_1$. It sends $\left(\widetilde{\crs},\rho_C \right)$ to $\adversary$ who responds with $\left( \widetilde{C}^{(1)},\widetilde{g}^{(1)},\widetilde{g_{\bot}}^{(1)},\pi^{(1)},\widetilde{C}^{(2)},\widetilde{g}^{(2)},\widetilde{g_{\bot}}^{(2)},\pi^{(2)},\sigma^*  \right)$. Compute $\left(A^*,r_o^*,r_A^*,r_{A^{\perp}}^*,C^*,x^* \right)$ by generating $ \simr_2(\st,(\widetilde{g}^{(2)},\widetilde{g_{\bot}}^{(2)},\allowbreak \widetilde{C}^{(2)} ),\allowbreak \pi^{(2)} )$, where $\st$ is as defined in $\hybprocess_1$. Output $x^*$.
\par Thus, $\reduction$ violates the q-input-hiding property of $\iho$ with probability $\delta_2$ and thus $\delta_2$ has to be negligible.
\end{proof}

\noindent Combining the above observations, we have that $\gamma_2 \leq \negl(\secparam)$ for some negligible function $\negl$. This completes the proof.

\end{proof}
}



\end{document}